%% file: tree-edge-coloring.tex
\documentclass[11pt]{article}
\usepackage{fullpage}
\usepackage[usenames]{xcolor}
\usepackage{xifthen}

\def\fontsettingup{2} 

\usepackage{amsmath, amsthm, amssymb}
\usepackage[T1]{fontenc}
\ifthenelse{\fontsettingup = 1}{ \usepackage{eulerpx, eucal, tgpagella}   }{}
\ifthenelse{\fontsettingup = 2}{ \usepackage{mathpazo, eufrak, tgpagella} }{}

\usepackage{hyperref}
\hypersetup{
  colorlinks=true,
  linkcolor=blue,
  filecolor=magenta,
  urlcolor=cyan,
}

\usepackage{algorithm2e}

\newboolean{arxiv}
\setboolean{arxiv}{true}

\usepackage{tikz}
\ifarxiv
\else
\usetikzlibrary{graphs,graphdrawing}
\usetikzlibrary{positioning}
\usetikzlibrary{backgrounds}
\usegdlibrary{trees}
\usegdlibrary{layered}
\fi

\usepackage{cleveref}

\usepackage[a4paper]{geometry}
 \newgeometry{
 textheight=9in,
 textwidth=6.5in,
 top=1in,
 headheight=14pt,
 headsep=25pt,
 footskip=30pt
 }

\newtheorem{theorem}{Theorem}
\newtheorem{observation}[theorem]{Observation}
\newtheorem{claim}[theorem]{Claim}
\newtheorem*{claim*}{Claim}

\newtheorem{lemma}[theorem]{Lemma}
\newtheorem{proposition}[theorem]{Proposition}
\newtheorem{corollary}[theorem]{Corollary}
\theoremstyle{definition}

\newtheorem{definition}[theorem]{Definition}
\newtheorem{remark}[theorem]{Remark}
\newtheorem*{remark*}{Remark}

\ifthenelse{\fontsettingup = 1}{
  \def\*#1{\mathbf{#1}} 
  \def\+#1{\mathcal{#1}} 
  \def\-#1{\mathrm{#1}} 
  \def\^#1{\mathbb{#1}} 
  \def\!#1{\mathfrak{#1}} 
}{}

\ifthenelse{\fontsettingup = 2}{
  \def\*#1{\boldsymbol{#1}} 
  \def\+#1{\mathcal{#1}} 
  \def\-#1{\mathrm{#1}} 
  \def\^#1{\mathbb{#1}} 
  \def\!#1{\mathfrak{#1}} 
}{}

\newcommand{\ap}[1][]{
  \ifthenelse{\isempty{#1}}{
    \mathrm{AP}_\sigma
  }{
    \mathrm{AP}_{#1}
  }
} 
\newcommand{\elist}[2]{\+L^{#1}_{#2}} 
\newcommand{\vlist}[2]{{#1}(\delta_{#2})} 

\def\oPr{\mathbf{Pr}}
\renewcommand{\Pr}[2][]{ \ifthenelse{\isempty{#1}}
  {\oPr\left[#2\right]}
  {\oPr_{#1}\left[#2\right]} } 

\def\oE{\mathbb{E}}
\newcommand{\E}[2][]{ \ifthenelse{\isempty{#1}}
  {\oE\left[#2\right]}
  {\oE_{#1}\left[#2\right]} }

\usepackage{xparse}
\def\oVar{\mathbf{Var}}
\NewDocumentCommand{\Var}{ O{} O{} m }{
  \ifthenelse{\isempty{#1}} {
    \ifthenelse{\isempty{#2}} {
      \oVar\left[#3\right]
    } {
      \oVar^{#2}\left[#3\right]
    }
  } {
    \ifthenelse{\isempty{#2}} {
      \oVar_{#1}\left[#3\right]
    } {
      \oVar_{#1}^{#2}\left[#3\right]
    }
  }
}

\def\oEnt{\mathbf{Ent}}
\NewDocumentCommand{\Ent}{ O{} O{} m }{
  \ifthenelse{\isempty{#1}} {
    \ifthenelse{\isempty{#2}} {
      \oEnt\left[#3\right]
    } {
      \oEnt^{#2}\left[#3\right]
    }
  } {
    \ifthenelse{\isempty{#2}} {
      \oEnt_{#1}\left[#3\right]
    } {
      \oEnt_{#1}^{#2}\left[#3\right]
    }
  }
}

\newcommand{\Trelax}{T_{\text{rel}}}
\newcommand{\Tmix}{T_{\text{mix}}}

\newcommand{\DTV}[2]{\-D_{\mathrm{TV}}\left({#1},{#2}\right)}

\newcommand{\e}{\mathrm{e}}
\renewcommand{\epsilon}{\varepsilon}
\renewcommand{\emptyset}{\varnothing}

\newcommand{\set}[1]{\left\{#1\right\}}
\newcommand{\tuple}[1]{\left(#1\right)}
\newcommand{\eps}{\varepsilon}

\newcommand{\tp}{\tuple}
\newcommand{\ol}{\overline}
\newcommand{\abs}[1]{\left\vert#1\right\vert}

\newcommand{\ftp}[1]{\left\lfloor#1\right\rfloor}

\newcommand{\cgap}{c_{\-{gap}}}
\newcommand{\csob}{c_{\-{sob}}}

\title{Optimal Mixing for Randomly Sampling Edge Colorings on Trees Down to the Max Degree}
\date{}

\newboolean{doubleblind}
\setboolean{doubleblind}{false}

\ifdoubleblind
\author{Author(s)}
\else

\author{
Charlie Carlson\thanks{Department of Computer Science, University of California, Santa Barbara. Email: \{charlieannecarlson,vigoda\}@ucsb.edu.  Research supported in part by NSF grant CCF-2147094.}  
\and 
Xiaoyu Chen\thanks{State Key Laboratory for Novel Software Technology, New Cornerstone Science Laboratory, Nanjing
University, China. E-mail: chenxiaoyu233@smail.nju.edu.cn}
\and
Weiming Feng\thanks{Institute for Theoretical Studies, ETH Z\"{u}rich, Switzerland. Email: weiming.feng@eth-its.ethz.ch. Research supported by Dr. Max R\"ossler, the Walter Haefner Foundation and the ETH Z\"urich Foundation.}
\and Eric Vigoda$^{*}$
}

\fi

\begin{document}

\maketitle

\begin{abstract}
We address the convergence rate of Markov chains for randomly generating an edge coloring of a given tree.  Our focus is on the Glauber dynamics which updates the color at a randomly chosen edge in each step.  For a tree $T$ with $n$ vertices and maximum degree $\Delta$, when the number of colors $q$ satisfies $q\geq\Delta+2$ then we prove that the Glauber dynamics has an optimal relaxation time of $O(n)$, where the relaxation time is the inverse of the spectral gap.  This is optimal in the range of $q$ in terms of $\Delta$ as Dyer, Goldberg, and Jerrum (2006) showed that the relaxation time is $\Omega(n^3)$ when $q=\Delta+1$.  For the case $q=\Delta+1$, we show that an alternative Markov chain which updates a pair of neighboring edges has relaxation time $O(n)$.  Moreover, for the $\Delta$-regular complete tree we prove $O(n\log^2{n})$ mixing time bounds for the respective Markov chain.
Our proofs establish approximate tensorization of variance via a novel inductive approach, where the base case is a tree of height $\ell=O(\Delta^2\log^2{\Delta})$, which we analyze using a canonical paths argument.

\end{abstract}

\thispagestyle{empty}

\newpage

\setcounter{page}{1}

\section{Introduction}

We study the Glauber dynamics for sampling $q$-edge colorings on trees of maximum degree $\Delta$.  
Given a set of colors $[q]:=\{1,2,\ldots,q\}$, a proper edge coloring $\sigma \in [q]^E$ assigns each edge a color such that adjacent edges receive distinct colors.  For a graph $G=(V,E)$, let $\Omega$ denote the set of proper $q$-edge colorings of $G$, and let $\mu$ denote the uniform distribution over $\Omega$.
We focus on a given graph $G$ and integer $q\geq 2$, generating a $q$-coloring almost uniformly at random.  More precisely, given $G$, $q$ and $\eps>0$, our goal is to generate a coloring $\sigma\in\Omega$ from a probability distribution $\nu$ which is within total variation distance $\leq\eps$ of the uniform distribution $\mu$, and to do so in time polynomial in $n=|V|$ and $\log(1/\eps)$.

The Glauber dynamics is a simple Markov chain which we study for this sampling problem.  The dynamics operates as follows.  From a coloring $X_t\in\Omega$, we choose an edge $e\in E$ and color $c\in [q]$ uniformly at random.  If no neighbor of $e$ has color $c$ in $X_t$ then we set $X_{t+1}(e)=c$ and otherwise we set $X_{t+1}(e)=X_t(e)$.  For all other edges $f\neq e$ we set $X_{t+1}(f)=X_t(f)$.  If the maximum degree of $G$ is $\Delta$ then when $q\geq 2\Delta+2$ then the Glauber dynamics is ergodic and the unique stationary distribution is the uniform distribution $\mu$.  

The {\em mixing time} $\Tmix$ is the number of steps $t$, from the worst initial state $X_0$, so that the distribution of $X_t$
is within total variation distance $\leq 1/4$ of $\mu$.  
The {\em relaxation time} $\Trelax$ is the inverse of the spectral gap.
More precisely, let $P$ denote the $N\times N$ transition matrix for the Glauber dynamics where $N=|\Omega|$, and let $\lambda_1=1>\lambda_2\geq\dots\lambda_N>-1$ denote the eigenvalues of $P$.  Then, $\Trelax:=(1-\lambda_*)^{-1}$ where $\lambda_*=\max\{\lambda_2,\abs{\lambda_N}\}$. We say the relaxation time is optimal if $\Trelax=O(n)$ since it is easy to establish the lower bound of $\Trelax=\Omega(n)$.   The optimal mixing time bound is $O(n\log{n})$~\cite{HayesSinclair}.

For graphs of maximum degree $\Delta$, there always exists a $q$-edge coloring when $q\geq 2\Delta+1$, and hence this presents a natural barrier at $q=2\Delta+1$ for general graphs.
For general graphs of maximum degree $\Delta$ when $\Delta=O(1)$ the following results were established.  Abdodolazimi, Liu, and Oveis Gharan~\cite{AbdolazimiLiuOveisGharan} proved $O(n\log{n})$ mixing when $q>(10/3)\Delta$; this improves upon the $q>(11/3-\eps_0)\Delta$ for a small, fixed $\eps_0>0$ which one obtains via vertex $q$-colorings~\cite{CDMPP19,Vigoda}.
A near-optimal bound was obtained by Wang, Zhang and Zhang~\cite{WZZ24} who proved $O(n\log{n})$ mixing time of Glauber dynamics when $q>(2+o(1))\Delta$ for any graph of maximum degree $\Delta$ for $\Delta=O(1)$. 

If we restrict attention to trees then the Glauber dynamics is ergodic for all $q\geq\Delta+1$, and hence the natural threshold is at $q=\Delta+1$ for trees.  Understanding the convergence rate of the Glauber dynamics on trees is interesting as it can provide intuition for its behavior on locally tree-like graphs such as sparse random graphs.  In fact, locally tree-like graphs serve as the crucial gadget for hardness of approximate counting results~\cite{sly2010computational,galanis2015inapproximability} by utilizing the stationary properties of the Gibbs distribution which are directly related to the mixing properties.

 Delcourt, Heinrich, and Perarnau~\cite{DHP} proved $\-{poly}(n)$ mixing time for the Glauber dynamics on trees with maximum degree $\Delta$ when $q\geq\Delta+1$.  Can we obtain optimal mixing for trees down to $q\geq\Delta+1$?  For a path with $q=3$ colors (which corresponds to $q=\Delta+1$ for $\Delta=2$), Dyer, Goldberg, and Jerrum~\cite{DGJ} proved that the Glauber dynamics has mixing time $\Theta(n^3\log{n})$ and hence relaxation time $\Trelax=\Omega(n^3)$.
 Therefore, optimal mixing and optimal relaxation time does not hold for the Glauber dynamics when $q=\Delta+1$ for all trees of maximum degree $\Delta$.

We prove that with one additional color optimal relaxation time of the Glauber dynamics holds for all trees of maximum degree $\Delta$.

\begin{theorem}
\label{thm:main-mixing}
  For any tree with maximum degree $\Delta$ and $n$ vertices, and any $q \geq \Delta + 2$, then the relaxation time of the Glauber dynamics for $q$-edge-colorings is $O_{\Delta,q}(n)$, where the big-$O$ notation hides a factor depending only on $\Delta$ and $q$.
\end{theorem}

Moreover, if we consider a slightly more robust dynamics then we can obtain optimal relaxation time down to $\Delta+1$ colors.  We consider the heat-bath dynamics on a pair of adjacent edges as formalized by the following process.  From a coloring $X_t$, choose a pair of adjacent edges $e_1,e_2$ in $G$ uniformly at random ($e_1$ is potentially equal to $e_2$) and then resample the edge-coloring on $\set{e_1, e_2}$ uniformly at random from those assignments consistent with the coloring on $G\setminus \{e_1,e_2\}$.  We refer to this chain as the {\em heat-bath neighboring edge dynamics}.

\begin{theorem}
\label{thm:heat-bath-pair-edges-mixing}
  For any tree with maximum degree $\Delta$ and $n$ vertices, any $q \geq \Delta + 1$, then the relaxation time of the heat-bath neighboring edge dynamics for $q$-edge-colorings is $O_{\Delta,q}(n)$.
\end{theorem}

We extend our optimal relaxation time bound for the Glauber dynamics with $q\geq\Delta+2$ to a near-optimal mixing time bound on regular trees.
A tree is said to be $\Delta$-regular complete if every non-leaf vertex has degree $\Delta$ and all leaf vertices are on the same level.  The root of the tree has $\Delta$ children and other non-leaf vertices have $d = \Delta - 1$ children.
For the case of regular trees, Delcourt, Heinrich, and Perarnau~\cite[Section 5]{DHP} can prove $O(n^{2+o_d(1)})$ mixing time when $q\geq\Delta+1$.  We obtain near-optimal mixing time bounds for regular trees.

\begin{theorem}
\label{thm:tmix-regular}
 For all $\Delta\geq 2$, all $q\geq\Delta+2$, for the $\Delta$-regular complete tree with $n$ vertices, the Glauber dynamics for $q$-edge-colorings has mixing time $O_{\Delta,q}(n\log^2{n})$.  Moreover, for an arbitrary tree $T$ of maximum degree $\Delta$ we obtain mixing time $O_{\Delta,q}(nD\log{n})$ where $D$ is the diameter of $T$.
\end{theorem}

  For regular trees we can further improve the relaxation time bound for Glauber dynamics in \Cref{thm:tmix-regular} to $q=\Delta+1$  but at the expense of a slightly worse exponent.
  
\begin{theorem}
\label{thm:regular-refined}
 For any $\Delta$-regular complete tree with $n$ vertices, the Glauber dynamics for  $(\Delta+1)$-edge-colorings has relaxation time $\Trelax = \Delta n^{1+ O(1/\log \Delta)}$ and mixing time $\Tmix= (\Delta \log \Delta)^2 \cdot n^{1+ O(1/\log \Delta)}$, where the big-$O$ notation hides a universal constant. 
\end{theorem}

We prove a lower bound for the relaxation time of the Glauber dynamics in \Cref{thm:lower-relax}.
This indicates that the relaxation time bound in \Cref{thm:regular-refined} is asymptotically tight for large~$\Delta$.
However, for regular trees with small degree $\Delta$, the picture is still very mysterious:
when $\Delta = 2$, there is an $\Omega(n^3)$ lower bound for the relaxation time due to~\cite{DGJ};
when $\Delta$ is sufficiently large, \Cref{thm:regular-refined} gives an $\Delta n^{1+o(1)}$ upper bound for the relaxation time.
Based on \Cref{thm:lower-relax}, we conjecture that the $\Delta n$ relaxation time holds as long as $\Delta \geq 3$.

\begin{theorem}
\label{thm:lower-relax}
  Let $\^T$ be a tree with $n$ vertices and maximum degree $\Delta$.
  Suppose there exists an edge $e$ in $\^T$ such that both endpoints of $e$ are of degree $\Delta$.
  Then, for the Glauber dynamics on $q$-edge-colorings with $q \in [\Delta+1, 2\Delta]$, the relaxation time satisfies $\Omega(\Delta n/(q-\Delta)^2)$ and the mixing time satisfies $\Omega(\Delta n\log{n})$. 
\end{theorem}

We provide here a very brief high-level overview of the proof of \Cref{thm:main-mixing} which establishes an optimal relaxation time bound on the Glauber dynamics when $q\geq\Delta+2$; a detailed overview is provided in \Cref{sec:glauber-relax-time}.
We bound the spectral gap of the Glauber dynamics by establishing a property known as approximation tensorization of variance; this is an important notion in the recent flurry of results using spectral independence, e.g., see~\cite{CLV21,BCCPSV21}.  

For a graph $G=(V,E)$, consider a distribution $\mu$ on $\Omega\subset [q]^E$ where $[q]=\{1,\dots,q\}$ for integer $q\geq 2$.  If $\mu$ is a product distribution (i.e., each $\e\in E$ is mutually independent of each other) then
the variance of any functional $f:\Omega\rightarrow\mathbb{R}$ under the distribution $\mu$ can be upper bounded by the sum of conditional variances on each coordinate $e\in E$; this property is known as tensorization of variance~\cite{CMT15}.  Our goal is to establish an approximate version of tensorization so that the variance functional $\Var[\mu]{f}$ is upper bounded by a weighted sum of the variances on each coordinate, where the weights are given by $C(e)$ for each $e\in E$.  It turns out that approximate tensorization of variance with constants $C(e)\leq C$ for some constant $C\geq 1$ is equivalent to the spectral gap of the Glauber dynamics being $\geq 1/(Cn)$ (and hence the relaxation time satisfying $\Trelax\leq Cn$).  We refer the reader to \Cref{def:tensorization} for a more formal introduction.

We establish approximate tensorization
of the $\Delta$-regular complete tree by an inductive approach.  We first establish a ``modified'' approximate tensorization property for a $\Delta$-regular complete tree of height $\ell=O(\Delta^2\log^2{\Delta})$ where the constants $C(e)$ are small for the leaf edges and exponentially large (in $\Delta$) for the internal edges.  We then have an inductive approach to deduce approximate tensorization for an arbitrarily large tree, see \Cref{thm:induction} in \Cref{sec:glauber-relax-time} and its proof in \Cref{sec:proof-induction}, where the constants $C(e)$ are recursively defined, see \Cref{eq:F-recursion} in \Cref{sec:proof-induction}.

To establish the modified approximate tensorization property for the smaller tree of height $\ell$ we utilize the canonical paths approach~\cite{jerrum1989approximating,Sin92}.  An important aspect of the proof is that we only need a good bound on the associated congestion for transitions which recolor leaf edges. 
The modified approximate tensorization property does not directly correspond to the relaxation time of a Markov chain, and this makes our canonical path approach somewhat different from the original version~\cite{jerrum1989approximating,Sin92}.
In particular, our canonical path approach involves couplings, and also the overhead from the length of paths does not play a role here (see \Cref{sub:congestion-overview} for details).
We believe this is the first proof utilizing the canonical paths approach in combination with approximate tensorization of variance/entropy.

Our canonical path analysis for Glauber dynamics requires $q \geq \Delta + 2$.
This is mainly because our current analysis works for the $\Delta = 2$ case and the $\Delta \geq 3$ case simultaneously.
 A good canonical path of this type should not exist for $q = \Delta + 1$ since otherwise our proof gives $O(n)$ relaxation time for Glauber dynamics when $\Delta = 2$ and $q = 3$, which contradicts the $\Omega(n^3)$ lower bound.
In practice, when $q = \Delta + 1$, we can find bad examples to fail the canonical path analysis for Glauber dynamics (see \Cref{fig:bad-example} for details).

However, the heat-bath neighboring edge dynamics in \Cref{thm:heat-bath-pair-edges-mixing} works even when $q = \Delta + 1$.
When $\Delta = 2$ and $q = 3$, one can show $\widetilde{O}(n)$ mixing time for this dynamics via path coupling, which means we do not have a similar theoretical barrier as the Glauber dynamics.
In practice, our canonical path analysis for the heat-bath neighboring edge dynamics heavily relies on the updates which exchange the colors of two neighboring edges if the resulting coloring is proper.
Simulating this kind of
update via the Glauber dynamics updates turns out to be very challenging. 
It still remains open whether a good collection of canonical paths exist for the Glauber dynamics when $q \geq \Delta + 1$ and $\Delta \geq 3$.

 In \Cref{sec:glauber-relax-time} we give a detailed overview of the proof of \Cref{thm:main-mixing} for the Glauber dynamics when $q\geq\Delta+2$; 
the proofs of the main technical results are provided in \Cref{sec:proof-induction,sec:canonical-path}.  
We present an overview of the refined approach for the neighboring edge dynamics for the case $q=\Delta+1$ 
in \Cref{sec:edge-overview}.
The proof of \Cref{thm:tmix-regular} is also presented in \Cref{sec:glauber-relax-time}, 
and the refinement from \Cref{thm:regular-refined} is presented in \Cref{sec:regular-refined}.   The proof of the lower bound result in \Cref{thm:lower-relax} is given in \Cref{sec:lower-relax}.

\paragraph{Related works} Besides edge colorings on trees~\cite{DHP}, the mixing time and relaxation time of Glauber dynamics has been studied extensively for other Gibbs distributions on trees~\cite{MSW03,berger2005Glauber,MSW04,lucier2009glauber,goldberg2010mixing,tetali2012phase,restrepo2014phase,sly2017glauber,eppstein2023rapid,chen2023combinatorial, efthymiou2023optimal, chen2024spectral}.

Many of these results on trees~\cite{MSW03, MSW04, tetali2012phase, sly2017glauber} follow the paradigm developed in~\cite{MSW03}.
In these approaches, the relaxation time for a block dynamics is proved first.
The block dynamics picks a vertex uniformly at random and resamples the sub-tree of some constant depth beneath that vertex.
Then the relaxation time for the Glauber dynamics is proved via a simple comparison argument between the block dynamics and the Glauber dynamics.
We note that this approach does not suit the edge coloring problems considered here, since even the comparison between the heat-bath neighboring edge dynamics and the Glauber dynamics is non-trivial.

Recently, the Glauber dynamics on trees has been re-studied~\cite{efthymiou2023optimal, chen2024spectral} by using the lens of \emph{spectral independence}~\cite{anari2020spectral}.
In fact, this framework can also be used to deduce rapid mixing of the Glauber dynamics for edge colorings on general graphs~\cite{AbdolazimiLiuOveisGharan, WZZ24} as detailed earlier.
These works offer some powerful tools that can be applied to general Gibbs distributions.
However, we also note that these approaches cannot apply to the edge coloring problems considered in this paper as spectral independence for arbitrary ``pinnings'' (see \Cref{sec:tensor-prelim}) does not hold in our setting as the corresponding list coloring problem may be a disconnected state space.

\section{Preliminaries}
\label{sec:prelim}

For a positive integer $q$, let $[q]:=\{1,\dots,q\}$.  
Let $G=(V,E)$ be a graph. 
For a vertex $v\in V$, we let $N(v)$ denote the neighbors of $v$, and for an edge $e\in E$ we let $N(e)$ denote the edges which share an endpoint with $e$.  
Let $q\geq 2$ be an integer, denote the set of $q$-edge colorings on $G$ as
\[ \Omega = \Omega_G := \{\sigma\in [q]^E: \mbox{ for all } e\in E, f\in N(e), \sigma(e)\neq \sigma(f)\}.
\]

Let $\mu$ denote the uniform distribution over $\Omega\subset [q]^E$.
For any subset $S \subset E$ and any $\sigma\in\Omega$, we use 
$\mu^{\sigma_S}$ to denote the distribution of $\mu$ conditional on the configuration on $S$ is fixed to be $\sigma_S$. 
Note, a random sample $X \sim \mu^{\sigma_S}$ is a full configuration on $E$ where $X_S = \sigma_S$.
Moreover, for any subset $T \subset E$, we use $\mu_T^{\sigma_S}$ to denote the marginal distribution on $T$ projected from  $\mu^{\sigma_S}$.

More generally, let $\+L_e \subseteq [q]$ be a color list for each edge $e \in E$.
We can define $\+L$-list edge-colorings for $G$ if we further require the color of each edge $e$ can only be chosen from its color palette $\+L_e$.
Denote the set of $\+L$-list edge-colorings on $G$ as
\[
  \Omega_{G, \+L} := \set{\sigma: \text{for all } e \in E, \sigma_e \in \+L_e} \cap \Omega_G.
\]
We will use $\mu_{G,\+L}$ to denote the uniform distribution over $\Omega_{G, \+L}$.
The conditional and marginal distributions can be defined accordingly.

\subsection{Markov chain fundamentals}

For a pair of distributions $\nu,\pi$ on $\Omega$, their total variation distance is denoted as
$\DTV{\nu}{\pi} := \frac12\sum_{\sigma\in\Omega}|\nu(\sigma)-\pi(\sigma)|$.

Consider an ergodic Markov chain $(X_t)_{t\geq 0}$ on state space $\Omega$ with transition matrix $P$ and unique stationary distribution $\pi$.
The {\em mixing time} $\Tmix$ is defined to be the minimum time, from the worst initial state $X_0$, to ensure that $X_t$ is within total variation distance $\leq 1/4$ from $\pi$.

Let $1=\lambda_1>\lambda_2\geq \dots \lambda_{|\Omega|}>-1$ denote the eigenvalues of $P$.  The {\em absolute spectral gap} is the quantity $1-\lambda_*$ where $\lambda_*=\max\{|\lambda_2|,|\lambda_{\Omega}|\}$.  Since all eigenvalues of the Glauber dynamics are non-negative~\cite{DGU14}, then $\lambda_*=1-\lambda_2$.  The inverse of the spectral gap $\Trelax:=1/\lambda_*$ is usually known as the {\em relaxation time} of the Glauber dynamics. 
The following is a standard bound on the mixing time in terms of the relaxation time (e.g., see~\cite[Theorem 12.4]{levin2017markov}):
\begin{align*}
  \Tmix \leq \Trelax (1-\log(\min_{\sigma\in\Omega}\mu(\sigma))\leq \Trelax (1+n\log{q}).
\end{align*}


For a tree of maximum degree $\Delta$, the Glauber dynamics on $q$-edge colorings is ergodic when $q\geq\Delta+1$~\cite{DHP,DFFHV}.  For general graphs of maximum degree $\Delta$ one needs $q\geq 2\Delta$ to guarantee ergodicity.  Since the Glauber dynamics is symmetric the unique stationary distribution is the uniform distribution $\mu$ over $\Omega$.

\subsection{Approximate tensorization of variance}
\label{sec:tensor-prelim}

For a function $f: \Omega \to \mathbb{R}$,
we denote the expectation and variance of $f$ with respect to~$\mu$ as follows:
\[ \mu[f] = \E[\mu]{f} := \sum_{\sigma}\mu(\sigma)f(\sigma)
\ \ \mbox{ and } \ \ 
\Var[\mu]{f} = \mu[f^2]-\mu^2[f].
\]

For a subset $S\subset E$, a {\em pinning} is an assignment $\tau:S\rightarrow [q]$.  We will restrict attention to feasible pinnings which are $\tau\in [q]^S$ for some $S\subset V$ such that there exists $\sigma\in\Omega$ where $\sigma_S=\tau_S$.

For a pinning $\tau \in [q]^S$ for $S\subset E$, we use $\mu^\tau [f]$ to denote the expectation with respect to the conditional distribution $\mu^\tau$. For an edge coloring $\sigma \in \Omega$ and a subset $S \subseteq E$, denote the expectation on $E\setminus S$ with respect to the fixed assignment $\sigma$ on $S$ as
\[ \mu_S[f](\sigma) := \mu^{\sigma_{E \setminus S }}[f].
\] 
Moreover, denote the variance with respect to the conditional distribution $\mu^{\sigma_{E \setminus S }}$ by
\[\Var[S]{f}(\sigma) := \Var[\mu^{\sigma_{E \setminus S }}]{f}.
\]
We note that both $\mu_S[f]$ and $\Var[S]{f}$ can also be seen as functions on $\Omega \to \^R$.

The key property of the Gibbs distribution which implies fast mixing is known as approximate tensorization of variance.

\begin{definition}
    \label{def:tensorization}
We say that the distribution $\mu$ over $\Omega\subseteq [q]^E$ satisfies {\em approximate tensorization of variance} with constants $(C(e))_{e\in E}$ if for all $f:\Omega \to \mathbb{R}_{\geq 0}$ it holds that:
\begin{equation}
\label{eqn:tensorization}
\Var{f} \leq \sum_{e\in E}C(e)\mu(\Var[e]{f}).
\end{equation}
\end{definition}
If there is a universal constant $\gamma$ where $C(e)\leq\gamma$ for all $e\in E$ then we say approximate tensorization of variance holds with constant $\gamma$.
In words, on the RHS of~\cref{eqn:tensorization}, for each edge $e\in E$, we are drawing a sample $\sigma$ from $\mu$, then fixing the configuration $\sigma$ on $E\setminus \{e\}$ and looking at the variance of $f$ with respect to the conditional distribution $\mu^{\sigma_{E\setminus e}}$. If $\mu$ is a product distribution then approximate tensorization of variance holds with constants $C(e)=1$ for all $e\in E$, which is optimal.
Variance tensorization with constants $C(e)$ bounds the relaxation time of the Glauber dynamics as $\Trelax \leq Cn$ where $C=\max_{e\in E} C(e)$, see~\cite{CMT15,Cap23}.

\subsection{Relaxation time and log-Sobolev constant on regular trees} \label{sec:log-sob}
Let $\mu$ be a distribution with support $\Omega \subseteq [q]^E$.
For an arbitrary non-negative function $f: \Omega \to \^R_{\geq 0}$, we denote the entropy of $f$ with respect to $\mu$ as follow:
\begin{align} \label{eq:def-ent}
    \Ent[\mu]{f} = \mu\left[f\log \frac{f}{\mu[f]}\right],
\end{align}
where we use the convention that $0\log 0 = 0$.
The log-Sobolev constant~\cite{diaconis1996logarithmic} for the Glauber dynamics on $\mu$ is given by 
\begin{align} \label{eq:def-log-sob}
   \csob := \inf \set{\left. \frac{\Ent[\mu]{f}}{\sum_{e \in E} \mu[\Var[e]{\sqrt{f}}]} \right\vert f:\Omega \to \^R_{\geq 0}, \Ent[\mu]{f} \neq 0}.
\end{align}
The mixing time is related to the log-Sobolev constant $c_{\-{sob}}$ as follows:
\begin{align} \label{eq:log-sob-mixing}
\Tmix(\epsilon) \leq n \cdot \csob^{-1} \tp{\log\log \frac{1}{\mu_{\min}} + \log \frac{1}{2\epsilon^2}}.
\end{align}
We refer the interested readers to~\cite[Lemma 2.8 and Lemma 2.4]{blanca2022entropy} and the references therein.

As noted in~\cite{MSW03}, there is a crude bound that relates the log-Sobolev constant and the spectral gap on complete $\Delta$-regular trees. 
Let $\^T_\ell$ be the complete $\Delta$-regular tree with depth $\ell$.
Let $\mu_\ell$ be the uniform distribution over all the $q$-edge-coloring on $\^T_\ell$.
Let $\^T^\star_\ell$ be a $d$-ary tree of depth $\ell$ with an extra hanging root edge $r$.
Let $\mu^{\star, c}_\ell = \mu_{\^T^\star_\ell, \+L}$ be the uniform distribution over all the $\+L$-edge-list-coloring on $\^T^\star_\ell$, where the color of the root edge $r$ is fixed to $c$ and $\+L_h = [q]$ for all the other edges.
We refer the reader to \Cref{def:Tk-Tk-star} and \Cref{def:mu-k} for the formal definitions of $\^T_\ell, \^T^\star_\ell$ and $\mu_\ell, \mu^{\star,c}_\ell$.
For $\mu \in \set{\mu_\ell, \mu^{\star,c}_\ell}$, we will use $\csob(\mu)$ and $\cgap(\mu)$ to denote the log-Sobolev constant and the spectral gap for the continuous time Glauber dynamics on $\mu$.

For convenience, we refer to the variant of \cite{MSW03} in~\cite{tetali2012phase} which considers vertex coloring.
We note that their proof also works for the edge coloring with small modifications.

\begin{lemma}[\text{\cite[Lemma 26]{tetali2012phase}}] \label{lem:gap-to-sob}
   For $\ell > 0$, we have the following recursions
   \begin{align}
   \label{TVVVY:1st}
       c_{\-{sob}}^{-1}(\mu^{\star,1}_\ell) &\leq c_{\-{sob}}^{-1}(\mu^{\star,1}_{\ell-1}) + \alpha \cdot c_{\-{gap}}^{-1}(\mu^{\star,1}_\ell) \\
        \label{TVVVY:new}
       \text{and} \quad 
       c_{\-{sob}}^{-1}(\mu_\ell) &\leq c_{\-{sob}}^{-1}(\mu^{\star,1}_{\ell-1}) + \alpha \cdot c_{\-{gap}}^{-1}(\mu_\ell).
   \end{align}
   where $\alpha = \frac{\log(k-2)}{1 - 2/(k-1)}$ with $k = q^\Delta$ and we use $\csob^{-1}(\mu_0) = \csob^{-1}(\mu^{\star,1}_0) = 1$ as conventions.
\end{lemma}

We note that \Cref{lem:gap-to-sob} is a slight generalization from the version in \cite{tetali2012phase}.  
In particular, \cref{TVVVY:1st} corresponds to the statement in \cite{tetali2012phase}; here every tree has a pinned parent of the root.
We also consider the case where the root does not have a parent, see~\cref{TVVVY:new}; this inequality can be proved by the same proof as in~\cite{tetali2012phase}.

\begin{corollary} \label{cor:sob-gap}
    For $\ell > 0$, it holds that
    \begin{align}
    \label{tetali:correction}
        \csob^{-1}(\mu_\ell) &\leq \alpha \ell \cdot \max\set{\cgap^{-1}(\mu_\ell), \max_{1\leq j < \ell} \cgap^{-1}(\mu^{\star,1}_j)}.
    \end{align}
\end{corollary}

\Cref{cor:sob-gap} corrects a slight error appearing in \cite[Proof of Theorem 6]{tetali2012phase}.  The analog of \Cref{tetali:correction} in \cite{tetali2012phase} is missing the factor $\ell$ on the right-hand side; however, this appears to be due to a mismatch of continuous-time and discrete-time chain statements in \cite[Lemmas 26 and 27]{tetali2012phase}.  Consequently, \cite{tetali2012phase} use $O(n)$ relaxation time of the complete regular tree to deduce $O(n\log{n})$ mixing time, but we believe there is an extra factor of $O(\log{n})$ that appears in this mixing time bound.

\section{Relaxation time analysis}
\label{sec:overview-Glauber}

In this section we will focus on establishing a tight upper bound on the relaxation time for $\Delta$-regular complete trees when $q\geq\Delta+2$.
We will then use this to establish \Cref{thm:main-mixing} which yields $O(n)$ relaxation time for the Glauber dynamics on any tree of maximum degree $\Delta$ when $q\geq\Delta+2$.  Then using a 
refinement of the approach presented in this section we will prove \Cref{thm:heat-bath-pair-edges-mixing} (in \Cref{sec:edge-overview}) which yields $O(n)$ relaxation time for an edge dynamics when $q=\Delta+1$.

We prove the upper bound on the relaxation time by establishing approximate tensorization of variance (see~\Cref{def:tensorization}) with a constant $C=C(\Delta)$; recall this implies $\Trelax=O(n)$ for the Glauber dynamics.

We will consider the more general setting of list edge colorings.
Given a tree $\mathbb{T} = (V,E)$ and a color list $\+L_e \subseteq [q]$ for each edge $e \in E$, let $\mu_{\mathbb{T},\+L}$ denote the uniform distribution of all proper list edge colorings, where a list edge coloring assigns each edge $e$ a color from $\+L_e$ such that adjacent edges receive different colors. 
We often use $\mu$ to denote $\mu_{\mathbb{T},\+L}$ if $\mathbb{T}$ and $\+{L}$ are clear from the context.

There are two $\Delta$-regular trees that will play a key role in our analysis.
Let $\^T_k$ denote the complete $\Delta$-regular tree with $k$ edge levels (we call $k$ the depth of the tree), where the root has $\Delta$ children and other non-leaf vertices have $d = \Delta - 1$ children. All edges incident to the root are in level 1. 
For $1\leq i \leq k$, we use $L_i(\^T_k)$ to denote all edges in the $i$-th level, which is the set of edges with distance $i-1$ to $L_1(\^T_k)$ in the line graph. 
Let $\mu_k=\mu_{\^T_k,\+L}$ to denote the uniform distribution of all proper list edge colorings where $\+L_e = [q]$ for all $e\in E$.

Let $\^T_k^\star$ denote the tree formed by taking the complete $d$-ary tree of depth $k$ (all non-leaf vertices, including the root, have $d=\Delta-1$ children) and adding an additional edge $r$ adjacent to the root vertex.   In this tree, every non-leaf vertex has degree $\Delta$, and there is one additional hanging root edge $r = \set{u,v}$
incident to the root where $u$ is the root and $v$ is an additional (leaf) vertex. As before, we can define all sets $L_i(\^T_k^\star)$ for $1 \leq i \leq k$ and let $L_0(\^T_k^\star) = \{r\}$.
Let $\mu^\star_k=\mu_{\mathbb{\^T_k^\star},\+L}$ denote the uniform list edge-coloring distribution for $\^T_\ell^\star$, where $\+L_r = [q-d]$ and $\+L_h = [q]$ for other edges $h \neq r$.
Moreover, we denote $\mu^{\star,c}_k$ when the color of the hanging root edge $r$ is fixed to color $c\in [q]$.  Since this is symmetric over the colors, we simply consider $\mu^{\star,1}_k$.

The following definition summarizes the above description of the two kinds of $\Delta$-regular trees which will be relevant in this paper. 
\begin{definition} [$\^T_k, \^T^\star_k$] \label{def:Tk-Tk-star} 
  $\^T_k$ is the tree that has $\Delta$ children for the root $u$ and then branches with factor $d=\Delta-1$ until level $k$; whereas $\^T_k^\star$
  has an extra hanging root edge $r=(u,v)$ where $v$ is an additional leaf vertex and then branches with factor $d=\Delta-1$ for all non-leaf nodes until level $k$.
\end{definition}

Then, on $\^T_k$ and $\^T^\star_k$, there are three specific distributions we consider.
\begin{definition} [$\mu_k, \mu_k^\star, \mu_k^{\star,1}$] \label{def:mu-k}
  Suppose $q \geq \Delta + 1$ is fixed, we have:
  \begin{itemize}
  \item $\mu_k = \mu_{\^T_k, \+L_1}$, where $\+L_1(e) = [q]$ for all $e$;
  \item $\mu_k^\star = \mu_{\^T^\star_k, \+L_2}$, where $\+L_2(r) = [q-d]$ and $\+L_2(e) = [q]$ for $e \neq r$;
  \item $\mu_k^{\star,1} = \mu_{\^T^\star_k, \+L_3}$, where $\+L_3(r) = \set{1}$ and $\+L_3(e) = [q]$ for $e \neq r$.
  \end{itemize}
\end{definition}

\subsection{Proof overview for \Cref{thm:main-mixing,thm:tmix-regular}: Glauber dynamics for $q \geq \Delta + 2$}
\label{sec:glauber-relax-time}

Consider a tree $T=(V,E)$ of maximum degree $\Delta$.  For integer $q\geq\Delta+1$, let $\mu$  denote a distribution over a subset of $[q]^E$. (We will use $\mu$ to be either $\mu_\ell$, $\mu^{\star}_\ell$, or $\mu^{\star,1}_\ell$) Recall, {\em approximate tensorization of variance} for $\mu$ which holds with constants $C(e), e\in E$ if for all $f:\Omega \to \mathbb{R}_{\geq 0}$:
\begin{equation}
    \tag{\ref{eqn:tensorization}}
    \Var[\mu]{f} \leq \sum_{e\in E}C(e)\mu(\Var[e]{f}).
\end{equation}

We will also consider the following modified notion.  Consider the tree $\^T_k^\star$ and let $r=(u,v)$ denote the hanging root edge.  We say that {\em approximate {\bf root}-tensorization of variance} for $\mu^\star_k$ holds with constants $C(e), e\in E$ if for all $f:\Omega \to \mathbb{R}_{\geq 0}$:
\begin{equation}
    \label{eqn:root-tensorization}
\Var[\mu_k^\star]{\mu_{\^T_k^\star\setminus\{r\}}[f]} \leq \sum_{e\in E}C(e)\mu_k^\star(\Var[e]{f}).
\end{equation}

Note, $\oVar_{\mu_k^\star}[\mu_{\^T_k^\star\setminus\{r\}}[f]] = \oVar_{\sigma \sim \mu_k^\star}[ \E[\tau \sim \mu^{\star,\sigma(r)}_k ]{f(\tau)} ]$.  In words, in the LHS of \cref{eqn:root-tensorization} we look at the variance of the following quantity $X$: choose a sample $\sigma\sim\mu^\star_k$, fix the color $\sigma(r)$ of the hanging root edge $r$, and then let $X$ be the expectation of $f(\tau)$ where $\tau\sim\mu^{\star,\sigma(r)}_k$.

\begin{theorem} \label{thm:induction}
  Let $\Delta \geq 2$, $q \geq \Delta+1$ be two integers.
  Suppose there exists an integer $\ell\geq 1$ where the following holds:
  \begin{enumerate}
  \item \label{item:root-tensor} There exist constants $\alpha=(\alpha_0,\dots,\alpha_\ell)$ and 
  approximate {\bf root}-tensorization of variance for $\mu^\star_\ell$ holds with constants $C(e)=\alpha_i$ where $e\in L_i(\^T_\ell^\star)$;
  \item \label{item:broute-force} There exists a constant $\gamma$ such that for $1 \leq j \leq \ell$, both $\mu_j$ and $\mu^{\star,1}_j$ satisfy approximate tensorization of variance with constant $\gamma$.
  \end{enumerate}
  Then, for any $k \geq 1$, $\mu_k$ and $\mu_k^{\star, 1}$ satisfy approximate tensorization of variance with constants 
  \begin{align*}
      C'(e) \leq \gamma \cdot \tp{\max_{0 \leq j \leq \ell} \alpha_j \cdot \sum_{i=0}^{\ftp{k/\ell}}\alpha_\ell^i + \max\set{1, \alpha_0}}.
  \end{align*}
\end{theorem}

Note that $\ell$ can depend on $\Delta$ and $q$, and $\alpha$ and $\gamma$ can depend on $\ell$ (and hence $\Delta$ and $q$).  The proof of \Cref{thm:induction} is contained in \Cref{sec:proof-induction}.

\begin{corollary}\label{cor-Glauber}
Under the same conditions as \Cref{thm:induction}, if $\alpha_\ell<1$ then the relaxation time of the Glauber dynamics is $O_{\alpha,\gamma}(n)$ on $\^T_k$, the $\Delta$-regular tree of depth $k$. 
\end{corollary}

We now give an overview on proving \Cref{thm:main-mixing} which gives an optimal relaxation time bound for the Glauber dynamics on trees of maximum degree $\Delta$ when $q\geq\Delta+2$.
First, we use \Cref{cor-Glauber} to prove an $O_{\Delta,q}(n)$ relaxation time bound for the Glauber dynamics on $\^T_k$.
The key step is to establish the conditions in \Cref{thm:induction}.
\begin{lemma}\label{lem:verify-condition}
Let $\Delta \geq 2$, $q = \Delta + 2$ be two integers. There exists $\ell = O(\Delta^2 \log^2 \Delta)$ such that the constants $\alpha=(\alpha_0,\dots,\alpha_\ell)$ and $\gamma$ in \Cref{thm:induction} exist with $\alpha_\ell = \frac{1}{2}$, $\alpha_j = q^{\Delta^{O(\ell)}}$ for all $j < \ell$ and $\gamma = q^{\Delta^{O(\ell)}}$.
\end{lemma}


%
Combining \Cref{cor-Glauber} with \Cref{lem:verify-condition}, one can immediately obtain the relaxation time for $\Delta$-regular trees. 
To prove \Cref{thm:main-mixing},
we need to explain how the relaxation time of $\Delta$-regular trees implies the relaxation time of arbitrary trees of maximum degree $\Delta$.
As noted by the previous work of~\cite{DHP}, there is a comparison argument of the relaxation time between the Glauber dynamics on the $\Delta$-regular tree and arbitrary trees.
This comparison argument for the relaxation time can be easily extend to the log-Sobolev constant $c_{\-{sob}}$ for continuous time Glauber dynamics (see \Cref{sec:log-sob} for a definition).

\begin{lemma}[\text{\cite[Proposition 9]{DHP}}] \label{lem:monotonicity-glauber}
  Let $\^T = (V, E)$ be a tree with root $r$ and maximum degree $\Delta$.
  Let $\^T_H = (V_H, H \subseteq E)$ be a connected sub-tree of $\^T$ containing $r$.
  Let $q \geq \Delta + 1$, and $\mu$, $\nu = \mu_H$ be the distribution of uniform $q$-edge-coloring on $\^T$ and $\^T_H$, respectively.
  Then,
  \begin{itemize}
      \item $\mu$ satisfies approximate tensorization of variance with constant $C$ implies $\nu$ satisfies approximate tensorization of variance with constant $Cq$;
      \item the log-Sobolev constant on $\nu$ is bounded by $c_{\-{sob}}^{-1}(\nu) \leq q \cdot c_{\-{sob}}^{-1}(\mu)$.
  \end{itemize}
\end{lemma}
The proof of \Cref{lem:monotonicity-glauber} could be done by a brute-force comparison between each term of~\eqref{eqn:tensorization} for $\mu$ and $\nu$.
Since the notation used here is quite different from the one in~\cite{DHP}, for completeness, we include a proof of \Cref{lem:monotonicity-glauber} in \Cref{sec:monotonicity-glauber}.

We can now prove \Cref{thm:main-mixing,thm:tmix-regular}.

\begin{proof}[Proof of \Cref{thm:main-mixing}]

Given a tree $\^T$ with maximum degree $\Delta$ and $q \geq \Delta + 2$.
Let $\nu$ be the distribution of uniform $q$-edge-coloring on $\^T$.
We consider the distribution of uniform $q$-edge-coloring $\mu$ on the $(q-2)$-regular tree $\^T_k$ for sufficiently large depth $k$ and degree $\Delta' = q-2$ such that $\^T$ is a subtree of $\^T_k$.
Since $q \geq \Delta + 2$, we have $\Delta' \geq \Delta$ and such $\^T_k$ must exist.

By \Cref{cor-Glauber} and \Cref{lem:verify-condition}, we know that when $q = \Delta' + 2$, the relaxation time of the Glauber dynamics on $\mu$ is $O_{\Delta',q}(n) = O_q(n)$.
This gives $O_{q}(1)$ approximate tensorization of variance of $\mu$.
Then, after applying \Cref{lem:monotonicity-glauber}, we know that $\nu$ also has $O_{q}(1)$ approximate tensorization of variance.
This proves the $O_{\Delta, q}(n)$ relaxation time for the Glauber dynamics on $\nu$.
\end{proof}

\begin{proof}[Proof of \Cref{thm:tmix-regular}]
    Let $\^T$ be a tree with maximum degree $\Delta$ and $q \geq \Delta + 2$ and let $\nu$ be the distribution of uniform $q$-edge-colorings on $\^T$.
    We consider the distribution of uniform $q$-edge-coloring $\mu$ on the $(q-2)$-regular tree $\^T_k$ with maximum degree $\Delta' = q - 2$ and sufficiently large depth $k$.
    In particular, we choose the depth $k$ according to the following rules:
    \begin{itemize}
        \item when $\^T$ is a complete $\Delta$-regular tree of depth $\ell$ (this means $\ell = O(\log_\Delta n)$), we let $k = \ell$;
        \item otherwise we let $k = D$ to be the diameter of $\^T$.
    \end{itemize}
    Since $q \geq \Delta + 2$, we have $\Delta' \geq \Delta$.
    This means in both cases, $\^T$ is a subtree of $\^T_k$.
    
    By \Cref{thm:induction} and \Cref{lem:verify-condition}, when $q = \Delta' + 2$, we know that for any $j \geq 1$, both $\mu_j$ and $\mu^{\star,1}_j$ satisfy approximate tensorization of variance with constants $O_{\Delta', q}(1) = O_q(1)$.
    This implies the relaxation time (or, the inverse of spectral gap) of the continuous time Glauber dynamics on $\mu_j$ and $\mu^{\star,1}_j$ are also $O_q(1)$.
    That is, for any $j \geq 1$, we have
    \begin{align*}
        c^{-1}_{\-{gap}}(\mu_j) = O_{q}(1) \quad \text{and} \quad c^{-1}_{\-{gap}}(\mu^{\star,1}_j) = O_{q}(1).
    \end{align*}
    Combining above equations with \Cref{cor:sob-gap}, for any $\ell \geq 1$, it holds that $c^{-1}_{\-{sob}}(\mu) = c^{-1}_{\-{sob}}(\mu_k) = k \cdot O_{q}(1)$.
    By \Cref{lem:monotonicity-glauber}, the log-Soblev constant for $\nu$ on $\^T$ is bounded by $c^{-1}_{\-{sob}}(\nu) \leq q c^{-1}_{\-{sob}}(\mu) = k \cdot O_q(1)$.
    Finally, according to the standard relation between log-Sobolev constant and the mixing time described in \eqref{eq:log-sob-mixing}, it holds that the mixing time for the Glauber dynamics on $\nu$ is bounded by
    \begin{align*}
        \Tmix = O_q(1) \cdot k \cdot n\log n.
    \end{align*}
    Note that by our choice of $k$, we have $k = O(\log_\Delta n)$ when $\^T$ is a regular complete tree; and $k = D$ be the diameter of $\^T$ when $\^T$ is an arbitrary tree.
    This finishes the proof.
\end{proof}

\subsection{Relaxation time and root-tensorization for complete regular tree}
\label{sub:congestion-overview}
In this section, we will outline the proof of \Cref{lem:verify-condition}.
In \Cref{lem:verify-condition}, we use a brute-force bound for all $\alpha_j$ with $0\leq j < \ell$ and $\gamma$. The most difficult part is to guarantee $\alpha_\ell = \frac{1}{2} < 1$ so that we can apply \Cref{cor-Glauber} to obtain \Cref{thm:main-mixing}.
In this subsection, we explain how to establish the \emph{first condition} in \Cref{thm:induction}, namely, the existance of constants $\alpha=(\alpha_0,\dots,\alpha_\ell)$ with $\alpha_\ell < 1$.


In \Cref{lem:verify-condition}, we assume $q = \Delta +2$. 
We first introduce some simplified notation that will be used in this subsection.
Since we focus on constants $\alpha$, we only care about the tree $\^T_\ell^\star$ (see \Cref{{def:Tk-Tk-star}}), and the list coloring distribution $\mu^\star_\ell$ in \Cref{thm:induction}. 
Recall that $\mu^\star_\ell$ is the uniform distribution of list colorings in $\^T_\ell^\star$, where the root hanging edge has the color list $[q-d]=[3]$ and all other edges have color list $[q]$.
For simplicity of notation, we drop the index $\ell$ and denote $\^T^\star = \^T^\star_\ell$. Also, let $\widetilde{\^T} := \^T^\star - r$. We use~$\mu$ to denote the distribution $\mu^\star_\ell$.

Our goal is to establish approximate root-tensorization of variance as defined in \eqref{eqn:root-tensorization}.  Recall this states that for any function $f:\Omega(\mu) \to \^R$, the following holds:
\begin{align}\label{eq:root-tensor}
  \Var[\mu]{\mu_{\widetilde{\^T}}[f]} 
  &\leq \sum_{i=0}^\ell \alpha_i \sum_{h \in L_i} \mu[\Var[h]{f}],
\end{align}
where we use $L_i$ to denote $L_i(\^T^\star)$ to simplify notation.

Let us expand the LHS of \eqref{eq:root-tensor}.  For a color $a\in [q]$, let $\mu^{ra}$ denote the conditional distribution $\mu^{r \gets a}$.
From the definition of variance we have the following identity:
\begin{equation}
\label{eqn:var-breakup}
 \Var[\mu]{\mu_{\widetilde{\^T}}[f]} 
   =\frac{1}{2} \sum_{a, b \in [q-d]: a \neq b} \mu_r(a)\mu_r(b)\tp{\mu^{ra}[f] - \mu^{rb}[f]}^2.
  \end{equation}
  
Our technique for bounding $\Var[\mu]{\mu_{\widetilde{\^T}}[f]}$, using the decomposition in \eqref{eqn:var-breakup}, is inspired by two classical tools: coupling and canonical paths. However, we use them in a non-standard manner.
For two colors $a,b \in [q-d]$ with $a \neq b$, let $\+C_{ab}$ denote a coupling of the distributions $\mu^{ra}$ and $\mu^{rb}$. 
Formally, for any random sample $(\sigma,\tau) \sim \+C_{ab}$, it holds that $\sigma \sim \mu^{ra}$ and $\tau \sim \mu^{rb}$.
Let $\Omega({\+C_{ab}}) \subseteq \Omega(\mu^{ra}) \times \Omega(\mu^{rb})$ denote the support of the coupling $\+C_{ab}$.
Let $\+C = \{\+C_{ab} \mid a,b \in [q-d], a \neq b\}$ denote the set of couplings and let $\Omega(\+C) = \bigcup_{\+C_{ab} \in \+C}\Omega(\+C_{ab})$ denote the union of their supports.
For two edge colorings $\gamma$ and $\gamma'$, we will use $\gamma \oplus \gamma' := \set{e: \gamma_e \neq \gamma'_e}$ to denote the set of edges on which $\gamma$ and $\gamma'$ are different.

Given a set of couplings, we will define the following collection of canonical paths.  Let $\Lambda(\mu)$ denote the ordered pairs $(\gamma,\gamma')\in\Omega(\mu)$ where
$|\gamma\oplus\gamma'|=1$, i.e., $\Lambda(\mu)$ are the pairs of states with positive probability in the transition matrix for the Glauber dynamics.

\begin{definition} [Canonical paths w.r.t. a coupling set]
Given a coupling set $\+C$, for any pair $(\sigma,\tau) \in \Omega(\+C)$, a canonical path $\gamma^{\sigma,\tau}$ is a simple path from $\sigma$ to $\tau$ in the graph $(\Omega(\mu),\Lambda(\mu))$.  More specifically,
the path $\gamma^{\sigma,\tau}=(\gamma_0, \gamma_1,\ldots,\gamma_{m})$ with the following properties: (1) $\sigma = \gamma_0$, $\tau = \gamma_{m}$ and all $\gamma_i$'s are distinct; (2) for any $1 \leq i \leq m$, $\gamma_{i-1}$ and  $\gamma_{i}$ differ only at the color on one edge.

Let $\Gamma = \{\gamma^{\sigma,\tau} \mid (\sigma,\tau) \in \Omega(\+C) \}$ denote the set of canonical paths w.r.t. the coupling set $\+C$.
\end{definition}


In the following we use $Q^{\text{Glauber}}_\mu(\gamma,\gamma')$ to denote the transition rate from $\gamma$ to~$\gamma'$ in the (continuous-time) heat-bath Glauber dynamics for $\mu$. 
Formally, 
\[
Q^{\text{Glauber}}_\mu(\gamma,\gamma') = \mu^{\gamma_{\^T \setminus e }}(\gamma') = \mu_e^{\gamma_{\^T \setminus e }}(\gamma'_e).
\]
This corresponds to the transition probability from $\gamma$ to $\gamma'$ for the heat-bath (discrete-time) Glauber dynamics conditional on the edge $\gamma \oplus \gamma'$ is recolored.

We define the following quantity as the expected congestion of the canonical paths $\Gamma$.

\begin{definition}[Congestion of canonical paths w.r.t. a coupling set]
Given a coupling set $\+C$ and a set of canonical paths $\Gamma$, define the expected congestion at level $0 \leq t \leq \ell$ by
\begin{equation}\label{eq:def-congestion}
    \xi_t = \max_{a,b\in [q-d]:\atop a\neq b} \sum_{s=(\gamma,\gamma')\in\Lambda(\mu):\atop (\gamma\oplus\gamma')\in L_t} \frac{ \tp{\Pr[(\sigma,\tau) \sim \+C_{ab}]{s\in\gamma^{\sigma,\tau}   } }^2}{\mu(\gamma)Q^{\text{Glauber}}_\mu(\gamma,\gamma')}.
\end{equation}
\end{definition}
The congestion in~\eqref{eq:def-congestion} sums over all pairs $(\gamma,\gamma') \in \Omega(\mu) \times \Omega(\mu)$ that differ only at one edge in level $L_t$. 
The numerator is the square of the probability that a random canonical path $\gamma^{\sigma,\tau}$, where $(\sigma,\tau) \sim \+C_{ab}$, uses $(\gamma,\gamma')$. 
The denominator is the capacity of $(\gamma,\gamma')$ in the (continuous-time) Glauber dynamics.

\begin{remark}[Comparison to standard canonical paths approach] The canonical paths technique is a classical tool for upper bounding the mixing time of Markov chains, see~\cite{jerrum1989approximating,DG91,Sin92}.
In the standard approach, we define a canonical path $\gamma^{\sigma,\tau}_{\text{classical}}$ between \emph{every pair} of colorings $\sigma,\tau \in \Omega(\mu)$.
The goal is to define canonical paths which minimize the congestion defined to be the maximum ``load'' through any transition.  The congestion then yields a bound on the spectral gap of the Glauber dynamics (equivalently, it establishes approximate tensorization of variance).

The standard canonical paths technique does not involve any coupling.
Here, our purpose is to establish approximate \textbf{root}-tensorization of variance. Our canonical path is constructed with respect to a set of couplings~$\+C$: for the coupling $\+C_{ab}$, we are going from a coloring $\sigma$ where the root edge $r$ has a color $a$ to a coloring $\tau$ where the root edge $r$ has color $b$, such that $(\sigma, \tau) \sim \+C_{ab}$.

The congestion~\cite{Sin92} in the classical canonical paths approach is defined as follows:
\begin{align*}
    \rho := \max_{(\gamma,\gamma')} \frac{1}{\mu(\gamma)Q_\mu^{\text{Glauber}}(\gamma,\gamma')} \sum_{\substack{(\sigma,\tau) \in \Lambda(\mu):\\ (\gamma,\gamma') \in \gamma^{\sigma,\tau}_{\text{classical}}}} \mu(\sigma)\mu(\tau) |\gamma^{\sigma,\tau}_{\text{classical}}|,
\end{align*}
where $|\gamma^{\sigma,\tau}_{\text{classical}}|$ is the length of the path $\gamma^{\sigma,\tau}_{\text{classical}}$.
Comparing $\rho$ with our $\xi_t$, there are some important differences. (1) Our definition of congestion involves the probability of a pair in the coupling~$\+C$. (2) We measure the total congestion in a level rather than the congestion in a single transition. (3) Our congestion does not contain the length of the path, this is because when we use $\xi_t$ to establish~\eqref{eq:root-tensor}, we apply the Cauchy-Schwarz inequality in a different manner than in~\cite{Sin92}, see the proof of \Cref{lem:con-var} for more details.

The most important difference is in our application of the technique.  To obtain a tight bound on the mixing time we need $\xi_\ell<1/(\ell+1)$ (which will imply $\alpha_\ell<1$) and we only need a trivial bound for levels $t<\ell$, see~\Cref{lem:verify-condition}.
\end{remark}

Finally, if we construct a set of canonical paths w.r.t. a coupling set $\+C$ which has congestion $\*\xi=(\xi_0,\xi_1,\dots,\xi_\ell)$ then we obtain approximate root-tensorization of variance with constant $\*\alpha=(\alpha_0,\dots,\alpha_\ell)$ where $\alpha_i=(\ell+1)\xi_i$.


\begin{lemma}\label{lem:con-var}
If there exist a set of couplings $\+C$ and a set $\Gamma$ of canonical paths such that the congestion with respect to the coupling $\+C$ is $\xi_t$ for $0 \leq t \leq \ell$, then~\eqref{eq:root-tensor} holds with $\*\alpha = (\ell+1)\*\xi$. 
\end{lemma}

\begin{proof}
 We note that by definition
  \begin{align*}
    \Var[\mu]{\mu_{\widetilde{\^T}}[f]}
    \nonumber &= \frac{1}{2} \sum_{a, b \in [q-d]:a\neq b} \mu_r(a)\mu_r(b) \tp{\mu^{ra}(f) - \mu^{rb}(f)}^2 \\
   \text{(by coupling)}\quad &= \frac{1}{2} \sum_{\substack{a, b \in [q-d]: a\neq b}} \mu_r(a) \mu_r(b) \tp{\E[(\sigma,\tau) \sim \+C_{ab}]{f(\sigma) - f(\tau)}}^2\\
   \text{(by telescoping sum)}\quad &= \frac{1}{2} \sum_{\substack{a, b \in [q-d]: a\neq b}} \mu_r(a) \mu_r(b) \tp{\E[(\sigma,\tau) \sim \+C_{ab}]{ \sum_{j=1}^{m(\sigma,\tau)} f(\gamma^{\sigma\tau}_j) - f(\gamma^{\sigma\tau}_{j-1})}}^2.
  \end{align*}    
Fix $a$ and $b$. Denote $ \E[(\sigma,\tau) \sim \+C_{ab}]{ \sum_{j=1}^{m(\sigma,\tau)} f(\gamma^{\sigma\tau}_j) - f(\gamma^{\sigma\tau}_{j-1})}$ by $E_{ab}$. 
Then $E_{ab}$ can be rewritten as
\begin{align} \label{eq:Eab}
 E_{ab} = \sum_{\substack{s = (\gamma,\gamma') \in \Omega(\mu) }}(f(\gamma)-f(\gamma'))\Pr[(\sigma,\tau) \sim \+C_{ab}]{s \in \gamma^{\sigma,\tau} },
\end{align}
where the sum is over all \emph{ordered} pairs $(\gamma,\gamma')$ that differ at one edge and the equation uses the fact that every canonical path uses $(\gamma,\gamma')$ at most once. We can partition all $(\gamma,\gamma')$'s according to the level of the edge $e = \gamma \oplus \gamma'$. By using Cauchy-Schwarz, 
\begin{align} \label{eq:def-Ai}
E_{ab}^2
\leq (\ell + 1) \sum_{i=0}^\ell \tp{ \sum_{\substack{s = (\gamma,\gamma') \in \Lambda(\mu):\\ \gamma \oplus \gamma' \in L_i }}(f(\gamma)-f(\gamma'))\Pr[(\sigma,\tau) \sim \+C_{ab}]{s \in \gamma^{\sigma,\tau}  }}^2 := (\ell+1)\sum_{i=0}^\ell A_i^2,
\end{align}
where we use $A_i^2$ to denote the term for each level $i$.
Let $Q(\gamma,\gamma')$ denote $\mu(\gamma)Q^{\text{Glauber}}_\mu(\gamma,\gamma')$.
By applying Cauchy-Schwarz on $A_i^2$ we obtain the following bound:
\begin{align*}
    A_i^2 &=\bigg(\sum_{\substack{s = (\gamma,\gamma') \in \Lambda(\mu): \\ \gamma \oplus \gamma' \in L_i }}{(f(\gamma)-f(\gamma'))\sqrt{Q(\gamma,\gamma')}}\frac{\Pr[(\sigma,\tau) \sim \+C_{ab}]{s \in \gamma^{\sigma,\tau}  } }{\sqrt{Q(\gamma,\gamma')}}\bigg)^2 \\
&\leq\sum_{\substack{s = (\gamma,\gamma') \in \Lambda(\mu):\\ \gamma \oplus \gamma' \in L_i }} (f(\gamma)-f(\gamma'))^2Q(\gamma,\gamma') \sum_{ \substack{s = (\gamma,\gamma') \in \Lambda(\mu) \\ \gamma \oplus \gamma' \in L_i }}\frac{\tp{\Pr[(\sigma,\tau) \sim \+C_{ab}]{s \in \gamma^{\sigma,\tau}  }}^2}{Q(\gamma,\gamma')}\\
    &\leq 2\xi_i\sum_{e \in L_i} \mu[\Var[e]{f}].
\end{align*}
Combining everything together, we have
\begin{align*}
   \Var[\mu]{\mu_{\widetilde{\^T}}[f]} \leq \sum_{\substack{a, b \in [q-d]: a\neq b}} \mu_r(a) \mu_r(b)  (\ell+1)\sum_{i=0}^\ell \xi_i\sum_{e \in L_i} \mu[\Var[e]{f}] \leq (\ell+1)\sum_{i=0}^\ell \xi_i\sum_{e \in L_i} \mu[\Var[e]{f}].
\end{align*}
Hence, approximate root-tensorization of variance in~\eqref{eq:root-tensor} holds with $\*\alpha = (\ell+1)\*\xi$.
\end{proof}

Using the technique developed above, we prove \Cref{lem:verify-condition} via the following result.

\begin{lemma}\label{lem:expectation-bound}
Let $\Delta \geq 2$, $q = \Delta + 2$ be two integers. There exists $\ell = O(\Delta^2 \log^2 \Delta)$ such that the following holds for the uniform distribution $\mu= \mu_{\ell}^\star$ of the list colorings on $\^T^\star = \^T_\ell^\star$. There exists a set of couplings~$\+C$ and canonical paths $\Gamma$ w.r.t. $\+C$ such that the congestion $\xi$ satisfies $\xi_\ell = \frac{1}{2(\ell+1)}$ and $\xi_j = q^{\Delta^{O(\ell)}}$ for $j < \ell$.
\end{lemma}

The high-level idea of the canonical path construction and the analysis of the associated congestion used in the proof of \Cref{lem:expectation-bound} is the following (see \Cref{sec:canonical-path} for more details).  We first define, for a pair of colors $a,b\in [q]$, the coupling $\+C_{ab}$ between a coloring $\sigma$ with the root edge having color $\sigma(r)=a$ and a coloring $\tau$ with $\tau(r)=b$.  To define this coupling, given $\sigma$ we consider the $(a,b)$-colored alternating path $\+E^*$ from the root edge $r$, and we define $\tau$ as the coloring obtained by flipping $\+E^*$, which means we interchange the colors $a$ and $b$ on $\+E^*$; this results in a valid coloring since $\+E^*$ is a maximal 2-colored component.

Now consider a pair of colorings $\sigma$ and $\tau$ which
only differ by a $(a,b)$-colored alternating path~$\+E^*$ from the root edge $r$. Let $\+E^*=(e_0,e_1,\dots,e_s)$ where $e_0=r$ and the edges are in order from the root towards the leaves.  In the first stage we aim to recolor the odd edges from $\sigma(e_i)=b$ to a color $c_i$ different than $a$ and $b$.  To do that we recolor a path $\+E_i$, which is defined starting from a sibling of $e_i$ towards the leaves.  By recoloring $\+E_i$ with particular colors we will free a color $c_i\neq a,b$ for edge $e_i$. However, we recolor the edges of $\+E_i$ and the odd edges of $\+E^*$ in the following manner.  Let $\+E$ denote the set of all edges recolored in this stage.  We order the edges of $\+E$ so that edges on lower levels (closer to the leaves) are recolored before edges on higher levels of the tree, and within a level the edges of the alternating path $\+E^*$ are last.  We then recolor the edges of $\+E$ in this order.  

In the second stage we recolor the even edges of $\+E^*$ from color $a$ to color $b$, which is available due to the first stage.  Finally in the last stage we recolor the edges of $\+E$ to their color in $\tau$, and we do these recolorings in reverse of the ordering used in the first stage.
The key to the congestion analysis is that we only need a non-trivial bound on the congestion for transitions which recolor leaf edges.  This is why in the first stage we first recolor all leaf edges before proceeding up the tree, and in the last stage we recolor the leaf edges last. 
When calculating the congestion on leaf edges, a specific transition on a leaf edge can still be used by exponentially large number of canonical paths. However, such large number of possibilities will be offset by the exponentially small probability of the path $\+E_i$ reaching the leaves.

\subsection{Proof overview for \Cref{thm:heat-bath-pair-edges-mixing}: Heat-bath edge dynamics for $q=\Delta+1$}
\label{sec:edge-overview}

In this section, we will prove \Cref{thm:heat-bath-pair-edges-mixing}.
Recall, approximate tensorization of variance implies optimal relaxation time for the Glauber dynamics.  Approximate tensorization can be generalized to a block factorization, which implies fast mixing for the more general block dynamics~\cite{CP21}.

\begin{definition} \label{def:block-factorization}
We say that the distribution $\mu$ over $\Omega\subseteq [q]^E$ satisfies {\em block factorization of variance} with constants $(C(B))_{B \subseteq E}$ if for all $f:\Omega \to \mathbb{R}_{\geq 0}$ it holds that:
\begin{align}
  \label{eqn:block-factorization}
  \Var{f} \leq \sum_{B \subseteq E} C(B) \mu(\Var[B]{f}).
\end{align}
For convenience we denote the support of $C$ as $\+B := \set{B \subseteq E \mid C(B) > 0}$.
\end{definition}

In particular, for $\+B = \set{\set{e} \mid e \in E}$, which is the set of all the singleton edges, then block factorization becomes approximate tensorization.

Consider the the heat-bath block dynamics $P_B$ defined by a collection of constants $(C(B))_{B \subseteq E}$ from \Cref{def:block-factorization}.
It updates a configuration $X_t$ to $X_{t+1}$ as follows:
\begin{enumerate}
\item pick a block $B \subseteq E$ with probability proportional to $C(B)$;
\item sample $X_{t+1} \sim \mu(\cdot \mid X_t(E\setminus B))$.
\end{enumerate}
Then \eqref{eqn:block-factorization} immediately implies an upper bound on the relaxation time of $P_B$ as $\Trelax(P_B) \leq \sum_{B\subseteq E} C(B)$.

We recall the definition of the heat-bath edge dynamics.
Let $\^T = (V, E)$ be a tree with maximum degree $\Delta$ and let
\begin{align} \label{def:P}
  \+P = \+P(\^T) := \set{\set{e} \mid e \in E} \cup \set{\set{e, e'} \mid \text{$e$, $e'$ are adjacent edges in $\^T$}}.
\end{align}
Also, let $q \geq \Delta + 1$ and let $\mu$ be the uniform distribution of $q$-edge-coloring on $\^T$.
Then, the heat-bath neighboring edge dynamics is exactly the heat-bath block dynamics defined by constants \[C(B) = \*1[B \in \+P], \text{ for } B \in \+P.\]

In order to prove $O_{\Delta,q}(n)$ relaxation time of the heat-bath edge dynamics, it is sufficient to prove the following tensorization bound of variance, for every $f:\Omega(\mu) \to \^R$, there is a constant $C = C(\Delta, q)$ such that:
\begin{align} \label{eq:def-edge-tensorization}
  \Var[\mu]{f} &\leq C \sum_{S \in \+P} \mu[\Var[S]{f}].
\end{align}
For convenience, if \eqref{eq:def-edge-tensorization} holds, then we say $\mu$ satisfies \emph{approximate edge-tensorization} with constant~$C$.
We will prove \eqref{eq:def-edge-tensorization} with a similar high-level plan as in the proof of \Cref{thm:main-mixing}.
All we need is a slightly generalized version of \Cref{thm:induction} which enables us to establish block factorization of variance with support $\+P$ on regular trees.
First, we generalize \eqref{eqn:root-tensorization} by allowing general blocks rather than singletons.
Let $\+B$ be a collection of edge sets defined as the following:
\begin{align} \label{eq:def-edge-block}
   \+B := \set{\set{e} \mid e \in \^T^\star_\ell} \cup \set{\set{e, r} \mid e \in L_1(\^T^\star_\ell)}. 
\end{align}
In words, $\+B$ contains all the singleton edges in $\^T^\star_\ell$ and the adjacent edge pairs which contain the root hanging edge $r$.
We say that {\em approximate {\bf root}-factorization of variance} for $\mu^\star_k$ holds with constants $C(B)$ if for all $f:\Omega \to \mathbb{R}_{\geq 0}$:
\begin{equation}
    \label{eqn:root-block-factorization}
\Var[\mu_k^\star]{\mu_{\^T_k^\star\setminus\{r\}}[f]} \leq \sum_{B\in \+B}C(B)\mu_k^\star(\Var[B]{f}).
\end{equation}


\begin{theorem} \label{thm:adjacentpair}
  Let $\Delta \geq 2$, $q \geq \Delta+1$ be two integers.
  Suppose there exists an integer $\ell\geq 2$ where the following holds:
  \begin{itemize}
  \item There exist constants $\alpha=(\alpha_0,\dots,\alpha_\ell)$ and $\beta$ and 
  approximate {\bf root}-factorization of variance for $\mu^\star_k$ holds with constants $C(\set{e})=\alpha_i$ where $e\in L_i(\^T_\ell^\star)$ and $C(B) = \beta$ where $\abs{B} = 2$;
  \item  There exists a constant $\gamma$ such that for $1 \leq j \leq \ell$, both $\mu_j$ and $\mu^{\star,1}_j$ satisfy approximate tensorization of variance with constant $\gamma$.
  \end{itemize}
  Then, for any $k \geq 1$, $\mu_k$ satisfies block factorizaiton of variance with constants 
  \begin{align*}
      C'(B) \leq 
     \begin{cases} 
      \beta \gamma \cdot \tp{\max_{0 \leq j \leq \ell} \alpha_j \cdot \sum_{i=0}^{\ftp{k/\ell}}\alpha_\ell^i + \max\set{1, \alpha_0}},   
       & \mbox{if } B\in\+P(\^T_k)
       \\
       0 & \mbox{if }B\not\in\+P(\^T_k).
       \end{cases}
  \end{align*}
\end{theorem}

We defer the proof of \Cref{thm:adjacentpair} to \Cref{sec:adjacentpair}.
The block factorization of variance directly implies the relaxation time of the heat-bath block dynamics.
We have the following corollary.

\begin{corollary} \label{cor:edge-dynamics}
Under the same conditions as \Cref{thm:adjacentpair}, if $\alpha_\ell<1$ then the relaxation time of the heat-bath $\+P(\^T_k)$-dynamics is $O_{\alpha,\beta,\gamma}(n)$ on $\^T_k$, the $\Delta$-regular tree of depth $k$. 
\end{corollary}

We now give an overview on proving \Cref{thm:heat-bath-pair-edges-mixing} which gives an optimal relaxation time bound for the neighboring edge dynamics on trees of maximum degree when $q \geq \Delta + 1$.
Similar to the proof of \Cref{thm:main-mixing}, we will prove the $O_{\Delta,q}(n)$ relaxation time bound on $\^T_k$ first.
The crucial part is to verify the conditions in \Cref{thm:adjacentpair}.

\begin{lemma} \label{lem:edge-dynamics-cond}
    Let $\Delta \geq 2$, $q = \Delta + 1$.
    There exists $\ell = O(\Delta^2 \log^2 \Delta)$ such that the constants $\alpha = (\alpha_0, \cdots, \alpha_\ell)$, $\beta$, and $\gamma$ in \Cref{thm:adjacentpair} exist with $\alpha_\ell = \frac{1}{2}$, $\alpha_j = q^{\Delta^{O(\ell)}}$ for all $j < \ell$, and $\beta = \gamma = q^{\Delta^{O(\ell)}}$.
\end{lemma}

The proof of \Cref{lem:edge-dynamics-cond} is very similar to the proof of \Cref{lem:verify-condition} for the Glauber dynamics.
Following the high-level plan in \Cref{sub:congestion-overview} and \Cref{sec:canonical-path}, one can prove \Cref{lem:edge-dynamics-cond} by slightly modifying the construction of the canonical path.
We omit the detailed proof for \Cref{lem:edge-dynamics-cond}.
Instead, we will give an proof overview for \Cref{lem:edge-dynamics-cond} in \Cref{sec:edge-dynamics-cond}, where we will discuss why the edge dynamics can bypass the obstacles on Glauber dynamics in the canonical path analysis.

To obtain relaxation time bound for arbitrary trees, we also need a comparison argument for the heat-bath neighboring edge dynamics which is similar to \Cref{lem:monotonicity-glauber}.
\begin{lemma} \label{lem:monotonicity-edge-dynamics}
  Let $\^T = (V, E)$ be a tree with root $r$ and maximum degree $\Delta$.
  Let $\^T_H = (V_H, H \subseteq E)$ be a connected sub-tree of $\^T$ containing $r$
  Let $q \geq \Delta + 1$, and $\mu$, $\nu = \mu_H$ be the distribution of uniform $q$-edge-coloring on $\^T$ and $\^T_H$, respectively.
  If~$\mu$ satisfies approximate edge-tensorization of variance with constant $C$, then $\nu$ satisfies approximate edge-tensorization of variance with constant $C(q+1)^2$.
\end{lemma}
The proof of \Cref{lem:monotonicity-edge-dynamics} is given in \Cref{sec:monotonicity-edge-dynamics}.

\begin{proof}[Proof of \Cref{thm:heat-bath-pair-edges-mixing}]
    The proof follows from a similar argument as the proof of \Cref{thm:main-mixing}.

   Given a tree $\^T$ with maximum degree $\Delta$ and $q \geq \Delta + 1$.
Let $\nu$ be the distribution of uniform $q$-edge-coloring on $\^T$.
We consider the distribution of uniform $q$-edge-coloring $\mu$ on the $(q-1)$-regular tree $\^T_k$ for sufficiently large depth $k$ and degree $\Delta' = q-1$ such that $\^T$ is a subtree of $\^T_k$.
Since $q \geq \Delta + 1$, we have $\Delta' \geq \Delta$ and such $\^T_k$ must exist.

By \Cref{cor:edge-dynamics} and \Cref{lem:edge-dynamics-cond}, we know that when $q = \Delta' + 1$, the relaxation time of the edge dynamics on $\mu$ is $O_{\Delta',q}(n) = O_q(n)$.
This gives $O_{q}(1)$ approximate edge-tensorization of variance of $\mu$.
Then, after applying \Cref{lem:monotonicity-edge-dynamics}, we know that $\nu$ also has $O_{q}(1)$ edge-approximate tensorization of variance.
This proves the $O_{\Delta, q}(n)$ relaxation time for the edge-dynamics on $\nu$.
\end{proof}

\section{Approx. tensorization via root-tensorization: Proof of \Cref{thm:induction}}\label{sec:proof-induction}
In this section, we prove \Cref{thm:induction}; this is the main inductive proof of approximate tensorization of variance for an arbitrarily large $\Delta$-regular complete tree using the root-tensorization property for a small tree of height $\ell$.
Without loss of generality, we focus on $\mu_k$.
The approximate tensorization of variance for $\mu^{\star,1}_k$ follows from the same proof.
Recall the definition of $\^T_k, \^T_k^\star$ and $\mu_k, \mu_k^{\star,1}$ in \Cref{def:Tk-Tk-star,def:mu-k}.

Fix parameters $\ell$, $\*\alpha$, and $\gamma$ in the theorem.
Our proof is based on an induction on the depth~$k$ of the regular tree $\^T_k$.
For every $k \geq 1$, we will show that $\mu$ satisfies approximate tensorization of variance with constants $C(e) = F_k(t)$ for $1 \leq t \leq k$ and $e \in L_t(\^T_k)$.
In other words, let $\mu = \mu_k$, we will prove the following inequality for $f:\Omega(\mu) \to \^R$,
\begin{align} \label{eq:induction-hypothesis}
  \Var[\mu]{f} &\leq \sum_{t=1}^k F_k(t) \sum_{e \in L_t(\^T_k)} \mu[\Var[e]{f}].
\end{align}
The function $F_\cdot(\cdot)$ is defined as follows: for any $1\leq t\leq s \leq \ell$, $F_{s}(t) = \gamma$; for any $s > \ell$,
\begin{align} \label{eq:F-recursion}
    F_s(t) = \begin{cases}
    F_{s-\ell}(t) &\text{if } 1 \leq t \leq s-\ell - 1\\
    \alpha_0 F_{s-\ell}(s-\ell) &\text{if } t = s - \ell\\
    \alpha_{t-s+\ell}F_{s-\ell}(s-\ell) + \gamma &\text{if } t \geq s-\ell + 1.
    \end{cases}
\end{align}
Then, the proof of \Cref{thm:induction} is reduced to giving an upper bound for $F_\cdot(\cdot)$.
We defer that part to \Cref{sec:upper-bound-F}.

\subsection{Basic properties of variance}

We begin with some basic properties about variance that will be useful in the proof of \Cref{thm:induction}.  The following lemma is a basic decomposition property of variance (e.g. see~\cite{MSW04,CLV21,CP21,Cap23}).

\begin{lemma}
  Let $\pi$ be a distribution over $\Lambda\subset [q]^n$ and $f:\Omega \to \mathbb{R}$ be a function. For any subset $S \subseteq [n]$,
\begin{align} \label{eq:var-total-law}
  \Var[\pi]{f} = \pi[\Var[S]{f}] + \Var[\pi]{\pi_{S}[f]}.
\end{align}
Furthermore, if $\pi = \pi_{S_1}\times \pi_{S_2} \times \ldots \times \pi_{S_\ell}$ is a product distribution of $\ell$ marginal distributions, where $S_1,S_2,\ldots,S_\ell$ is a partition of $[n]$, then
\begin{align} \label{eq:var-sub-add}
  \Var[\pi]{f} \leq \sum_{i=1}^\ell \pi[\Var[S_i]{f}].
\end{align}
\end{lemma}

The identity in~\eqref{eq:var-total-law} is called the law of total variance.
The inequality in~\eqref{eq:var-sub-add} states that approximate tensorization of variance holds for product distributions with constant $C=1$ which is the optimal constant.

In this paper, the distribution $\mu$ is over list edge colorings of a graph. Let $G=(V,E)$ be a graph. Let $\mu$ over $[q]^E$ be a distribution satisfying the following conditional independence property. For three disjoint $A,B,C \subseteq E$ such that the removal of $C$ disconnects $A$ and $B$ in the line graph of $G$, 
\begin{align}\label{eq:cond-ind}
 \forall \sigma \in \Omega(\mu_{C}), \quad \mu^\sigma_{A \cup B} = \mu^\sigma_A \times \mu^\sigma_B.
\end{align}

For distributions with the conditional independence property (e.g., the uniform distribution of (list) edge colorings), the following lemma holds.
\begin{lemma}[\cite{MSW04}] \label{lem:basic-var-fact}
  Let $G=(V,E)$ be a graph and $\mu$ be distribution over $[q]^E$ satisfying the conditional independence property in \eqref{eq:cond-ind}.
  For any $f: \Omega \to \^R$, any $S_1,S_2 \subseteq E$ such that $\partial S_1 \cap S_2 = \emptyset$,
  \begin{align} \label{eq:var-convex}
      \mu[\Var[S_1]{\mu_{S_2}[f]}] \leq \mu[\Var[S_1]{\mu_{S_1\cap S_2}[f]}],
  \end{align}
  where $\partial S_1 := \set{e \in E\setminus S_1 \mid \exists h \in S_1, e\cap h \neq \emptyset}$ denotes the exterior boundary of $S_1$.
\end{lemma}
The lemma in \cite{MSW04} is stated for distributions defined over $[q]^V$ instead of $[q]^E$ and the proof is omitted there, and hence we provide a proof of  \Cref{lem:basic-var-fact} in \Cref{app:proof} for completeness. 

\subsection{Establishing approximate tensorization~\eqref{eq:induction-hypothesis}}

 Under the condition of \Cref{thm:induction}, for every $k \geq 1$, we do the induction as follows.
 \paragraph{Base case ($k \leq \ell$):} 
  When $k \leq \ell$, 
  \eqref{eq:induction-hypothesis} follows from the second assumption of \Cref{thm:induction} with $F_k(t) = \gamma$ .

  \paragraph{Inductive case ($k > \ell$):}
  Suppose~\eqref{eq:induction-hypothesis} holds for all smaller $k' < k$.
  We prove that it also holds for~$k$.
  For convenience, let $\^T = \^T_k$.
  Let $R = \cup_{i=k-\ell+1}^k L_i(\^T)$ be the edge set of the bottom~$\ell$ levels of~$\^T$.
  According to the law of total variance \eqref{eq:var-total-law}, it holds that
  \begin{align} \label{eq:law-total-var}
    \Var[\mu]{f} &= \mu[\Var[R]{f}] + \Var[\mu]{\mu_R[f]}.
  \end{align}

 We give the following bound for the first term $\mu[\Var[R]{f}]$ in the RHS of~\eqref{eq:law-total-var}.

 \begin{claim}\label{claim-bound}
 Assuming \Cref{item:broute-force} of \Cref{thm:induction}, we have
  \begin{align} \label{eq:mu[Var_R[f]]}
    \mu[\Var[R]{f}] &\leq \gamma \sum_{i=1}^\ell \sum_{h\in L_{k - \ell + i}(\^T)} \mu[\Var[h]{f}].
  \end{align}
 \end{claim}
We now prove the theorem assuming the above claim. 
The proof of \Cref{claim-bound} will be given later.

  Next, we bound the second term $\Var[\mu]{\mu_R[f]}$ in the RHS of~\eqref{eq:law-total-var} using the inductive hypothesis.
  Note that for any $\sigma \in \Omega(\mu)$, the value of the function $\mu_R[f](\sigma)$ depends only on $\sigma_{\^T \setminus R}$. 
  By viewing $\mu_R[f]$ as a function $f':\Omega(\mu_{\^T \setminus R}) \to \^R$, we have
  $\Var[\mu]{\mu_R[f]} = \Var[\mu_{\^T - R}]{f'} = \Var[\mu_{{k-\ell}}]{f'}$, where the last equation uses the fact that the marginal distribution $\mu_{\^T \setminus R}$ is the same as $\mu_{{k-\ell}}$, and $\mu_{{k-\ell}}$ is the uniform distribution over proper $q$-edge-coloring in the tree $\^T_{k-\ell}$.
  For convenience, we use $\^T^h$ to denote the subtree with the hanging root edge $h$ and use $\widetilde{\^T}^h := \^T^h \setminus h$ to denote the subtree beneath the edge $h$.
  By the inductive hypothesis on $\mu_{k-\ell} = \mu_{\^T \setminus R}$, we have
  \begin{align} 
    \Var[\mu]{\mu_R[f]}
    \nonumber &\leq \sum_{t=1}^{k-\ell} F_{k-\ell}(t) \sum_{h \in L_t(\^T)} \mu_{\^T \setminus R}[\Var[h]{f'}]\\
   \label{middle-star}
   &= \sum_{t=1}^{k-\ell-1} F_{k-\ell}(t) \sum_{h \in L_t(\^T)} \mu[\Var[h]{\mu_R[f]}] 
                + F_{k-\ell}(k-\ell) \sum_{h \in L_{k-\ell}(\^T)} \mu\left[\Var[\^T^h]{\mu_{R}[f]}\right] \\
    \label{eq:ind-hypo}
              &\leq \sum_{t=1}^{k-\ell-1} F_{k-\ell}(t) \sum_{h \in L_t(\^T)} \mu[\Var[h]{f}] 
                + F_{k-\ell}(k-\ell) \sum_{h \in L_{k-\ell}(\^T)} \mu\left[\Var[\^T^h]{\mu_{\widetilde{\^T}^h}[f]}\right],
  \end{align}

where \eqref{eq:ind-hypo} holds by \eqref{eq:var-convex}; and \eqref{middle-star} holds because of the following:
  \begin{itemize}
      \item for all $h \in L_t(\^T)$ with $t < k - \ell$, all the neighboring edges of $h$ are in $\^T \setminus R$, i.e. $\partial h \subseteq \^T \setminus R$, by conditional independence, $\mu_{\^T \setminus R}[\Var[h]{f'}] =  \mu[\Var[h]{\mu_R[f]}]$;
      \item for all $h \in L_{k-\ell}(\^T)$, $\mu_{\^T \setminus R}[\Var[h]{f'}] = \mu\left[\Var[\^T^h]{\mu_{R}[f]}\right]$ because for any $\sigma \in \Omega(\mu)$, the value of $\mu_R[f](\sigma)$ is independent from $\sigma_{\^T_h \setminus h}$. 
  \end{itemize}
  Since $\mu\left[\Var[\^T^h]{\mu_{\widetilde{\^T}^h}[f]}\right]=\sum_{\sigma \in \Omega(\mu_{\^T  \setminus \widetilde{\^T}^h }) } \mu_{\^T  \setminus \widetilde{\^T}^h }(\sigma)\Var[\mu^\sigma]{ \mu_{\widetilde{\^T}^h} [f]}$. Fix one $\sigma$. One can view $f$ as a function from $\Omega(\mu^\sigma_{\widetilde{\^T}^h})$ to $\^R$. By applying \Cref{item:root-tensor} of \Cref{thm:induction} to each $\sigma$ individually and rearranging, it holds that
  \begin{align} \label{eq:route-decay}
    \sum_{h \in L_{k-\ell}(\^T)} \mu\left[\Var[\^T^h]{\mu_{\widetilde{\^T}^h}[f]}\right]
    &\leq \sum_{i=0}^\ell \alpha_i \sum_{g \in L_{k-\ell + i}(\^T)} \mu[\Var[g]{f}].
  \end{align}
  Combining \eqref{eq:law-total-var}, \eqref{eq:mu[Var_R[f]]}, \eqref{eq:ind-hypo}, and \eqref{eq:route-decay}, it holds that
  \begin{align*}
    \Var[\mu]{f}
    &\leq \sum_{t=1}^{k-\ell-1} F_{k-\ell}(t) \sum_{h \in L_t(\^T)} \mu[\Var[h]{f}] \\
    &\quad + F_{k-\ell}(k-\ell) \cdot \alpha_0 \sum_{h \in L_{k-\ell}(\^T)} \mu[\Var[h]{f}] \\
    &\quad + \sum_{i=1}^\ell (\alpha_i F_{k-\ell}(k-\ell) + \gamma) \sum_{h\in L_{k-\ell+i}(\^T)} \mu[\Var[h]{f}].
  \end{align*}
  The above inequality immediately gives the recurrence of $F$.

\begin{proof}[Proof of \Cref{claim-bound}]
  Note that $R$ is a set of disjoint trees 
  $\^T^{(j)}$ with depth $\ell$ for $j \in [z]$. 
 We use $\^T^{(j,\star)}$ to denote the tree obtained by combining  $\^T^{(j)}$ with the edge connecting the root of $\^T^{(j)}$ to its parent in $\^T$.
  Conditional on any coloring $\sigma \in \Omega(\mu_{E \setminus R})$, $\mu^\sigma$ is a product distribution over marginal distributions of all $\^T^{(j,\star)}$'s and $H = E \setminus (\cup_{i=1}^z \^T^{(j,\star)} )$. We have
  \begin{align}
  \nonumber
    \mu[\Var[R]{f}] 
    &= \sum_{\sigma \in \Omega(\mu_{E \setminus R}) }\mu_{E \setminus R}(\sigma) \Var[\mu^\sigma]{f} \\
    \label{eq:Var-R-1}
    &\leq   \sum_{\sigma \in \Omega(\mu_{E \setminus R}) }\mu_{E \setminus R}(\sigma) \sum_{j=1}^z \mu^\sigma[ \Var[\^T^{(j,\star)}]{f} ]= \sum_{j=1}^z \mu[ \Var[\^T^{(j,\star)}]{f} ],
  \end{align}
  where the inequality \eqref{eq:Var-R-1} 
  holds because of~\eqref{eq:var-sub-add} and $\mu^\sigma[ \Var[H]{f} ]=0$.
  Note that conditioning on any feasible coloring $\sigma$ of $E \setminus \^T^{(j)}$, $\mu^\sigma_{T^{(j,\star)}}$  is the same as the distribution $\mu^{\star,1}_\ell$ in~\Cref{item:broute-force} of \Cref{thm:induction} by the symmetry of colors. We can open the expectation in $\mu[\Var[\^T^{(j,\star)}]{f}]$ then apply \Cref{item:broute-force} of \Cref{thm:induction} to obtain
  \begin{align} \label{eq:Var-R-2}
    \mu[\Var[\^T^{(j,\star)}]{f}] & = \sum_{\sigma \in \Omega(\mu_{E \setminus \^T^{(j,\star) }})} \mu_{E \setminus \^T^{(j,\star) }}(\sigma)\Var[\mu^\sigma]{f} \notag\\
    &\leq \sum_{\sigma \in \Omega(\mu_{E \setminus \^T^{(j,\star) }})} \mu_{E \setminus \^T^{(j,\star) }}(\sigma)  \gamma \sum_{i=1}^\ell  \sum_{h \in L_i(\^T_\ell^\star)} \mu^\sigma[\Var[h]{f}] \notag\\
    &\leq \gamma \sum_{i=1}^\ell  \sum_{h \in L_i(\^T^{(j,\star)})} \mu[\Var[h]{f}].     
  \end{align}
  By the definition of $\^T^{(j)}$, we know that $h \in L_i(\^T^{(j)})$ means $h \in L_{k-\ell+i}(\^T)$.
  So, \Cref{claim-bound} is proved by combining \eqref{eq:Var-R-1} and \eqref{eq:Var-R-2}.
\end{proof}

\subsection{Upper bound for the recursion $F$}
\label{sec:upper-bound-F}
In this section, we will give an upper bound for the recursion $F_\cdot(\cdot)$ defined in \eqref{eq:F-recursion} and then finish the proof of \Cref{thm:induction} for $\mu_k$.
The proof of \Cref{thm:induction} for $\mu^{\star,1}_k$ follows from the same proof.

   Let $S_k := \set{a \in \^Z \mid a \equiv k (\-{mod}\; \ell)}$ be a set of integer. By resolving the recursion in \eqref{eq:F-recursion}, we claim that $F_k:[k] \to \^R$ has the following upper bound:
   \begin{align} \label{eq:Fk-explicit}
     \forall 1 \leq t \leq k, \quad \widehat{F}_k(t) := \gamma \cdot \tp{ \tp{\sum_{i=0}^{\abs{[t-1] \cap S_k}} \alpha_\ell^{i}} \cdot \alpha_{t - \max([t-1]\cap S_k)} + 1 }\cdot \alpha_0^{\*1[t \in S_k \text{ and } t\neq k]}.
   \end{align}

   \begin{proof}[Proof of \Cref{thm:induction}]
     The upper bound stated in \Cref{thm:induction} follows directly from \eqref{eq:Fk-explicit}.
   \end{proof}
   
   Now, we only left to prove \eqref{eq:Fk-explicit}.
   We will finish the proof by induction.
   For the base case, that $1 \leq t \leq k \leq \ell$, it holds that
   $F_k(t) = \gamma \leq \gamma (1 + \alpha_t) = \widehat{F}_k(t)$.
   For the inductive case where $k \geq \ell$, suppose for any $s < k$ and $t \leq s$, it holds that $F_s(t) \leq \widehat{F}_s(t)$.
   By definition, $S_k = S_{k-\ell}$.
   This implies $F_k(t) = F_{k-\ell}(t) \leq \widehat{F}_{k-\ell}(t) = \widehat{F}_{k}(t)$ for $t \leq k - \ell - 1$ and $F_k(k-\ell) = F_{k-\ell}(k-\ell)\cdot \alpha_0 \leq \widehat{F}_{k-\ell}(k-\ell) \cdot \alpha_0 = \widehat{F}_{k}(k-\ell)$ by definition.
   For $t \geq k - \ell + 1$, let $j := t - (k - \ell)$, we have
   \begin{align*}
       F_k(t) &= \alpha_j F_{k-\ell}(k-\ell) + \gamma 
       \leq \alpha_j \widehat{F}_{k-\ell}(k-\ell) + \gamma \\
       &= \gamma\tp{\alpha_j  \tp{ \tp{\sum_{i=0}^{\abs{[k-\ell-1] \cap S_k}} \alpha_\ell^{i}} \cdot \alpha_\ell + 1 } + 1 } \\
        &= \gamma\tp{\alpha_j  \tp{\sum_{i=0}^{\abs{[k-\ell] \cap S_k}} \alpha_\ell^{i}} + 1 }
        = \widehat{F}_k(t). \qedhere
   \end{align*}
   This finishes the proof of \eqref{eq:Fk-explicit}.

\section{Coupling and canonical paths} \label{sec:canonical-path}
In this section, we give a proof for \Cref{lem:expectation-bound}, which shows that for the $\Delta$-regular complete tree of height $\ell=O(\Delta^2\log^2{\Delta})$ there is a set of couplings $\mathcal{C}$ and a set of canonical paths $\Gamma$ with good congestion on the leaf edges. This lemma then completes the proof of \Cref{lem:verify-condition} as presented in \Cref{sub:congestion-overview}.

In \Cref{sec:const}, we will describe the construction of canonical paths. 
Then, we analyze the congestion for these canonical paths and finish the proof of \Cref{lem:expectation-bound} in \Cref{sec:congestion} by bounding the number of canonical paths that may go through a specific coloring $\gamma$ via leaf edges (see \Cref{lem:congestion}); and the probability of such a coloring $\gamma$ appearing (see \Cref{lem:bad-event-prob}). 
Then, we finish the proofs of  \Cref{lem:congestion,lem:bad-event-prob} in \Cref{sec:congestion,sec:bad-event-prob}, respectively.

\subsection{Construction}\label{sec:const}
Fix $a,b \in [q-d]$ with $a \neq b$, we construct a coupling $\+C_{ab}$ between $\mu^{ra}$ and $\mu^{rb}$.
Consider a random sample $\sigma \sim \Omega(\mu^{ra})$. The root hanging edge $r$ must receive the color $\sigma_r = a$. Starting from $r$, we can find the maximal $(a,b)$-alternating path $e_0,e_1,e_2,\ldots,e_s$, which we denote as $\+E^*$.  Note that $\+E^*=(e_0,e_1,e_2,\ldots,e_s)$ has the following properties: $e_0 = r$; $\sigma(e_i) = a$ for even $i$ and $\sigma(e_j) = b$ for odd $j$. We can ``flip'' the colors $a$ and $b$ in this alternating path $\+E^*$ to obtain a new coloring $\tau$ (by flip we mean interchange the colors $a$ and $b$). We denote this flip operation by $\tau:= \sigma \oslash_r b$.
This notation means, given a coloring $\sigma$, an edge $r$, and a color $b$, we find the maximal $(\sigma(r),b)$-alternating path containing $r$, and flip the colors $(\sigma(r),b)$ in the maximal alternating path to obtain the new coloring $\tau$.
The following observation is easy to verify.
\begin{observation}
For $\sigma \sim \mu^{ra}$, let $\tau=\sigma \oslash_r b$.  Then, 
$\tau\sim \mu^{rb}$ and $\tau$ is uniquely determined by $\sigma$ and $b$. 
\end{observation}

Hence, this construction of $(\sigma,\tau)$ defines a coupling $\+C_{ab}$ of $\mu^{ra}$ with $\mu^{rb}$;
we refer to this coupling $\+C_{ab}$ as the \emph{flip coupling}.  Note the coupling is symmetric, i.e., $\+C_{ab}=\+C_{ba}$.  Let $\+C=\cup_{a,b}\+C_{ab}$. 

Fix an initial coloring $\sigma \in \Omega(\mu^{ra})$ and the final coloring $\tau = \tau(\sigma,b) \in \Omega(\mu^{rb})$ where $\tau=\sigma \oslash_r b$.
 Let $v_0,v_1,\ldots,v_{s+1}$ list the vertices of $\+E^*$ in order from the root downwards, and let $\+E^*=(e_0,e_1,\ldots,e_s)$ be the edges in downwards order.  Therefore, $r=e_0=\{v_0,v_1\}$ and $e_i = \{v_{i},v_{i+1}\}$.

We now detail the construction of the canonical path $\gamma^{\sigma,\tau}$ from $\sigma$ to $\tau$. 
Recall,  $\+E^*=(e_0,e_1,e_2,\ldots,e_s)$ is the $(a,b)$-alternating path in $\sigma$ which we flip to obtain $\tau$ and the path $\gamma^{\sigma,\tau}$ needs to perform this flip by a sequence of single-edge updates (i.e., Glauber transitions).
 
Fix an arbitrary ordering $\+O$ on all colors in $[q]$ such that the colors $a$ and $b$ are the last two colors. 
The following is the high level idea for constructing the canonical path $\gamma^{\sigma,\tau}$ from $\sigma$ to $\tau$, see \Cref{fig:canonical-path} for an illustration. 

{\bf \noindent Stage-I:} For every $e_i = \{v_{i},v_{i+1}\} \in \sigma \oplus \tau$, where $i \in [s]=\{1,2,\ldots,s\}$ and $i$ is odd, we would like to change the color $\sigma(e_i) = b$ to some color $c_i \notin \{a,b\}$. If $e_i$ has multiple such available colors (i.e., do not appear in the neighborhood of $e_i$ in $\sigma$) that do not in $\set{a,b}$ then we let $c_i$ be the first such available color in the ordering $\+O$.  If $e_i$ does not have any available color, excluding its current color $b$ (note that $a$ is not available as $\sigma(e_{i-1})=a$) then we will define a path $\+E_i$ from $v_{i}$ downwards using the ordering $\+O$; this path $\+E_i$ may reach a leaf edge.  The construction of this path $\+E_i$ will be detailed later.  We let $c_i$ be the initial color in $\sigma$ of the first edge in $\+E_i$, and after recoloring $\+E_i$ then color $c_i$ will be available for $e_i$.  The new color for each edge in $\+E_i$ will be the initial color in $\sigma$ of its descendant edge in the path $\+E_i$ (and we will recolor these edges from the leaves upwards).
    
    We recolor all of these paths $\+E_i$ and the odd edges of $\+E^*$ in the following order.  Let $\+E = \left(\bigcup_{\textit{odd }i} e_i\right)\cup\left(\bigcup_i \+E_i\right)$ denote the set of edges we are recoloring in this stage.  We order the edges in $\+E$ so that edges in lower levels (closer to the leaves) are before higher levels, and thus the leaf edges are first.  Moreover, within a level all of the edges in $\cup_i \+E_i$ are before those in $\+E^*$; this ensures that the edges in $\+E\cap\+E^*$ can be recolored to $c_i$. 
    To change the color of $e_i$, this process may  change all colors along a path $(e_i \cup \+E_i)$, where $(e_i \cup \+E_i)$ intersects with $\sigma \oplus \tau$ only at edge $e_i$.

    We describe how to construct the path $\+E_i$ for odd $i \in [s]$. If $e_i=\{v_i,v_{i+1}\}$ has an available color $c_i \notin \{a,b\}$, then let $\+E_i = \emptyset$. Otherwise, let $N_e(v_{i+1})$ denote the set of edges incident to $v_{i+1}$ and let $c_i$ be the first color in $[q] \setminus \sigma(N_e(v_{i+1}))$ according to the ordering $\+O$. 
    We call $c_i$ the first available color for vertex $v_{i+1}$ in $\sigma$.
    Note that the color $c_i \notin \{a,b\}$ because $|[q] \setminus \sigma(N_e(v_{i+1}))| = q - \Delta = 2$ and the color $b$ must in $\sigma(N_e(v_{i+1}))$. 
    The vertex $v_{i+1}$ must have an available color that is not $a$ nor $b$ and that color would be the first one according to $\+O$. We find the incident edge $e=\{v_i,v\}$ to $v_i$ whose color is $c_i$ in $\sigma$. Such edge $e$ must exist, otherwise  $c_i \notin \{a,b\}$ is an available color for $e_i$ and we do not need to construct $\+E_i$. Hence, we defined the first edge $e$ in the path $\+E_i$. All other edges can be constructed by the following recursive procedure. Specifically, suppose $\{u_j,u_{j+1}\}$ is the $j$-th edge in $\+E_i$ and $u_{j+1}$ is a child of $u_j$.
    If  $\{u_j,u_{j+1}\}$ has two available colors, we stop the construction of $\+E_i$. Otherwise, let $c$ be the first available color  for $u_j$ in $\sigma$; we find the edge $e'$ incident to $u_{j+1}$ whose color is $c$ in $\sigma$; and let $e'$ be the $(j+1)$-th edge in $\+E_i$.  
    This process must stop once it touches a leaf edge but it may stop earlier before going down to the leaves.

    We now detail the process of changing colors along the paths $(e_i \cup \+E_i)$ for odd $i \in [s]$.
    For two different paths $(e_i \cup \+E_i)$ and $(e_{i'} \cup \+E_{i'})$, their distance is at least 2 in the line graph. 
    Hence, changing the colors in $(e_i \cup \+E_i)$ has no effect on the process of changing colors in $(e_{i'} \cup \+E_{i'})$. 
    For a specific path  $e_i \cup \+E_i= \{e_i =: h_0, h_1,h_2,\ldots,h_t\}$, we first change the color of $h_t$ to the first available color $c$ for $h_t$ in $\sigma$ with $c \neq \sigma(h_t)$. For each $j$ from $t-1$ down to $0$, after changing the color of $h_{j+1}$, $\sigma(h_{j+1})$ becomes an available color for $h_j$ and we change the color of $h_j$ to $\sigma(h_{j+1})$. 
    
    
    After Stage-I, for any even $j \in \{0\}\cup [s]$, we claim that color $b$ is available for edge $e_j=\{v_j,v_{j+1}\}$ in the current coloring.
    For even $j$, both $e_{j-1}$ and $e_{j+1}$ switched to colors not in $\{a,b\}$. Note that $e_{j}$ may have an incident edge $e$ in $\+E_{j+1}$ such that $e \cap e_j = \set{v_{j+1}}$. We need to show that the edge~$e$ cannot take $b$ in the current coloring. This is because, in the initial coloring $\sigma$, $v_{j+1}$ has incident edges with colors $a$ and $b$. Note that after Stage-I, $e$ either takes the first available color for $v_{j+1}$ in $\sigma$ ($\+E_i$ has at least 2 edges) or an available color of $e$ in $\sigma$ ($e$ is the only edge in $\+E_i$).  In both cases, the new color cannot be $b$ and cannot be $a$. 

    A detailed illustration of \textbf{Stage-I} is given at \Cref{fig:path-step-I}.
    
    {\bf \noindent Stage-II:} For every $e_i = \{v_{i},v_{i+1}\} \in \sigma \oplus \tau$, where $i \in \{0\}\cup[s]$ and $i$ is even, we change the color of $e_i$ from $\sigma(e_i) = a$ to color $b$.

 {\bf \noindent Stage-III:} 
 In this stage we recolor all of $\+E_i$, for all odd $i$, and the odd edges of $\+E^*$ by the following
 process.  First, consider the same set $\+E$ from Stage-I.  Then, order the edges in $\+E$ in reverse order from Stage-I's ordering, and hence the edges closer to the root are earlier in the ordering and edges in $\+E^*$ are before $\+E_i$.  For every edge $e \in \+E$ we recolor it to $\tau(e)$ (proceeding in the order we just defined).  Note that for $e\in\+E^*\cap\+E$ then $\tau(e)=a$ and for
 $e\in\+E_i$ then $\tau(e)=\sigma(e)$ is its original color in $\sigma$.

\begin{figure}[!h]
  \centering
  \colorlet{cR}{magenta}
  \colorlet{cG}{lime}
  \colorlet{cB}{cyan}
  \colorlet{cY}{yellow}
  \colorlet{cD}{darkgray}
  \colorlet{cO}{orange}
  \colorlet{cBG}{black!20}

  \ifarxiv
  \includegraphics{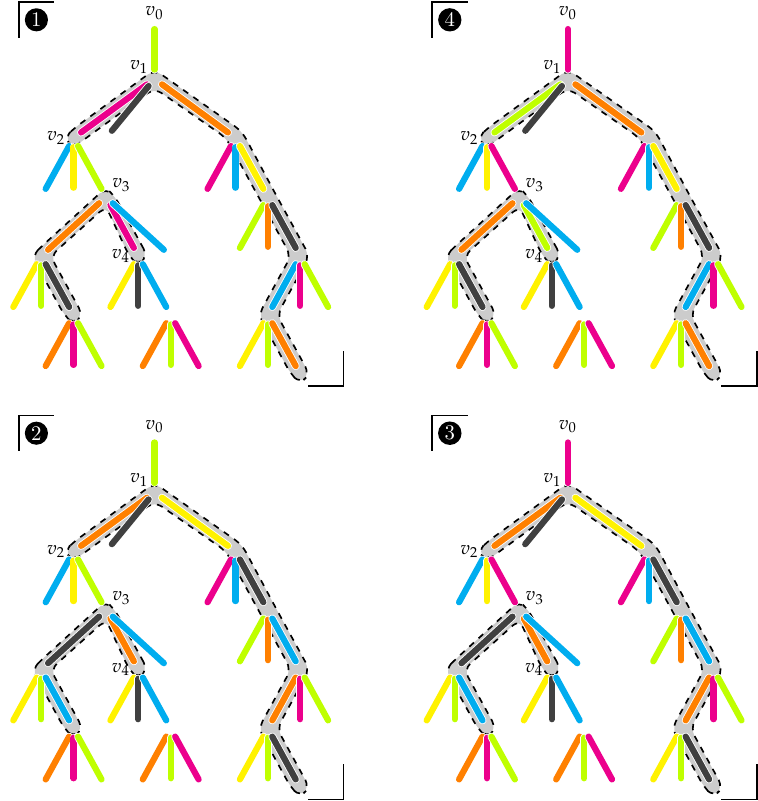}
  \else
  \input{Figures/Fig1.tex} 
  \fi
  
  \newcommand{\colorCirc}[1]{\tikz \node[circle, fill=#1, inner sep=0pt, minimum size = 5pt] {};}
  \newcommand*\circled[1]{\tikz[baseline=(char.base)]{
    \node[fill=black, circle, inner sep=1pt, minimum size = 2pt] (char) {\color{white} #1};
  }}
  \caption[An illustration for the construction of canonical path]{ \label{fig:canonical-path}
    The figure illustrates the canonical path $\gamma^{\sigma,\tau}$.
    For simplicity, we only draw edges whose color will be used in the construction rather than the entire tree.
    The edge-coloring in \circled{1} is $\sigma$ and the coloring in \circled{4} is $\tau$.
    Let $a = \colorCirc{cG}$ and $b = \colorCirc{cR}$, we mark the vertices in $(a,b)$-alternating path $\+E^*$ as $(v_0, v_1, v_2, v_3, v_4)$.
    Then, we fix the order $\+O$ of colors as $\+O=\{ \colorCirc{cO} \prec \colorCirc{cD} \prec \colorCirc{cY}  \prec \colorCirc{cB} \prec \colorCirc{cG} \prec \colorCirc{cR}\}$; note, the colors $\colorCirc{cG}, \colorCirc{cR}$ are the last two colors in this order.
    The odd edges $e_i \in \+E^*$ and their associated paths $\+E_i$ are marked by dashed lines (and grey background).
    The coloring in \circled{2} is obtained after Stage-I, and hence after we recolored the edges in $\+E$ in order from the leaves up.  Then the coloring in \circled{3} is after Stage-II, and finally the coloring in \circled{4} is after Stage-III and is $\tau$.
  }

\end{figure}

\begin{figure}[!h]
  \centering
  \colorlet{cR}{magenta}
  \colorlet{cG}{lime}
  \colorlet{cB}{cyan}
  \colorlet{cY}{yellow}
  \colorlet{cD}{darkgray}
  \colorlet{cO}{orange}
  \colorlet{cBG}{black!20}

  \ifarxiv
  \includegraphics{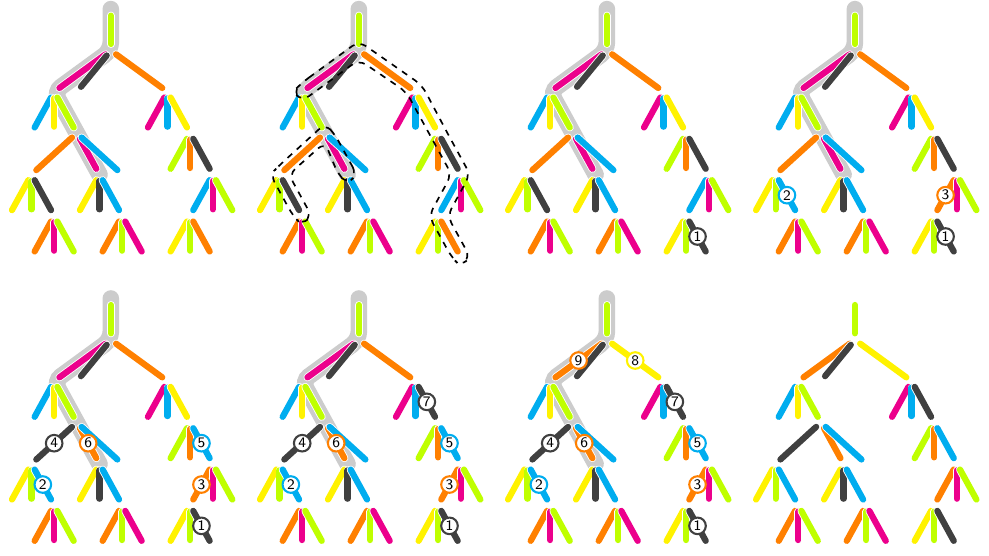}
  \else
  \input{Figures/Fig2.tex} 
  \fi
  
  \newcommand{\colorCirc}[1]{\tikz \node[circle, fill=#1, inner sep=0pt, minimum size = 5pt] {};}
  \caption[An illustration for the step-I]{ \label{fig:path-step-I}
    The figure details \textbf{Stage-I} of the canonical path.
    For convenience, we call the sub-figure located in the $i$-th row and $j$-th column as $\text{fig}(i, j)$.
    The canonical path is as follows.
    First, as in $\text{fig}(1, 1)$, we find the $(\colorCirc{cG}, \colorCirc{cR})$-alternating path; this is marked with a shadowed background.
    Then, for the odd edges on the alternating path (i.e., the \colorCirc{cR}-edges), we want to change their color to some color other than \colorCirc{cG} and \colorCirc{cR}.
    To do that, we run Stage-I according to the priority of colors $\colorCirc{cO} \prec \colorCirc{cD} \prec \colorCirc{cY}  \prec \colorCirc{cB} \prec \colorCirc{cG} \prec \colorCirc{cR}$ (see also \Cref{fig:canonical-path}) and get two lists of edges.
    We mark these lists $\+E$ with dashed circles in $\text{fig}(1,2)$.
    Intuitively, to change the color of the odd edges on the alternating path, we need to change the color of edges in these lists.
    Then, as showed from $\text{fig}(1,3)$ to $\text{fig}(2,3)$, we change the color of the edges in the lists according to the order of their depth.
    For the edges on the same level, to break the tie, we first consider edges from left to right, but with an exception that we keep the edge in the alternating path to be the last one.
    We add number on these edges to denote this order.
    Lastly, when we change the color of some edge $i$, we always change its color to a different color which has the lowest priority among all the available colors.
  }
\end{figure}
\begin{proposition}
If $q = \Delta + 2$, the above process constructs a feasible canonical path $\gamma^{\sigma,\tau}$ from $\sigma$ to $\tau$.
\end{proposition}
\begin{proof}
It is easy to verify the canonical path constructed at Stage-I and Stage-II is a feasible path starting from $\sigma$ to $\gamma$, where $\gamma$ denotes the coloring obtained after Stage-II. 
Stage-III constructs a canonical path from $\gamma$ to $\tau$.
We only need to show that any coloring in the canonical path constructed at Stage-III is a feasible edge coloring.

Consider the following process $\+P$, we run Stage-I starting from the coloring $\tau$ and obtain a canonical path from $\tau$ to a coloring $\gamma'$. 
Specifically, we first find the maximal $(b,a)$-alternating path $\+E^*$ from $r$ in $\tau$. 
We then do the same procedure described in Stage-I except that we flip the roles of $\sigma$ and $\tau$, and also the roles of $a$ and $b$.
\begin{claim}\label{claim:reverse}
The canonical path generated by $\+P$ is the reverse of the canonical path generated by Stage-III. In particular, $\gamma =\gamma'$.
\end{claim}
It is easy to verify  Stage-I generates a feasible canonical path, the above claim implies the canonical path generated by Stage-III is also feasible, which proves the proposition.

We now prove \Cref{claim:reverse}.
We use Stage-I($\tau$) to denote the process $\+P$, where the notation emphasizes that $\+P$ is the Stage-I starting from $\tau$.
Similarly, we use Stage-I($\sigma$) to denote the original process starting from $\sigma$ (described in \Cref{sec:const}).
To prove the claim, we show that two processes Stage-I($\sigma$) and Stage-I($\tau$) find the same set of edges $\+E$, and for any $e \in \+E$, two processes change the color of $e$ to the same color $\gamma(e) = \gamma'(e)$.

We first verify that Stage-I($\tau$) generates the same paths $\+E^*$ and $(\+E_i)_{ i \in [s] \text{ is odd}}$ as Stage-I($\sigma$). 
It is easy to see that the alternating path $\+E^*$ is the same.
If a path $\+E_i$ is empty in Stage-I($\sigma$), then $e_i$ must have an available color in $\sigma$ that is not in $\{a,b\}$, the same color must be available for $e_i$ in $\tau$, and the path $\+E_i$ is also empty in Stage-I($\tau$). If a path $\+E_i$ is not empty in Stage-I, we need to show that $\+E_i$ is also constructed by Stage-I($\tau$). The first edge of $\+E_i$ is determined by the first available color of $v_{i+1}$.
Since $a,b$ are the last two colors in $\+O$ and $q = \Delta + 2$,
in both Stage-I($\tau$) and Stage-I($\tau$),
the first available color is not in $\{a,b\}$.
Then, the two processes find the same first available color and the same first edge in $\+E_i$.
The path $\+E_i$ intersects with $\+E^*$ only at the vertex $v_i$.
Since $\sigma$ and $\tau$ agree everywhere except the alternating path and $\sigma(N_e(v_i)) = \tau(N_e(v_i))$, if the first edge is the same, it is easy to verify that all other edges are the same. 
Hence, the two processes obtain the same set of edges $\+E$.
We only need to show that for any $e \in \+E$, Stage-I($\tau$) changes the color $\tau(e)$ to the color $\gamma(e)$ (namely, $\gamma'(e) = \gamma(e)$), which proves the claim.

For each odd $i$, $e_i$ intersects with the first edge in $\+E_i$ at the vertex $v_i$. We can put $e_i$ before the first edge of $\+E_i$ to obtain a path $\+E_i^+$. Note that all $\+E_i^+$ are not adjacent to each other. It suffices to analyze one path $\+E_i^+$. Say $\+E_i^+$ has edges $e_1, h_1,h_2,\ldots,h_t$. Since $\sigma$ and $\tau$ differ only at the alternating path, for all $h_k$ with $k \geq 2$, $h_k$ is not adjacent to $\+E^*$, it is easy to see Stage-I($\tau$) and Stage-I($\sigma$) change $h_k$ to the same color, which implies $\gamma'(h_k)=\gamma(h_k)$. We will take care of $h_1$ and $e_i$ in the following two paragraphs.

Consider the edge $h_1$. Suppose $\+E_i$ only has edge $h_1$. In Stage-I($\sigma$), the color of $h_1$ is changed to the first available color $c$ of $h_1$ in $\sigma$ such that $c \neq \sigma(h_1)$. Hence $\gamma(h_1) = c$. In Stage-I($\tau$), the color of $h_1$ is changed to the first available color $c'$ of $h_1$ in $\tau$ such that $c' \neq \tau(h_1)$. 
Since $h_1$ contains the vertex $v_i$, it holds that $\sigma(N_e(v_i)) = \tau(N_e(v_i))$ and $\{a,b\} \subseteq \sigma(N_e(v_i))$.
Since $\sigma$ and $\tau$ take the same color for edges not in $\+E^*$,
One can verify that $c = c'$.
Suppose $\+E_i$ has at least 2 edges. In Stage-I($\sigma$), the color of $h_1$ is changed to the first available color $c$ of $v_i$ in $\sigma$.  In Stage-I($\tau$), the color of $h_1$ is changed to the first available color $c'$ of $v_i$ in $\tau$. Again, by $\sigma(N_e(v_i)) = \tau(N_e(v_i))$, we know $c = c'$.

Consider the edge $e_i$. If $e_i$ has an available color $c_i \notin \{a,b\}$ in $\sigma$, it must have the same available color in $\tau$. In this case, $e_i$ is changed to the same color in Stage-I($\sigma$) and Stage-I($\tau$) . If $e_i$ does not have such an available color in $\sigma$, then it does not have it in $\tau$. In this case, in Stage-I($\sigma$), $e_i$ is changed to the first available color $c$ of $v_{i+1}$ in $\sigma$.
In Stage-I($\tau$), $e_i$ is changed to the first available color $c'$ of $v_{i+1}$ in $\tau$.
Since, $a$ and $b$ are last two colors in the ordering $\+O$,  $c \notin \{a,b\}$ and $c' \notin \{a,b\}$. One can verify $c =c'$.  

Combining all the above analysis together proves the claim, which proves the proposition.
\end{proof}

\subsection{Analysis of Congestion (Proof of \Cref{lem:expectation-bound})} \label{sec:ana-con}
In this section, we use the coupling and canonical path constructed above to prove \Cref{lem:expectation-bound}.
We will also specify the parameter $\ell$ for the lemma. We first point out that the bound for $\xi_j$ with $j < \ell$ is trivial. By the definition of congestion,
\begin{align*}
    \xi_j \leq \sum_{(\gamma,\gamma'): \gamma \oplus \gamma' \in L_j} \frac{1}{\mu(\gamma)\frac{1}{q}} \leq q^{d^{O(\ell)}}
\end{align*}
because the total number of list colors is at most $q^{d^{O(\ell)}}$. This bound holds for any choice of $\ell$.
Our main task is to bound the congestion $\xi_{\ell} = \max_{a,b}\xi_{\ell}^{ab} \leq \frac{1}{2(\ell + 1)}$ where
 \begin{align}\label{eq:conab}
    \xi_\ell^{ab} = \sum_{\substack{s=(\gamma,\gamma')\in \Lambda(\mu): \gamma \oplus \gamma' \in L_\ell }} \frac{ \tp{\Pr[(\sigma,\tau) \sim \+C_{ab}]{s \in \gamma^{\sigma,\tau}   } }^2}{\mu(\gamma)Q^{\text{Glauber}}_\mu(\gamma,\gamma')},
\end{align}
where $\Lambda(\mu) \subseteq \Omega(\mu) \times \Omega(\mu)$ is the  set of feasible transitions made by the Glauber dynamics.

\paragraph{Proof Overview:}
We first provide a high-level proof overview before diving into the details.  Note that Stage-II only changes the colors of edges in $\+E^*$ that are in even levels. If we set $\ell$ to be an odd number, then we only change the colors of leaf edges in Stage-I and Stage-III.  
Let us consider a move $(\gamma,\gamma')$ in Stage-I such that $\gamma \oplus \gamma' \in L_\ell$. 
The analysis is similar if it is in Stage-III. 
To bound the congestion, the key step is to recover all possible initial colorings $\sigma$ that can use the move $(\gamma,\gamma')$ in Stage-I of the canonical path $\gamma^{\sigma,\tau}$, where the final coloring $\tau$ is fixed by $\sigma$. 
According to the ordering of edges in $\+E$ defined in Stage-I, we always first change the colors of leaves, which implies for any non-leaf edges $e$, $\sigma(e) = \gamma(e)$. 
Now, we only need to recover the colors of leaf edges in $\sigma$.
Furthermore, given $\gamma$, we can almost recover all the paths $\+E^*$ and $(\+E_i)_{i \in [s] \text{ is odd}}$ constructed in Stage-I. 
The only part we do not know is that if a path recovered by $\gamma$ ends at level $\ell - 1$, then in $\sigma$, this path may end at level $\ell - 1$ or end at level $\ell$. 
The second case may occur because when we move from $\sigma$ to $\gamma$ along the canonical path, we may change the colors of some leaf edges, which reduces the length of the path by 1. 
Hence, these paths can help us find all leaf edges where $\sigma$ and $\gamma$ may disagree.
Suppose that in the paths constructed by $\gamma$, there are $P$ paths ending at level $\ell  - 1$.
%
Then there are $(1 + (\Delta - 1)(q-\Delta))^P < (2\Delta)^P$ possible initial colorings $\sigma$.
This is because for each of these paths, $\sigma$ may have the same path (1 possibility) or have one additional leaf edge ($\Delta - 1$ possibilities) and there are at most $(q-\Delta) = 2$ ways the give the colors to the additional leaf edge.  
To bound the congestion, we hope that the probability in~\eqref{eq:conab} can balance the number of possibilities.
First note that $P$ cannot be too large because the total number of paths $\+E^*,(\+E_i)_{i \in [s] \text{ is odd}}$ is at most $1+\lceil \ell / 2\rceil$.
Also note that if $\gamma$ has $P$ paths touching level $\ell - 1$, then $\sigma$ also has at least $P$ paths touching level $\ell - 1$.
If we sample a random $\sigma\sim \mu^{ra}$, 
then the probability of having $P$ such paths also decays fast with respect to $P$. 
%

In the detailed analysis we need to deal with the summation over all $(\gamma,\gamma')$ in~\eqref{eq:conab}.
We first change the probability space such that instead of considering a random $\sigma$, we consider a random $\gamma$. Given a coloring $\gamma$, we reconstruct all $\+E^*, (\+E_i)_{i \in [s] \text{ is odd}}$ using $\gamma$ so that we recover all possible initial colorings $\sigma$. We also need to find all possible next coloring $\gamma'$ such $(\gamma,\gamma')$ changes the color of a leaf edge $e$.  The choices of $\gamma'$ is also limited because there must exist a path in $\+E^*, (\+E_i)_{i \in [s] \text{ is odd}}$ that ends at $e$ so that the canonical path will change the color of $e$.\footnote{By explicitly specifying the ordering used in Stage-I,  one can even show that given $\gamma$, there is at most one possible $\gamma'$ with $\gamma \oplus \gamma' \in L_{\ell}$ such that the move $(\gamma,\gamma')$ can appear in the canonical path. However, we do not need to use this tight bound in our proof.}
Finally, we show that the $\gamma$ with many possible $\sigma$ and $\gamma'$ appears with low probability and we can bound the congestion.

\paragraph{Detailed Analysis:}
Fix a leaf edge $h \in L_\ell$. Let $\gamma \in \Omega(\mu)$. We use $\+L_h^\gamma$ to denote the set of available colors for edge $h$ given coloring $\gamma({\^T^\star \setminus h})$.
Let $c \in \+L_h^\gamma \setminus \gamma(h)$.
The coloring $\gamma'$ can be written as $\gamma \oslash_h c$, 
which is the coloring obtained from $\gamma$ by replacing the color on $h$ with $c$.
This is because the maximal $(\gamma(h),c)$-alternating path only has one edge $h$.
Let $p^{ra} := 1/\abs{\Omega(\mu^{ra})}$.
In the distribution $\mu^{ra}$, every coloring $x \in \Omega(\mu^{ra})$ has probability $p^{ra}$ as $\mu^{ra}$ is the uniform distribution. Similarly, in $\mu$, each coloring $y \in \Omega(\mu)$ has probability $p := 1/\abs{\Omega(\mu)}$.
We note that $p^{ra}$ and $p$ are two numbers rather than distributions.
Note that $Q^{\text{Glauber}}_\mu(\gamma,\gamma \oslash_hc) = \frac{1}{q-d}$ where $d = \Delta - 1$. Since $\+C_{ab}$ is the flip coupling, $\tau$ is fixed by $\sigma$ and $b$. We simply denote $\gamma^{\sigma,\tau}$ by $\gamma^{\sigma,b}$ to emphasize that the path is fixed once $\sigma$ and $b$ are given. 
We have 
\begin{align*}
  \xi_\ell^{ab} &= \frac{q-d}{p}\sum_{\gamma \in \Omega(\mu)}\sum_{h \in L_\ell} \sum_{c \in \+L^\gamma_h \setminus \gamma(h)} \tp{\sum_{\sigma \in \Omega(\mu^{ra})}p^{ra} \cdot \textbf{1}[(\gamma,\gamma \oslash_h c) \in \gamma^{\sigma,b} ] }^2\\
  &= \frac{(q-d)(p^{ra})^2}{p^2}\sum_{\gamma \in \Omega(h)}p\sum_{h \in L_\ell} \sum_{c \in \+L^\gamma_h \setminus \gamma(h)} \abs{\set{\sigma \in \Omega(\mu^{ra}) \mid (\gamma,\gamma \oslash_h c) \in \gamma^{\sigma,b} }}^2.
\end{align*}
Note that $\frac{p^{ra}}{p} = 1/\mu_r(a) = 2$ and $q = d + 3$. We have
\begin{align}\label{eq:def+E}
   \frac{\xi_\ell^{ab}}{12} =  \E[\gamma \sim \mu]{\sum_{h \in L_\ell} \sum_{c\in \elist{\gamma}{h} \setminus \gamma(h) } \abs{\set{\sigma \in \Omega(\mu^{ra}) \mid (\gamma,\gamma \oslash_h c) \in \gamma^{\sigma,b}}}^2} := \+R^{ab}. 
\end{align}
In the rest of this section, we construct canonical paths to make $\+R^{ab} \leq \frac{1}{24(\ell+1)}$ for all $a,b$, which implies $\xi_\ell \leq \frac{1}{2(\ell+1)}$.
Comparing to ${\xi_\ell^{ab}}$, $\+R^{ab}$ changes the randomness of expectation from the randomness of $(\sigma,\tau) \sim \+C_{ab}$ to the randomness of $\gamma \sim \mu$, which will make the analysis easier. We analyze $\+R^{ab}$ in the rest of this section.

Fix $\gamma \in \Omega(\mu)$. Without loss of generality, we assume the color of $\gamma(r)$ is either $a$ or $b$.
If $\gamma(r) = a$, then let $\+E^*(\gamma)$ denote the longest $(a,b)$-alternating path starting from $r$.
If $\gamma(r) = b$, then let $\+E^*(\gamma)$ denote the longest $b$-$a$ alternating path starting from $r$.
We list all vertices along $\+E^*(\gamma)$ as $u_0,u_1,\ldots,u_{S(\gamma)+1}$, where $r=\{u_0,u_1\}$ and $S(\gamma)+1$ is the length of $\+E^*$. In other words, $S(\gamma)$ is the number of edges in $\+E^*(\gamma)$ except the root $r$.  
For every edge $f_i = \{u_i,u_{i+1}\} \in \+E^*(\gamma)$, where $i \in [S(\gamma)]$ is odd, we define $\+E_i(\gamma)$ be the path obtained by applying the same procedure for finding the path $\+E_i$ in Stage-I. 
If $\gamma(r) = a$, it is exactly in the same setting as the procedure described in Stage-I. 
If $\gamma(r) = b$, one can go through the procedure in Stage-I to find the path.

For each $i \in [\ell]$ such that $i$ is odd, we use $P_i(\gamma) \in \{0,1\}$ to indicate the event that $i \leq S(\gamma)$ and the $\+E_i(\gamma) \cap L_{\ell-1}(\^T^\star) \neq \emptyset$. In words, $P_i(\gamma) = 1$ means the alternative path $\+E^*(\gamma)$ has length at least $i$ (so that $e_i$ and $\+E_i(\gamma)$ exist) and the path $\+E_i(\gamma)$ contains some edge in level $\ell - 1$ of the tree $\^T^\star$. Define the sum $P(\gamma)$ as follows
\begin{align}\label{eq:defX}
    P(\gamma) := \sum_{i \in [\ell]: i \text{ is odd}} P_i(\gamma).
\end{align}
Define $Z(\gamma)$ to indicate whether $\+E^*(\gamma)$ touches level $\ell - 1$ in the tree $\mathbb{T}^{\star}$, formally, 
\begin{align}\label{eq:defZ}
    Z(\gamma) := \*1[S(\gamma) \geq \ell - 1].
\end{align}

\begin{lemma} \label{lem:congestion}
Let $\Delta \geq 2$ and $q = \Delta +2$ be two integers. Fix $a$ and $b$. If $\ell \geq 1$ is odd, then for any $\gamma \in \Omega(\mu)$,
 \begin{align*}
    \sum_{h \in L_\ell} \sum_{c\in \elist{\gamma}{h} \setminus \gamma(h)} \abs{\set{\sigma \in \Omega(\mu^{ra}) \mid (\gamma,\gamma \oslash_h c) \in \gamma^{\sigma,b}  }}^2 \leq  2(P(\gamma) + Z(\gamma)) \cdot (2\Delta)^{2(P(\gamma) + Z(\gamma))}.
  \end{align*}
\end{lemma}

The proof of \Cref{lem:congestion} is deferred to \Cref{sec:congestion}. 
We assume $\ell$ is odd in \Cref{lem:congestion} because of the following two reasons. (1) To prove \Cref{lem:expectation-bound}, we only need to show such $\ell$ exists. The existence of an odd $\ell$ is sufficient for us. (2) We can avoid discussing some corner cases when proving \Cref{lem:congestion} if $\ell$ is odd.
Hence, we use this assumption only for the simplicity of the proof. It is not difficult to generalize our proof in \Cref{sec:congestion} to make it work for all $\ell$.

\begin{lemma} \label{lem:bad-event-prob}
  Let $\Delta \geq 2$ and $q = \Delta +2$ be two integers.
  Then, for any $0 \leq s \leq \ell$ and $0 \leq x \leq \lceil s/2 \rceil$,
  \begin{align*}
    \Pr[\gamma \sim \mu]{S(\gamma) = s \land P(\gamma) = x} &\leq \tp{1 - \frac{1}{\Delta}}^{s} \binom{\lceil s/2 \rceil}{x} \prod_{i=1}^x \tp{1 - \frac{2}{\Delta}}^{\ell + 2i - s- 3}.
  \end{align*}
\end{lemma}
The proof of \Cref{lem:bad-event-prob} is deferred to \Cref{sec:bad-event-prob}.
Intuitively, there is an exponential decay of rate $1 - 1/\Delta$ which prevents the length of $\+E^*(\gamma)$ and $P(\gamma)$ from being too large.
Given \Cref{lem:congestion,lem:bad-event-prob}, we are ready to bound $\+R^{ab}$ in~\eqref{eq:def+E}.
Now, we always assume $\ell$ is an odd number so that we can use \Cref{lem:congestion}.
By \Cref{lem:congestion}, it holds that
\begin{align*}
  \+R^{ab} &\leq \E[\gamma \sim \mu]{2(P(\gamma) + Z(\gamma)) \cdot (2\Delta)^{2(P(\gamma) + Z(\gamma))}}\\
  &=  \sum_{s=0}^{\ell} \sum_{x=0}^{\lceil s/2 \rceil} \Pr[\gamma \sim \mu]{S(\gamma) = s \land P(\gamma) = x} \cdot 2(x + \*1[s \geq \ell - 1])(2\Delta)^{2(x + \*1[s \geq \ell - 1])} :=  \sum_{s=0}^{\ell} \sum_{x=0}^{\lceil s/2 \rceil} F(s,x).
\end{align*}
The equation holds because it enumerates all possible values for $S(\gamma)$ and $P(\gamma)$. Note that $\binom{\lceil s/2 \rceil}{x} \leq \ell^x$. We can simplify the bound in \Cref{lem:bad-event-prob} as 
\begin{align}\label{eq-simple}
 \Pr[\rho \sim \mu]{S(\gamma) = s \land P(\gamma) = x} \leq    \exp\tp{-\frac{1}{\Delta}x^2 + \tp{\ln \ell - \frac{\ell - s - 2}{\Delta}} x - \frac{s}{\Delta}}.
\end{align}
Suppose $s \geq \ell-1$. Using \Cref{lem:bad-event-prob}, we have 
\begin{align*}
    F(s,x) &\leq 2(x+1)(2\Delta)^{2(x+1)} \exp\tp{-\frac{1}{\Delta}x^2+\tp{\ln \ell + \frac{2}{\Delta}}x - \frac{\ell - 1}{\Delta}}\\
    &\leq 2(\ell+1)\exp\tp{-\frac{1}{\Delta}x^2+\tp{\ln \ell + \frac{2}{\Delta} + 2 \ln(2\Delta)}x  + 2\ln(2\Delta) - \frac{\ell - 1}{\Delta}}
\end{align*}
Note that $g(x) = -\frac{1}{\Delta}x^2+\tp{\ln \ell + \frac{2}{\Delta} + 2 \ln(2\Delta)}x  + 2\ln(2\Delta) = O(\Delta(\ln\ell + \ln \Delta)^2)$. We can take $\ell = O(\Delta^2 \log^2 \Delta)$ to make $F(s,x) \leq \frac{1}{100(\ell+1)^3}$.
 
Suppose $s < \ell - 1$, by rearranging terms in~\eqref{eq-simple}, we have
\begin{align*}
   F(s,x) &\leq 2x(2\Delta)^{2x}  \exp\tp{-\frac{1}{\Delta}x^2 + \tp{\ln \ell - \frac{\ell - 2}{\Delta}} x + \frac{s}{\Delta}(x-1)}.   
\end{align*}
If $x \leq 1$, then $F(s,x) \leq F(0,1) = 8\Delta^2 \exp(-\frac{\ell - 3}{\Delta} + \ln{\ell}) \leq \frac{1}{100(\ell+1)^3}$ for large  $\ell = O(\Delta^2 \ln^2 \Delta)$. If $x > 1$, then $F(s,x) \leq F(\ell, x) = 2x(2\Delta)^{2x} \exp\tp{-\frac{1}{\Delta}x^2+\tp{\ln \ell + \frac{2}{\Delta}}x - \frac{\ell}{\Delta}}$. Then, the upper bound $F(s,x) \leq \frac{1}{100(\ell +1)^3}$ can be obtained from a similar analysis for the case $s \geq \ell - 1$. Hence, there exists an odd $\ell = O(\Delta^2 \log^2 \Delta)$ such that 
\begin{align*}
  \+R^{ab} \leq   \sum_{s=0}^{\ell} \sum_{x=0}^{\lceil s/2 \rceil} F(s,x) \leq  \sum_{s=0}^{\ell} \sum_{x=0}^{\lceil s/2 \rceil}\frac{1}{100(\ell + 1)^3} \leq \frac{1}{100(\ell+1)}.
\end{align*}
Finally, by~\eqref{eq:def+E}, we have $\xi_\ell^{ab} \leq 12\+R^{ab} < \frac{1}{2(\ell + 1)}$ for all $a,b$.

\subsection{Recovering the initial and final colorings: Proof of \Cref{lem:congestion}}
\label{sec:congestion}
For any initial $\sigma$ and final $\tau$, let $\gamma^{\sigma,\tau}$ be the canonical path. 
Since we assume $\ell$ is an odd number, the canonical path cannot change the color of a leaf edge at Stage-II. If $(\gamma, \gamma')$ changes the color of the leaf edge, the move $(\gamma, \gamma')$ is either in Stage-I or Stage-III. 
We can distinguish Stage-I or Stage-III by looking at the color of $\gamma(r) = \gamma'(r)$. 
If the color of the root hanging edge $r$ is $a$, then the move is in Stage-I. Otherwise, the color of $r$ is $b$ and the move is in Stage-III. 
We first assume $(\gamma, \gamma')$ is in Stage-I. We show how to recover the initial state $\sigma$. Note that since $a$ and $b$ are fixed. Once we recover $\sigma$, the $\tau$ is uniquely fixed.
The other case is $(\gamma,\gamma')$ is in Stage-III. We will discuss it at the end of the proof.

Let $\gamma \in \Omega(\mu)$ with $\gamma(r) = a$.
Let $h \in L_\ell(\^T^\star)$ be a leaf node. Let $c$ be an available color for $h$ in $\gamma$.
Given a move $(\gamma, \gamma \oslash_h c)$, we want to recover how many possible pairs $\sigma$ such that the canonical path $\gamma^{\sigma,\tau}$ uses the move, where $\tau = \sigma \oslash_r b$ is fixed.
Recall $\+E^*(\gamma)$ is the longest $(a,b)$-alternating path in $\gamma$, with vertex sequence $u_0,u_1,\ldots,u_{S(\gamma)+1}$. 
Recall $\+E_i(\gamma)$ is defined in~\Cref{sec:ana-con}.

\begin{lemma}\label{lem:recover-1}
Fix a coloring $\gamma \in \Omega(\mu)$ with $\sigma(r) = a$.
For any $h \in L_\ell(\^T^\star)$ and any $c\in \elist{\gamma}{h} \setminus \gamma(h)$, the move $(\gamma, \gamma \oslash_h c)$ can appear in some canonical path $\gamma^{\sigma,\tau}$ only if one of the following two conditions holds
\begin{itemize}
    \item $h \in \+E^*(\gamma)$;
    \item there exists an odd $i \in [S(\gamma)]$ with $i \neq \ell$ such that $h \in \+E_i(\gamma)$.
\end{itemize}
\end{lemma}
\begin{proof}
Suppose the move $(\gamma, \gamma \oslash_h c)$ appears in some canonical path $\Gamma^{\sigma,\tau}$. 
It must appear in Stage-I.
In Stage-I, the algorithm uses a specific ordering to change the colors of edges. It holds that $\gamma(h) = \sigma(h)$ and for all non-leaf edge $e$, $\gamma(e) = \sigma(e)$.
Let $\+E^*(\sigma) = \sigma \oplus \tau$ be the $(a,b)$-alternating path for the initial coloring $\sigma$.
Recall $\+E^*(\sigma)$ has vertices $v_0,v_1,\ldots,v_{S(\sigma)+1}$.
Let $\+E_i(\sigma)$ denote the path $\+E_i$ constructed in Stage-I.
Since the move is used in the canonical path, there are only two cases.
\begin{itemize}
\item The leaf edge $h$ is the last edge in $\+E^*(\sigma)$. In this case $\+E^*(\sigma) = \+E^*(\gamma)$ and $h \in \+E^*(\gamma)$.
\item The leaf edge $h$ belongs to $\+E_j(\sigma)$, where  $j \in [S(\sigma)]$ is odd and $j \leq \ell - 2$. In this case $\+E_j(\gamma) = \+E_j(\sigma)$, because no edges in $\+E_j(\sigma)$ changed its color in the path from $\sigma$ to $\gamma$. Hence, $h \in \+E_j(\gamma)$.
\end{itemize}
In the second case, we assume $j \leq \ell - 2$. This is because both $j$ and $\ell$ are odd numbers. Also note that $j \neq \ell$, because any leaf edge must have an available color $c \notin \{a,b\}$ and $\+E_{\ell}(\sigma) = \emptyset$.
\end{proof}

Fix a coloring $\gamma \in \Omega(\mu)$ with $\sigma(r) = a$.
Fix $h \in L_\ell(\^T^\star)$ and $c\in \elist{\gamma}{h} \setminus \gamma(h)$ satisfying the condition in \Cref{lem:recover-1}.
For any odd $i \in [S(\gamma)]$ with $i \leq \ell - 2$, we say the path $\+E_i(\gamma)$ is \emph{suspicious} if the last edge is in the level $\ell - 1$ of $\^T$, formally, $\+E_i(\gamma) \cap L_{\ell-1}(\^T^\star) \neq \emptyset$ and $\+E_i(\gamma) \cap L_{\ell}(\^T^\star) = \emptyset$.
We use $i_1,i_2,\ldots,i_\beta$ to list the indices of all suspicious paths.
For any $j \in [\beta]$, we use $H(i_j)$ to denote the set of leaf edges adjacent to the last edge in path $\+E_{i_j}(\gamma)$. 
Finally, we say the alternating path $\+E^*(\gamma)$ is \emph{suspicious} if $S(\gamma) =\ell - 1$.
If $\+E^*(\gamma)$ is suspicious, let $H_0$ be the set of leaf edges adjacent to $\+E^*(\gamma)$; otherwise, let $H_0 = \emptyset$.
By the definition of paths $\+E_i(\sigma)$ and $\+E^*(\sigma)$, it is easy to verify that 
\begin{itemize}
\item $H_0,H(i_1),\ldots,H(i_\beta)$ are disjoint.
\item Since the degree of the tree is $\Delta$, we have $|H_0| < \Delta$ and $|H(i_j)| < \Delta$ for all $j \in [\beta]$.
\end{itemize}

Next, we show the following results.
\begin{lemma}\label{lem:recover-2}
The move $(\gamma, \gamma \oslash_h c)$ appears in a canonical path $\gamma^{\sigma,\tau}$ only if
\begin{itemize}
\item $\sigma \oplus \gamma \subseteq H_0 \cup H(i_1) \cup \ldots \cup H(i_\beta)$;
\item For any $H \in \{H_0,H(i_1),\ldots,H(i_\beta)\}$, $|(\sigma \oplus \gamma) \cap H| \leq 1$.
\end{itemize}
\end{lemma}
\begin{proof}
Consider a canonical path $\Gamma^{\sigma,\tau}$ where $\sigma$ and $\tau$ come from the flipping coupling. Since we assume $\gamma(r) =a$, the move $(\gamma, \gamma \oslash_h c)$ can only occur at the Stage-I. Consider all the paths $\+E_i(\sigma)$, where $i \in [S(\sigma)]$ is odd and $i \neq \ell$.
If $\+E_i(\sigma)$ contains a leaf edge, then we denote this leaf edge by $h(i)$.
If $S(\sigma) = \ell$, then we use $h(\ell)$ to denote the last leaf edge in $\+E^*(\sigma)$. 
Let $j_1,j_2,\ldots,j_\alpha$ be all odd indices in $[S(\sigma)]$ such that the leaf edge $h(j_k)$ is defined.
Then, all these $\alpha$ leaf edges are the first $\alpha$ edges in $\+E$ according to the ordering in Stage-I. Assume the first $\alpha$ steps in Stage-I are to change the colors of edges $h(j_1),h(j_2),\ldots,h(j_\alpha)$ so that we obtain a path $(\sigma=\gamma_0, \gamma_1,\ldots,\gamma_\alpha)$. 

Fix $1 \leq k \leq \alpha$. Recall $S(\gamma_k)$ denotes the length of the maximal $(a,b)$-alternating path $\+E^*(\gamma_k)$ from $r$ in $\gamma_k$. For any odd $i \in [S(\gamma_k)]$, again we can define the path $\+E_i(\gamma_k)$ (applying Stage-I from $\gamma_k$ to construct paths $\+E_i(\gamma_k)$). 
We claim the following properties.
\begin{itemize}
    \item for any $ 1\leq k' \leq k-1$ with $j_{k'} \neq \ell$, it holds that $\+E_{j_{k'}}(\sigma) =  (\+E_{j_{k'}}(\gamma_k),h(j_{k'}))$.
    In words, $\+E_{j_{k'}}(\gamma_k)$ is a suspicious path obtained from $\+E_{j_{k'}}(\sigma)$ by removing the last edge (the leaf edge).

    \item if $j_{k^*} = \ell$ for some $1\leq k^* \leq k - 1$, then $S(\gamma_k) = \ell - 1$.
\end{itemize}
The first property holds because once we change the color of $h(j_{k'})$, $\+E_{j_k}(\sigma)$ becomes the suspicious path $\+E_{j_k'}(\sigma) = \+E_{j_k}(\sigma) \setminus h(j_{k'})$ and $\+E_{j_k}(\sigma) = \+E_{j_k'}(\sigma)$.  
The second property holds because once we change the color of the leaf edge in the alternating path $\+E^*(\sigma)$, then $\+E^*(\gamma_{j_{k^*}}) = \+E^*(\sigma) \setminus h(j_{k^*})$, which means the length of the alternating path decreases by 1 and we have $S(\gamma_k) = S(\gamma_{k^*}) = \ell - 1$.
The above two properties says the first $k-1$ moves are ``recorded'' in $\gamma_k$ by suspicious paths and $S(\gamma_k)$. Formally, let $R$ be the set of all indices of suspicious paths in $\gamma_k$. It holds that  $\{ j_1,j_2,\ldots,j_{k-1} \} \subseteq R \cup \set{\ell}$ and $\ell \in \{ j_1,j_2,\ldots,j_{k-1} \}$ only if $S(\gamma_k)=\ell - 1$.

Now given a move  $(\gamma, \gamma \oslash_h c)$, we show how to recover all possible starting colorings $\sigma$.
The $\sigma$ must have the same colors as $\gamma$ for all edges at level $\leq \ell - 1$. 
We list all the indices $i_1,i_2,\ldots,i_\beta$ of suspicious paths for $\gamma$. Let $H(i_k)$ and $H_0$ be the sets in the lemma.
By above analysis, for any $k \in [\beta]$, in $\sigma$, either $\sigma(H(i_k)) = \gamma(H(i_k))$ (the case $\+E_{i_k}(\sigma) = \+E_{i_k}(\gamma)$) or there exists only one edge $h(i_k) \in H(i_k)$ such that $\sigma(H(i_k))$ differs from $\gamma(H(i_k))$ at the color of $h(i_k)$ (the case $\+E_{i_k}(\sigma) = (\+E_{i_k}(\gamma),h(i_k))$). 
If $S(\gamma) = \ell - 1$, then $\sigma$ and $\gamma$ can differ at $\leq 1$ edge in $H_0$.
For all other leaf edges not in $H_0 \cup H(i_1) \cup \ldots \cup H(i_\beta)$, $\sigma$ and $\gamma$ must take the same color.
This proves the lemma.
\end{proof}

Now, we are ready to prove \Cref{lem:congestion}.  Fix a $\gamma$ with $\gamma(r) = a$. By \Cref{lem:recover-1}, the number of possible $h \in L_\ell(\^T^\star)$ are at most $P(\gamma) + Z(\gamma)$. For each $h$, the number of $ c\in \elist{\gamma}{h} \setminus \gamma(h)$ is at most $|\elist{\gamma}{h} \setminus \gamma(h)| = q - \Delta = 2$.\footnote{We remark that one can further reduce the number of $c's$ to $1$ because the algorithm always uses the first available color in the ordering $\+O$. But the upper bound $2$ is enough for us.} Next, given $(\gamma,\gamma\oslash_hc)$, we need to recover the number of possible $\sigma$'s. By \Cref{lem:recover-2}, for each $H \in \{H_0,H(i_1),\ldots,H(i_\beta)\}$, either $\sigma(H) = \gamma(H)$ or there exists $h \in H$, such that $\gamma$ and $\sigma$ differ only at $h$, and to get a proper coloring, the number of possible colors for $\sigma(h)$ is $q - \Delta = 2$. 
Recall $|H| \leq \Delta - 1$ for all $H \in \{H_0,H(i_1),\ldots,H(i_\beta)\}$.
Hence, the total number of possibilities is $1 + 2(\Delta - 1) < 2\Delta$.
Note that $H_0 \neq \emptyset$ only if $Z(\gamma) \geq 1$. 
\begin{align*}
       \sum_{h \in L_\ell} \sum_{c\in \elist{\gamma}{h} \setminus \gamma(h)} \abs{\set{\sigma \in \Omega(\mu^{ra}) \mid (\gamma,\gamma \oslash_hc) \in \gamma^{\sigma,b}  }}^2 \leq  2(P(\gamma) + Z(\gamma)) \cdot (2\Delta)^{2(X(\gamma) + Z(\gamma))}.
\end{align*}

Finally, consider the case  $\gamma(r) = b$, and thus $(\gamma, \gamma')$ is in Stage-III. 
By \Cref{claim:reverse}, Stage-III is the reverse of the Stage-I starting from $\tau$.
Hence, Stage-III changes the colors of the leaves at the end of the process.
Instead of recovering the initial coloring $\sigma$, we recover the final coloring $\tau$ in this case.
The proof in this case follows by symmetry.



\subsection{Bounding the probability of reaching the leaves: Proof of \Cref{lem:bad-event-prob}}\label{sec:bad-event-prob}

Now, we prove \Cref{lem:bad-event-prob}. Suppose $S(\gamma) = s$ and $P(\gamma) = x$. We can fix the $x$ odd indices in $[s]$ such that $P_i(\gamma) = 1$. There all at most $\binom{\lceil s/2 \rceil}{x}$ ways to choose these indices. Fix a choice $j_1 > j_2 > \ldots > j_{x}$. We bound the probability of $S(\gamma) = s$ and $P_{j_i}(\gamma) = 1$ for all $i \in [x]$. Consider the following process, we first sample a random $\gamma \in \mu$ and then reveal the colors on edges one by one.
We first reveal $\gamma(r)$. Without loss of generality, we assume $\gamma(r) = a$. The case $\gamma(r) = b$ follows by symmetry. We then reveal colors of all children edges of $r$. If one child edge $e$ takes the color $b$, which happens with probability $\frac{\Delta - 1}{q-1}$, we reveal the colors of all children edges of $e$ and look for the color $a$. If $S(\gamma) = s$, we must repeat this process for $s$ steps and reveal colors on a subset $\Lambda$. It gives 
\begin{align*}
    \Pr[]{S(\gamma) = s} \leq \tp{\frac{\Delta-1}{q-1}}^s \leq \tp{1-\frac{q-\Delta}{q-1}}^s \leq \tp{1-\frac{1}{\Delta}}^s.
\end{align*}

For any $i \in [x]$, 
we simulate the process in Stage-I for finding the path $\+E_i$. Given the coloring in $\Lambda$, we either know $\+E_{j_i} = \emptyset$ or we know the first edge $(u_1 = v_{j_i},u_2)$ in $\+E_{j_i}$. This is because the coloring on $\Lambda$ gives the colors of all edges adjacent to the $(a,b)$-alternating path. If $\+E_{j_i} = \emptyset$, then the probability of $P_{j_i}(\gamma) = 1$ is $0$. Assume the worst case that the first edge $(u_1,u_2)$ is fixed. At this moment, we already revealed the colors of all edges incident to $u_1$. Since $q = \Delta +2$, there are 2 missing colors, say colors $c$ and $c'$. We will keep constructing the path $\+E_{j_i}$ only if $(u_1,u_2)$ has less than 2 available colors, which happens only if both $c$ and $c'$ appear in child-edges of $u_2$.
We reveal in coloring of child-edges of $u_2$, conditional on all information revealed so far,\footnote{Let $\Lambda$ denote the set of edges revealed so far. The $\Lambda$ is a random subset. However, the random  $\Lambda$ is fully determined by $\gamma(\Lambda)$. Hence, given $\Lambda$ and $\gamma(\Lambda)$, all other variables in $\overline{\Lambda}$ follows the distribution of $\mu$ conditional on $\gamma(\Lambda)$. } this happens with probability $\binom{q-3}{\Delta - 1}/\binom{q-1}{\Delta - 1} =\frac{\Delta - 2}{\Delta}$. This event needs should happen for $\ell - 1 -j_1$ times if $P_{j_i}(\gamma) = 1$. The probability is at most $(1-\frac{2}{\Delta})^{\ell - 1 - j_1}$. 
Finally, we can use this argument for all paths $\+E_{j_i}$ for $i \in [x]$, and all the edges we considered above are disjoint. This gives
\begin{align*}
  \Pr[]{S(\gamma = s) \land \forall i \in [x], P_{j_i}(\gamma)  = 1 } &\leq \tp{1-\frac{1}{\Delta}}^s \prod_{i =1}^x\tp{1 - \frac{2}{\Delta}}^{\ell - 1 - j_i}\\
(\text{by $j_i \geq s - 2(i-1)$})\quad  &\leq \tp{1-\frac{1}{\Delta}}^s \prod_{i =1}^x\tp{1 - \frac{2}{\Delta}}^{\ell + 2i -s - 3}.
\end{align*}
The above bound relies on the fact the graph is a tree, 
so that we can use conditional independence property to take the product of all probabilities.
The lemma follows from a union bound over all possible $j_1,j_2,\ldots,j_x$.

\section{Coupling and canonical paths for the edge dynamics} \label{sec:edge-dynamics-cond}
In this section, we give a proof overview for \Cref{lem:edge-dynamics-cond}.
The proof will follow the same high-level plan described in \Cref{sub:congestion-overview} and \Cref{sec:canonical-path}.
We assume the readers are already familiar with those sections.
For convenience, we denote $\^T^\star_\ell, \mu^\star_\ell$ as $\^T$ and $\mu$, respectively.
Also, we will use $\widetilde{\^T}$ to denote the tree $\^T - r$.

Recall that the collection of available update is defined in $\+B$ in \eqref{eq:def-edge-block}. 
For convenient, we will use
\begin{align} 
   \label{eq:def-collection-A}
  \+A &:= \set{\set{e,r} \mid e \in L_1(\^T^\star_\ell)} \\
   \label{eq:def-collection-B}
  \+B &= \set{\set{e} \mid e\in \^T^\star_\ell} \cup \+A.
\end{align}

Our goal is to establish approximate root-factorization of variance as defined in \eqref{eqn:root-block-factorization}.
In particular, we need to verify it with parameter $\*\alpha$ and $\beta$ as in \Cref{lem:edge-dynamics-cond}.
Recall this states that for any function $f:\Omega(\mu) \to \^R$, the following holds:
\begin{align}\label{eq:root-factorize}
  \Var[\mu]{\mu_{\widetilde{\^T}}[f]} 
  &\leq \tp{\sum_{i=0}^\ell \alpha_i \sum_{h \in L_i} \mu[\Var[h]{f}]} + \beta \sum_{B \in \+A} \mu[\Var[B]{f}].
\end{align}

To allow the update on multiple edges, we generalize some definitions in \Cref{sub:congestion-overview}.
\begin{definition}[Canonical paths w.r.t. a coupling set]
    Given a coupling set $\+C$, for any pair $(\sigma, \tau) \in \Omega(\+C)$, a canonical path $\gamma^{\sigma, \tau}$ is a simple path from $\sigma$ to $\tau$ on the graph $(\Omega(\mu), \Lambda(\mu))$, where an edge $(\gamma, \gamma') \in \Lambda(\mu)$ if and only if $\gamma \oplus \gamma' \in \+B$.

    Let $\Gamma = \{\gamma^{\sigma,\tau} \mid (\sigma,\tau) \in \Omega(\+C) \}$ denote the set of canonical paths w.r.t. the coupling set $\+C$.
\end{definition}

In the following, we use $Q^{\+B}_\mu(\gamma, \gamma')$ to denote the transition rate from $\gamma$ to $\gamma'$ (with $\gamma \oplus \gamma' = B \in \+B$) in the (continuous-time) heat-bath $\+B$-block dynamics for $\mu$.
Formally,
\begin{align*}
    Q^{\+B}_\mu(\gamma, \gamma') &= \mu^{\gamma_{\^T\setminus B}}(\gamma') = \mu^{\gamma_{\^T\setminus B}}_{B}(\gamma_B').
\end{align*}

\begin{definition}[Congestion of canonical paths w.r.t. a coupling set]
Given a coupling set $\+C$ and a set of canonical paths $\Gamma$, define the expected congestion at level $0 \leq t \leq \ell$ by
\begin{align}\label{eq:def-congestion-level}
    \xi_t &= \max_{a,b\in [q-d]:\atop a\neq b} \sum_{s=(\gamma,\gamma')\in\Lambda(\mu):\atop (\gamma\oplus\gamma')\in L_t} \frac{ \tp{\Pr[(\sigma,\tau) \sim \+C_{ab}]{s\in\gamma^{\sigma,\tau}   } }^2}{\mu(\gamma)Q^{\+B}_\mu(\gamma,\gamma')}.
\end{align}
Then, define the expected congestion at $\+A$ by
\begin{align}\label{eq:def-congestion-block}
    \xi_{\+A} &= \max_{a,b\in [q-d]:\atop a\neq b} \sum_{s=(\gamma,\gamma')\in\Lambda(\mu):\atop (\gamma\oplus\gamma')\in \+A} \frac{ \tp{\Pr[(\sigma,\tau) \sim \+C_{ab}]{s\in\gamma^{\sigma,\tau}   } }^2}{\mu(\gamma)Q^{\+B}_\mu(\gamma,\gamma')}.
\end{align}
\end{definition}

Note that in \eqref{eq:def-congestion-level} the summation is over pairs $\gamma,\gamma'$ that only differ at a single edge.
In this case, it holds that $Q^{\+B}_\mu(\gamma, \gamma') = Q^{\text{Glauber}}_\mu(\gamma, \gamma')$ and hence \eqref{eq:def-congestion-level} has the same form as \eqref{eq:def-congestion}.

Similar to \Cref{lem:con-var}, we have the following result.
\begin{lemma}\label{lem:con-var-block}
If there exist a set of couplings $\+C$ and a set $\Gamma$ of canonical paths such that the congestion with respect to coupling is $\xi_t$ for $0 \leq t \leq \ell$ and $\xi_{\+A}$ for $\+A$, then~\eqref{eq:root-factorize} holds with $\*\alpha = 2(\ell+1)\*\xi$ and $\beta = 2\xi_{\+A}$.
\end{lemma}

Now, following the proof of \Cref{lem:con-var}, we have,
\begin{proof}
 We note that by definition
  \begin{align*}
    \Var[\mu]{\mu_{\widetilde{\^T}}[f]}
   &= \frac{1}{2} \sum_{\substack{a, b \in [q-d]: a\neq b}} \mu_r(a) \mu_r(b) \tp{\E[(\sigma,\tau) \sim \+C_{ab}]{ \sum_{j=1}^{m(\sigma,\tau)} f(\gamma^{\sigma,\tau}_j) - f(\gamma^{\sigma,\tau}_{j-1})}}^2.
  \end{align*}    
Fix $a$ and $b$. Denote $ \E[(\sigma,\tau) \sim \+C_{ab}]{ \sum_{j=1}^{m(\sigma,\tau)} f(\gamma^{\sigma,\tau}_j) - f(\gamma^{\sigma,\tau}_{j-1})}$ by $\widehat{E}_{ab}$. 
Then $\widehat{E}_{ab}$ can be rewritten as
\begin{align*}
 \widehat{E}_{ab} = E_{ab} + \sum_{\substack{(\gamma,\gamma'):\\ \gamma \oplus \gamma' \in \+A }}(f(\gamma)-f(\gamma'))\Pr[(\sigma,\tau) \sim \+C_{ab}]{\gamma^{\sigma,\tau} \text{ uses } (\gamma,\gamma')  },
\end{align*}
where $E_{ab}$ is defined in \eqref{eq:Eab}.
By using Cauchy-Schwarz, 
\begin{align*}
\widehat{E}_{ab}^2
\leq 2E_{ab}^2 + 2\tp{\sum_{\substack{(\gamma,\gamma'):\\ \gamma \oplus \gamma' \in \+A }}(f(\gamma)-f(\gamma'))\Pr[(\sigma,\tau) \sim \+C_{ab}]{\gamma^{\sigma,\tau} \text{ uses } (\gamma,\gamma')  }}^2 := 2E_{ab}^2 + 2A^2.
\end{align*}
We note that $E_{ab}$ is already bounded in the proof of \Cref{lem:con-var}, we only need to bound $A^2$.
Note that $A$ and $A_i$ in \eqref{eq:def-Ai} has the same form.
By the argument in the proof of \Cref{lem:con-var}, we have
\begin{align*}
    A^2 &\leq 2 \xi_{\+A} \sum_{B \in \+A} \mu[\Var[B]{f}].
\end{align*}

Combining everything together, we have
\begin{align*}
   \Var[\mu]{\mu_{\widetilde{\^T}}[f]} 
   &\leq 2\sum_{\substack{a, b \in [q-d]: a\neq b}} \mu_r(a) \mu_r(b)  \tp{(\ell+1)\sum_{i=0}^\ell \xi_i\sum_{e \in L_i} \mu[\Var[e]{f}] + \xi_{\+A} \sum_{B \in \+A} \mu[\Var[B]{f}]} \\
   &\leq 2(\ell+1)\sum_{i=0}^\ell \xi_i\sum_{e \in L_i} \mu[\Var[e]{f}] + 2\xi_{\+A} \sum_{B\in \+A} \mu[\Var[B]{f}].
\end{align*}
Hence, the root-factorization of variance in~\eqref{eq:root-factorize} holds with $\*\alpha = 2(\ell+1)\*\xi$ and $\beta = 2\xi_{\+A}$.
\end{proof}

\subsection{Canonical paths and congestion analysis}
With \Cref{lem:con-var-block} in hand, we could bound $\alpha_i$ for $i < \ell$ and $\beta$ by using the crude $q^{\Delta^{O(\ell)}}$ bound.
In order to prove \Cref{lem:edge-dynamics-cond}, we only need to verify $\alpha_\ell \leq 1/2$ for some large $\ell$.
Just like what we did in \Cref{sec:canonical-path}.
Recall that in the setting of \Cref{lem:edge-dynamics-cond}, $q = \Delta + 1 = d + 2$. 
The major difference of this section from \Cref{sec:canonical-path} is that we will use a slightly different construction of canonical path.

Fix $a, b \in [q - d]$ with $a \neq b$, we will use the same coupling $\+C_{a,b}$ between $\mu^{ra}$ and $\mu^{rb}$ as we did in \Cref{sec:canonical-path}.
%
Following the notations in \Cref{sec:canonical-path}, we can find the maximal $(a,b)$-alternating path $\+E^* = (e_0, e_1, e_2, \cdots, e_s)$ of length $s$ in $\sigma$ such that $e_0 = r$; $\sigma(e_i) = a$ for even $i$ and $\sigma(e_j) = b$ for odd $j$.
The destination configuration $\tau$ is then obtained by exchanging the colors $a$ and $b$ in $\+E^*$.

Note that since $q = \Delta + 1$, every edge $e$ has at most $2$ available colors when we fix the colors of its neighbors.
This means when we say we ``flip'' the color of some edge, there is only one choice of its new color.
Hence, in this section, we could omit the ordering $\+O$ that is used in \Cref{sec:canonical-path} for simplicity.
Now, we start our construction of the canonical path $\gamma^{\sigma, \tau}$.
The construction takes care of the parity of $\abs{\+E^*}$.
We also assume $\ell$ is odd as in \Cref{sec:canonical-path} to avoid annoying corner cases.

\paragraph{When $\abs{\+E^*}$ is odd or $\abs{\+E^*} = \ell + 1$:}
In this case, we use the same construction for $\gamma^{\sigma, \tau}$ as in \Cref{sec:canonical-path}.
The condition on $\abs{\+E^*} =: s$ ensures that for each odd $i \leq \ell$, we have either
\begin{enumerate}
    \item $i \neq s$, i.e., $e_i$ is not the last edge in $\+E^*$ (happens when $\abs{\+E^*}$ is odd);
    \item $i = s$ but $e_s = e_i$ is a leaf of $\^T$ (happens when $\abs{\+E^*} = \ell + 1$ is even).
\end{enumerate}

Now, for an odd $i$, in \textbf{Stage-I} of \Cref{sec:canonical-path}, the edge list $\+E_i$ is constructed as follow.

When $e_i$ has $2$ available colors, we know that it must have an available color $c \not\in \set{a, b}$, and hence we could let $\+E_i = \emptyset$.
We note that the case $i = s$ and $e_i$ is a leaf falls into this case.

When $e_i$ only has $1$ available color, i.e., its current color, then $e_i$ cannot be a leaf.
By our assumption, this means $e_i$ cannot be the last edge in $\+E^*$.
Suppose $e_i = (v_i, v_{i+1})$ such that $v_{i+1}$ is a child of $v_i$.
Then the unique edge $e \in N_e(v_i) \setminus \set{e_i}$ such that $\sigma(e) \in [q]\setminus \sigma(N_e(v_{i+1}))$ will be add to $\+E_i$.
This process will be applied recursively on $e$ until it finally reaches an edge with $2$ available colors (e.g., leaves).
We refer the reader to \Cref{sec:canonical-path} for a more detailed description.

Here, a crucial observation is that $e \neq e_{i-1}$.
This is because we know $\sigma(e_{i+1}) = a$ and hence $\sigma(e_{i-1}) = a \not\in [q] \setminus \sigma(N_e(v_{i+1}))$.
But we know that $\sigma(e) \in [q] \setminus \sigma(N_e(v_{i+1}))$ which means $e$ cannot be $e_{i-1}$.
This property ensures that for different odd $i,j$, the distance between $(e_i \cup \+E_i)$ and $(e_j \cup \+E_i)$ is at least $2$ in the line graph. 
This means all the argument in \Cref{sec:canonical-path} also works in this case.

\paragraph{Otherwise:}
In this case, the last edge $e_s$ in $\+E^*$ must be an odd edge, i.e., $s$ is odd.
If we use the same construction of $\gamma^{\sigma, \tau}$ as in \Cref{sec:canonical-path}, we can no longer ensure that for different odd $i,j$, the distance between $(e_i \cup \+E_i)$ and $(e_j \cup \+E_i)$ is at least $2$.
Hence the argument in \Cref{sec:canonical-path} will fail.
See \Cref{fig:bad-example} for an example.

\begin{figure}[!h]
  \centering
  \colorlet{cR}{magenta}
  \colorlet{cG}{lime}
  \colorlet{cB}{cyan}
  \colorlet{cY}{yellow}
  \colorlet{cD}{darkgray}
  \colorlet{cO}{orange}
  \colorlet{cBG}{black!20}

  \ifarxiv
  \includegraphics{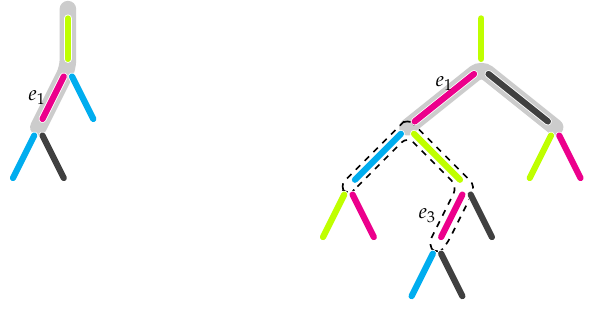}
  \else
  \input{Figures/Fig3.tex} 
  \fi
  
  \newcommand{\colorCirc}[1]{\tikz \node[circle, fill=#1, inner sep=0pt, minimum size = 5pt] {};}
  \caption[An illustration for the construction of canonical path]{ \label{fig:bad-example}
    We consider two $3$-regular trees with \colorCirc{cG}, \colorCirc{cR}, \colorCirc{cB}, \colorCirc{cD} as avaliable colors.
    In both trees, we set $a = \colorCirc{cG}$ and $b = \colorCirc{cR}$.
    In the first example, we mark the odd numbered edge $e_1$ with corresponding $\+E_1$ using shadowed background.
    This means, in order to change the color of $e_1$, by the construction of $\+E_1$ we used in \Cref{sec:canonical-path}, we want to change the color of the root edge.
    Since the root edge only have two avaliable color \colorCirc{cG} and \colorCirc{cR}, this is not possible.
    
    In the second example, we mark $e_1$ and $\+E_1$ with shadowed background and we mark $e_3$ and $\+E_3$ with dashed circle.
    It is then easy to see that these paths have neiboring edges.
    This means, after fliping the edges in $e_3 \cup \+E_3$, the path $e_1 \cup \+E_1$ will be different.
    This will fail the construction in \Cref{sec:canonical-path}, which havily relies on these paths can be handled independently.
  }
\end{figure}

However, in this case, we note that all the even edge $e_j$ cannot be the last edge of $\+E^*$.
Hence, we could modify the original construction of $\gamma^{\sigma, \tau}$ as follow.
\begin{itemize}
    \item [\textbf{Stage-I}] For every positive even edge $j$ (excluding $e_0 = r$), construct $\+E_j$ as we did for odd edges.
    Then we flip these edges ($\cup_{\text{pos. even } j} (\+E_j \cup \set{e_j})$) from button up while making sure that edges in $\+E_j$ is flipped before $e_j$.
    After that, each edge $e_j$ have changed its color from $a$ to $c_i \not\in \set{a, b}$.
    \item [\textbf{Stage-II}] Exchange the color of $e_0$ and $e_1$. Since $\abs{\+E^*}$ is even, we know that $\abs{\+E^*} \geq 2$. Then, change the color of all the other odd edge $e_i$ from $b$ to $a$.
    \item [\textbf{Stage-III}] Reverse of \textbf{Stage-I}.
\end{itemize}
We note that \textbf{Stage-II} is the only place where we use the update $\set{e_0 = r, e_1}$.
Recall that the definition of $\xi_\ell$ that we used in this section has exactly the same form as in \Cref{sec:canonical-path}.
With this construction of $\gamma^{\sigma, \tau}$, one could then verifies $\alpha_\ell \leq 1/2$ and finish the proof of \Cref{lem:edge-dynamics-cond} by routinely following the argument used in the proof of \Cref{lem:congestion} and \Cref{lem:bad-event-prob}.
We omit this part for simplicity.

\section{Open problems}

We provide some interesting directions to explore in future works.
 
\begin{enumerate}
\item For $q=\Delta+1, \Delta\geq 3$, what is the relaxation time of the Glauber dynamics on $\Delta$-regular trees?   We have a lower bound of $\Omega(\Delta n)$ and an upper bound of $\Delta n^{1 + O(1/\log\Delta)}$.
\item Same question for general trees of maximum degree $\Delta$.  For $\Delta=2$ there is a mixing time lower bound of $\Omega(n^3\log{n})$ for $\Delta=2$ by \cite{DGJ}.  Is there a similar lower bound for general trees when $\Delta\geq 3$?
\item For $q\geq\Delta+2$ can we get tight bounds on the mixing time for general trees?  
\item 
For general graphs (not necessarily trees), \cite{WZZ24} proved rapid mixing of the Glauber dynamics when $q>(2+o(1))\Delta$.  Is there a MCMC algorithm (or other method) when $q\geq\Delta+2$,  or is there a hardness of approximate counting result in this parameter region?
\end{enumerate}

\newpage

\bibliographystyle{alpha}
\bibliography{refs.bib}

\newpage

\appendix

\section{The $\Delta n$ lower bound for relaxation time}\label{sec:lower-relax}
In this section, we prove \Cref{thm:lower-relax}.
Let $\^T = (V, E)$ be a tree of $n$ edges with maximum degree $\Delta$.
Let $q \geq \Delta + 1$ and let $\mu$ be the uniform distribution over all the $q$-edge-colorings on $\^T$.
We will give a lower bound on the relaxation time $\Trelax(\mu)$ of the Glauber dynamics on $\mu$.

\begin{lemma} \label{lem:relax-lb}
    If there is an edge $\set{u, v} = e \in E$ such that $\-{deg}_{\^T}(u) = \-{deg}_{\^T}(v) = \Delta$, then for $\Delta + 1 \leq q \leq 2\Delta$, it holds that $\Trelax(\mu) \geq \frac{\Delta n}{2(q-\Delta)^2}$.
\end{lemma}

To prove \Cref{thm:lower-relax}, it is sufficient for us to prove \Cref{lem:relax-lb}.
We will prove the lower bound in \Cref{lem:relax-lb} by using the notion of conductance.
Let $\Omega$ be the support of $\mu$ and let $P \in \^R^{\Omega \times \Omega}$ be the transition matrix for the Glauber dynamics on $\mu$.
The conductance is defined as follow.

\begin{definition} \label{def:conductance}
    For a set $S \subseteq \Omega$, the \emph{conductance $\Phi(S)$ on $S$} is defined as
    \begin{align*}
        \Phi(S) := \frac{\sum_{x \in S, y \in \Omega \setminus S} \mu(x) P(x, y)}{\sum_{z \in S} \mu(z)}.
    \end{align*}
    Then the conductance $\Phi_\star$ for the whole chain is defined as $\Phi_\star := \min_{S:\mu(S) \leq 1/2} \Phi(S)$.
\end{definition}
For convenience, we may use $\mu(S) = \sum_{x \in S} \mu(x)$ to simplify the notation.

The relaxation time is closely related to the conductance, which is known as Cheeger's inequality~\cite{SJ89, lawler1988bounds}.
Here, for convenience, we use the version stated in the textbook.
\begin{lemma}[\text{\cite[Theorem 13.10]{levin2017markov}}] \label{lem:conductance-relax}
   It holds that $\frac{\Phi_\star^2}{2} \leq \Trelax^{-1} \leq 2\Phi_\star$.
\end{lemma}

How, we are ready to prove \Cref{lem:relax-lb}.

\begin{proof} [Proof of \Cref{lem:relax-lb}]
    Let $S \subseteq \Omega$ be defined as
    \begin{align}
        S := \set{\sigma \in \Omega \mid \sigma_e = 1}.
    \end{align}
    It is easy to see that $\mu(S) = 1/q \leq 1/2$ according to the definition.
    Then, in order to escape from $S$, the Glauber dynamics needs to update the edge $e$ to some other color $c \neq 1$.
    This means
    \begin{align*}
        \mu(S) \Phi(S) = 
        \sum_{\sigma \in S, \tau \in \Omega\setminus S} \mu(\sigma) P(\sigma, \tau)
        &= \frac{1}{n} \sum_{\sigma \in S}\mu(\sigma) \frac{\abs{\elist{\sigma}{e}}-1}{\abs{\elist{\sigma}{e}}}
        \leq  \frac{1}{n} \sum_{\sigma \in S}\mu(\sigma) \*1[\abs{\elist{\sigma}{e}} \geq 2],
    \end{align*}
    where we use $\elist{\sigma}{e} = [q] \setminus \set{\sigma_h \mid \text{$h$ incidents $e$}}$ to denote the set of available colors of $e$.
    Hence,
    \begin{align} \label{eq:relax-lb-1}
        \Phi(S) &\leq \frac{1}{n} \Pr[\sigma \sim \mu^{e1}]{\abs{\elist{\sigma}{e}} \geq 2}.
    \end{align}
    According to our assumption in \Cref{lem:relax-lb}, $e = \set{u, v}$ such that $\-{deg}_{\^T}(u) = \-{deg}_{\^T}(v) = \Delta$.
    In the rest part of the proof, we assume $\sigma \sim \mu^{e1}$.
    Let $\vlist{\sigma}{v}$ and $\vlist{\sigma}{u}$ be the set of colors that appears around vertex $v$ and $u$, respectively.
    Then $\abs{\elist{\sigma}{e}} \leq 1$ is equivalent to $[q]\setminus \vlist{\sigma}{v} \subseteq \vlist{\sigma}{u}$ and
    \begin{align*}
        \Pr{\abs{\elist{\sigma}{e}} \leq 1} 
        &=\Pr{[q]\setminus \vlist{\sigma}{v} \subseteq \vlist{\sigma}{u}}
        = \E{\Pr{[q]\setminus \vlist{\sigma}{v} \subseteq \vlist{\sigma}{u} \mid \vlist{\sigma}{v}}} \\
        &= \E{\left.\binom{q - (q - \Delta)}{\Delta - (q - \Delta)}\right/\binom{q-1}{\Delta-1}} = \frac{\Delta! (\Delta - 1)!}{(2\Delta - q)! (q - 1)!}.
    \end{align*}
    Let $q = \Delta + c$ such that $c \in [1, \Delta]$, we have
    \begin{align} \label{eq:relax-lb-2}
         \Pr{\abs{\elist{\sigma}{e}} \leq 1} &=
         \frac{\Delta (\Delta - 1) \cdots (\Delta - c + 1)}{(\Delta + c - 1) (\Delta + c - 2) \cdots \Delta} = \prod_{i=1}^c \frac{\Delta + 1 - i}{\Delta + c - i} \geq \tp{1 - \frac{c}{\Delta}}^c \geq 1 - \frac{c^2}{\Delta},
    \end{align}
    where in the last inequality we use the Bernoulli inequality: $(1 + x)^r \geq 1 + rx$ for $x \geq -1$ and $r \geq 1$.
    Combining \eqref{eq:relax-lb-1} and \eqref{eq:relax-lb-2}, we have
    \begin{align*}
        \Phi_\star \leq \Phi(S) \leq \frac{c^2}{n\Delta}.
    \end{align*}
    According to \Cref{lem:conductance-relax}, we have $\Trelax \geq \frac{n \Delta}{2c^2} = \frac{n\Delta}{ 2(q-\Delta)^2}$.
\end{proof}

\section{Relaxation time of Glauber dynamics when $q = \Delta 
+ 1$}
\label{sec:regular-refined}
In this section, we prove \Cref{thm:regular-refined}.
The proof will follow the same high-level plan as \Cref{thm:main-mixing}.
The crucial part is verifying the conditions required in \Cref{thm:induction}.
However, different from \Cref{sec:glauber-relax-time} where we use \Cref{thm:induction} for $\ell = \Theta(\Delta^2 \log^2 \Delta)$, in this section, we will focus on the $\ell = 1$ case.
Formally, we will prove the following result.

\begin{lemma} \label{lem:verify-condition-ell=1}
    Let $\Delta \geq 2, q = \Delta + 1$, and $\ell = 1$.
    The constants in \Cref{thm:induction} exist with $\alpha_0 = 4\Delta, \alpha_1 = 8$, and $\gamma = 6$.
\end{lemma}

The bounds on $\*\alpha$ and $\gamma$ in \Cref{lem:verify-condition-ell=1} will be established in \Cref{sec:ell=1-alpha} and \Cref{sec:ell=1-gamma}, respectively.
We remark that we want \Cref{thm:regular-refined} to serve as an asymptotic tight upper bound for trees with large degree, which should match the lower bound stated in \Cref{thm:lower-relax}.
This means we need very accurate bounds for $\*\alpha$ and $\gamma$ rather than just say that they are some function of $q$ and $\Delta$.

How, we are able to prove \Cref{thm:regular-refined}.
\begin{proof} [Proof of \Cref{thm:regular-refined}]
    Let $\^T_\ell$ be a complete $\Delta$-regular tree with depth $\ell$.
    Suppose $\^T_\ell$ has $n$ vertices, as it is a complete $\Delta$-regular tree, it holds that $\ell = \Theta(\log_\Delta n)$.
    Let $\mu = \mu_\ell$ be the uniform distribution over all the $(\Delta + 1)$-edge-colorings on $\^T_\ell$.
    According to \Cref{lem:verify-condition-ell=1} and \Cref{thm:induction}, it holds that $\mu$ satisfies approximate tensorization of variance with constant $C \leq 10\Delta \cdot 8^{\ell+1} = \Delta \cdot n^{O(1/\log \Delta)}$. 
    This means $\Trelax$ is bounded by $\Delta \cdot n^{1 + O(1/\log \Delta)}$.
    Finally, according to \Cref{cor:sob-gap} and \eqref{eq:log-sob-mixing}, it holds that $\Tmix$ is bounded by $O(1) \times \Trelax \times \ell (\Delta \log q) \times \log \log \frac{1}{\mu_{\min}} = (\Delta \log \Delta)^2 \cdot n^{1 + O(1/\log \Delta)}$.
\end{proof}

\subsection{Root tensorization on a star}
\label{sec:ell=1-alpha}
In this section, we will give an upper bound for the vector $\*\alpha$ in \Cref{lem:verify-condition-ell=1}.
Before we start, let us first recall the definition of $\^T^\star_\ell$ and $\mu^\star_\ell$ in \Cref{def:mu-k}.
Let $\^T^\star_\ell$ be the complete $d$-ary tree of depth $\ell$ with a hanging root edge $r$. 
For convenience, let $\widetilde{\^T} := \^T^\star - r$.
Let $\mu = \mu^\star_\ell$ be the list coloring distribution on $\^T^\star$, where $\+L_r = \set{1, 2}$ and $\+L_h = [q], q = \Delta + 1$ for other edges $h \in \^T^\star_\ell$.
In order to bound the vector $\*\alpha$, we need to establish the \textbf{root}-tensorization of variance for $\mu^\star_\ell$.
That is, for any function $f:\Omega(\mu) \to \^R$, we want the following inequality to hold:
\begin{align*}
  \Var[\mu]{\mu_{\widetilde{\^T}}[f]}
  &\leq \sum_{i=0}^\ell \alpha_i \sum_{h \in L_i(\^T^\star_\ell)} \mu[\Var[h]{f}].
\end{align*}
%
To achieve this, we do a similar canonical path analysis as we have done in \Cref{sub:congestion-overview}.
The different part is that we use the Cauchy's inequality in a slightly different way to make our bound on $\*\alpha$ tighter.
For each $\sigma \in \Omega(\mu^{r1})$, we pick the unique destination $\tau(\sigma) \in \Omega(\mu^{r2})$ such that $\tau(\sigma)$ is obtained by replacing the maximal $(1,2)$-alternating path from the root $r$ in $\sigma$ with an $(2,1)$-alternating path.
Then, we can bound $\*\alpha$ as follows.

For convenience, given a configuration $\rho \in \Omega(\mu)$ and an edge $h$, we will use $\rho \oslash h$ to denote the configuration that differs at $h$ comparing to $\rho$;
we note that since $q = \Delta + 1$, there is at most one such configuration.
If there is no such configuration, then we let $\rho \oslash h$ be $\rho$.

\begin{lemma} \label{lem:tensor-via-routing}
  Suppose for each $\sigma \in \Omega(\mu^{r1})$, there is a path of configuration $\gamma^\sigma = (\gamma_0^\sigma, \gamma_1^\sigma, \cdots, \gamma_{m(\sigma)}^\sigma)$ of length $m(\sigma)$ that $\gamma_0^\sigma = \sigma$, $\gamma_{m(\sigma)}^\sigma = \tau(\sigma)$ and $\gamma_i^\sigma \oplus \gamma_{i-1}^\sigma = \set{h_i^\sigma}$ for $1 \leq i \leq m(\sigma)$.
  Note that we use $\sigma \oplus \tau := \set{h: \sigma_h \neq \tau_h}$ to denote the symmetric difference between two configurations.
  Then, $\forall 0 \leq t \leq \ell$,
  \begin{align}
    \label{eq:tensor-bound-2}
    \alpha_t &\leq 2(\ell + 1) \cdot \E[\sigma \sim \mu^{r1}]{\abs{\set{i \mid h^\sigma_i \in L_\ell}}} \cdot \max_{h \in L_t, \rho \in \Omega(\mu)} \abs{\set{\sigma \mid (\rho, \rho \oslash h) \in \gamma^\sigma}}.
  \end{align}
\end{lemma}

\begin{proof}[Proof of \Cref{lem:tensor-via-routing}]
  We note that by definition
  \begin{align}
    \Var[\mu]{\mu_{\widetilde{\^T}}[f]}
    \nonumber &= \mu_r(1)\mu_r(2) \tp{\mu^{r1}(f) - \mu^{r2}(f)}^2 \\
    \label{eq:var-all} &= \mu_r(1) \mu_r(2) \tp{\E[\sigma \sim \mu^{r1}]{f(\sigma) - f(\tau(\sigma))}}^2,
  \end{align}
  %
  Then, we could use the telescoping sum for $f(\sigma) - f(\tau(\sigma))$ and we have
  \begin{align} \label{eq:telescope}
    f(\sigma) - f(\tau(\sigma)) &= \sum_{i=1}^{m(\sigma)} \tp{f(\gamma_{i-1}^\sigma) - f(\gamma_i^\sigma)}.
  \end{align}
  Now, combining \eqref{eq:var-all} and \eqref{eq:telescope}, we have
  \begin{align*}
    \tp{\E[\sigma \sim \mu^{r1}]{f(\sigma) - f(\tau(\sigma))}}^2
    &= \tp{\E[\sigma \sim \mu^{r1}]{\sum_{i=1}^{m(\sigma)} \tp{f(\gamma^\sigma_{i-1}) - f(\gamma^\sigma_i)} }}^2 \\
    &= \tp{\sum_{t=0}^\ell \E[\sigma \sim \mu^{r1}]{\sum_{i:h^\sigma_i \in L_t} \tp{f(\gamma^\sigma_{i-1}) - f(\gamma^\sigma_i)} }}^2,
  \end{align*}
  where $L_t$ is the set of edges at level $t$ of the tree $\^T$.
  By using the Cauchy's inequality, we have
  \begin{align} \label{eq:cauchy-1}
    \tp{\E[\sigma \sim \mu^{r1}]{f(\sigma) - f(\tau(\sigma))}}^2 
    \leq (\ell + 1) \sum_{t=0}^{\ell} \tp{\E[\sigma \sim \mu^{r1}]{\sum_{i:h^\sigma_i \in L_t} \tp{f(\gamma^\sigma_{i-1}) - f(\gamma^\sigma_i)} }}^2.
  \end{align}
  For convenience, for $0 \leq t \leq \ell$, we let
  \begin{align} \label{eq:cauchy-2}
    \mathbf{Var}_t := \tp{\E[\sigma \sim \mu^{r1}]{\sum_{i:h^\sigma_i \in L_t}\tp{f(\gamma^\sigma_{i-1}) - f(\gamma^\sigma_i)} }}^2.
  \end{align}
  Since $\mu^{r1}$ is a uniform distribution, we let $p^{r1} := 1/\abs{\Omega(\mu^{r1})}$.
  Now, applying Cauchy's inequality again on $\oVar_t$, we have
  \begin{align}
    \mathbf{Var}_t
    \nonumber
    &= (p^{r1})^2\tp{ \sum_{\sigma \in \Omega(\mu^{r1})} \sum_{i:h^\sigma_i \in L_t} \tp{f(\gamma^\sigma_{i-1}) - f(\gamma^\sigma_i)} }^2 \\
    \nonumber
    &\leq p^{r1}\tp{\sum_{\sigma \in \Omega(\mu^{r1})} \abs{\set{i \mid h^\sigma_i \in L_t}}} \cdot p^{r1} \tp{\sum_{\sigma \in \Omega(\mu^{r1})} \sum_{i:h^\sigma_i \in L_t} \tp{f(\gamma^\sigma_{i-1}) - f(\gamma^\sigma_i)}^2}  \\
    \nonumber
    &= \E[\sigma\sim \mu^{r1}]{\abs{\set{i\mid h^\sigma_i \in L_t}}} \cdot
      p^{r1} \sum_{h \in L_t} \sum_{\rho \in \Omega(\mu)} \abs{\set{\sigma \mid (\rho, \rho\oslash h) \in \gamma^\sigma}} \tp{f(\rho) - f(\rho \oslash h)}^2 \\
    \nonumber
    &\leq 
      \E[\sigma\sim \mu^{r1}]{\abs{\set{i\mid h^\sigma_i \in L_t}}} \cdot
      \max_{h, \rho}\abs{\set{\sigma \mid (\rho, \rho \oslash h) \in \gamma^\sigma}} 
      \cdot p^{r1} \sum_{h \in L_t} \sum_{\rho \in \Omega(\mu)} \tp{f(\rho) - f(\rho \oslash h)}^2 \\
    \label{eq:cauchy-route-2}
    &= 
      \E[\sigma\sim \mu^{r1}]{\abs{\set{i\mid h^\sigma_i \in L_t}}} \cdot
      \max_{h, \rho}\abs{\set{\sigma \mid (\rho, \rho \oslash h) \in \gamma^\sigma}}
      \cdot 8 
\sum_{h \in L_t} \mu[\Var[h]{f}]
  \end{align}
  Then \eqref{eq:tensor-bound-2} could be proved by combining \eqref{eq:var-all}, \eqref{eq:cauchy-1}, \eqref{eq:cauchy-2}, and \eqref{eq:cauchy-route-2}, where we use the fact that $\mu_r(1) = \mu_r(2) = 1/2$.
\end{proof}

Now, we are ready to bound $\*\alpha$ in \Cref{lem:verify-condition-ell=1}.

\begin{proof} [Proof (bound $\*\alpha$ in  \Cref{lem:verify-condition-ell=1})]
Recall that we only consider $\ell = 1$ in \Cref{lem:verify-condition-ell=1}.
This choice makes everything simple and we could finish the proof by a case study.

For convenience, given a configuration $\rho \in \Omega(\mu)$, we will use $\ap[\rho]$ to denote the maximal $(1,2)$-alternating path or $(2, 1)$-alternating path in $\rho$ initiated from the root edge $r$ depending on $\rho_r$.
We construct $\gamma^\sigma$ by considering two cases according to the length of $\ap$: (1) $\abs{\ap} = 1$; (2) $\abs{\ap} = 2$.

When $\abs{\ap} = 1$, it holds that $\ap = \set{r}$ and we have $\tau(\sigma) = \sigma \oslash r$.

When $\abs{\ap} = 2$, it holds that $\ap = \set{r, h}$ where $h$ is some edge in the first level.
Then we can route $\sigma$ to $\tau$ by noticing  that $\tau = \sigma \oslash h \oslash r \oslash h$.

First, we claim that $\alpha_1$ is bounded by $8$.
Since for $\rho \in \Omega(\mu)$ and $h \in L_1$, there is only $1$ possible starting point $\sigma$.
When $\abs{\ap[\rho]} = 1$, we know that $\rho \oslash h = \tau$, so that we could determine $\sigma$ unqiuely.
When $\abs{\ap[\rho]} = 2$, we know that $\rho = \sigma$.
Formally, for every $\rho \in \Omega(\mu)$ and $h \in L_1$, it holds that
\begin{align*}
    \abs{\set{\sigma \mid (\rho, \rho \oslash h) \in \gamma^\sigma}} = 1
\end{align*}
Then, by our construction of $\gamma^\sigma$, it is straight forward to see that
\begin{align*}
    \E[\sigma \sim \mu^{r1}]{\abs{\set{i \mid h^\sigma_i \in L_1}}} \leq 2 \Pr{\abs{\ap} = 2} \leq 2.
\end{align*}
Now, our claim for $\alpha_1$ comes from \Cref{lem:tensor-via-routing}

Then, we claim that $\alpha_0$ is bounded by $4\Delta$.
Given $\rho \in \Omega(\mu)$ and $h = r \in L_0$, we will consider two cases: (1) $\abs{\ap[\rho]}=1$; (2) $\abs{\ap[\rho]} = 2$.
When $\abs{\ap[\rho]}=1$, there are at most $\Delta$ possible starting configuration $\sigma$: (1) $\sigma = \rho$; (2) $\sigma = \rho \oslash g$ for some $g \in L_1$.
This means 
\begin{align*}
\abs{\set{\sigma \mid (\rho, \rho \oslash h) \in \gamma_\sigma}} \leq \Delta.
\end{align*}
When $\abs{\ap[\rho]}=2$, we know that $\abs{\set{\sigma \mid (\rho, \rho \oslash h) \in \gamma_\sigma}} = 0$ from our construction of $\gamma_\sigma$.
Also, we know that
\begin{align*}
    \E[\rho \sim \mu^{r1}]{\abs{\set{i \mid h^\sigma_i 
\in L_1}}} = 1.
\end{align*}
Now, our claim for $\alpha_0$ comes from \Cref{lem:tensor-via-routing}.
\end{proof}

\subsection{Approximate tensorization on a star}
\label{sec:ell=1-gamma}
\newcommand{\Cor}{\Psi^{\-{cor}}}
In this section, we bound $\gamma$ in \Cref{lem:verify-condition-ell=1}.
Let $S_\Delta = (V, [\Delta])$ be the star graph with $\Delta$ edges.
In the rest part of this section, we fix $q = q(\Delta) = \Delta + 1$ and consider the uniform distribution $\mu_\Delta$ over all the $q$-edge-colorings on $S_\Delta$.
We will prove approximate tensorization of variance for $\mu_\Delta$ which is formally stated as in \Cref{lem:var-ell=1-tensor}.
Then, \Cref{lem:var-ell=1-tensor} directly implies that $\gamma \leq C = \exp(\pi^2/6) \leq 6$.
\begin{lemma} \label{lem:var-ell=1-tensor}
  For $\Delta \geq 1$, $\mu_\Delta$ satisfies approximate tensorization of variance with constant $C = \exp(\pi^2/6)$.
\end{lemma}

In the rest part of this section, we will focus on proving \Cref{lem:var-ell=1-tensor}.
The proof relies on the standard local-to-global argument for the high-dimensional expander~\cite{alev2020improved}.

Before we start the proof, we first give a brief introduction on this technique.
For simplicity we specialize the notions for the special distribution $\mu$ that we will consider in this section.

For any integer $\Delta \geq 2$, let $\mu = \mu_\Delta$, we consider the following matrices:
\begin{itemize}
\item Let $\Cor \in \^R^{q\Delta \times q\Delta}$ be the correlation matrix of $\mu$ defined as follow
  \begin{align} \label{eq:cor-matrix}
    \forall u, v \in [\Delta], \quad \forall a, b \in [q], \quad \Cor(ua, vb) := \mu^{ua}_v(b) - \mu_v(b).
  \end{align}
\item Let $P_\Delta \in \^R^{q\Delta \times q\Delta}$ be the transition matrix of the non-lazy local random walk defined as follow
  \begin{align} \label{eq:local-walk}
    \forall u, v \in [\Delta], \quad \forall a, b \in [q], \quad P_\Delta(ua, vb) := \frac{\*1[u \neq v]}{\Delta - 1} \cdot \mu^{ua}_v(b).
  \end{align}
\end{itemize}

\begin{lemma}[\text{\cite[Theorem 1.5]{alev2020improved}}] \label{lem:local-to-global}
  The distribution $\mu_\Delta$ satisfies approximate tensorization of variance with constant $C = \prod_{i=0}^{\Delta-2} (1 - \lambda_2(P_{\Delta-i}))^{-1}$ where $\lambda_2$ refers to the second largest eigenvalue of a matrix.
\end{lemma}

According to the definition of $\mu$, \eqref{eq:cor-matrix} and \eqref{eq:local-walk}, it is straightforward to verify that
\begin{align} \label{eq:cor-local-walk}
  \Cor_\Delta = (\Delta - 1) P_\Delta - \^1 \^1^\top/q + \^I,
\end{align}
where we use $\^1$ to denote the all-$1$ vector over $\^R^{q\Delta}$ and $\^I$ is the identity matrix over $\^R^{q\Delta \times q\Delta}$.
Since $\Cor_\Delta \^1 = 0$, according to \eqref{eq:cor-local-walk}, we have $\lambda_{\max}(\Cor_\Delta) = (\Delta - 1)\lambda_2(P_\Delta) + 1$, which means
\begin{align} \label{eq:cor-local-walk-eigenvalue}
  \lambda_2(P_\Delta) &= (\lambda_{\max}(\Cor_\Delta) - 1) / (\Delta - 1).
\end{align}

We have the following upper bound for $\lambda_{\max}(\Cor_\Delta)$.

\begin{lemma} \label{lem:cor-lambda-max}
  It holds that $\lambda_{\max}(\Cor_\Delta) \leq 1 + \frac{1}{\Delta}$. 
\end{lemma}

Now, we are ready to prove \Cref{lem:var-ell=1-tensor}.
\begin{proof}[proof of \Cref{lem:var-ell=1-tensor}]
  According to \Cref{lem:local-to-global}, \eqref{eq:cor-local-walk-eigenvalue}, and \Cref{lem:cor-lambda-max}, it holds that
  \begin{align*}
    C &\leq \exp\tp{\sum_{i=0}^{\Delta-2} \frac{1}{(\Delta-i)(\Delta - i - 1)}} \leq \exp\tp{\sum_{i=1}^\infty \frac{1}{i^2}} = \exp(\pi^2/6). \qedhere
  \end{align*}
\end{proof}

We only left to prove \Cref{lem:cor-lambda-max}.

\begin{proof} [Proof of \Cref{lem:cor-lambda-max}]
  According to the definition of $\Cor$ in \eqref{eq:cor-matrix}, for $u, v \in S_\Delta$ and $a, b \in [q]$, we have
  \begin{align} \label{eq:cor-matrix-explicit}
    \Cor_\Delta(ua, vb) &= - \frac{1}{\Delta + 1} + \*1[u\neq v \land a\neq b]\frac{1}{\Delta} + \*1[u = v \land a = b].
  \end{align}
  We may consider $\Cor$ as a block matrix by letting coordinates with the same color in the same block.
  Then, it holds that
  \begin{align} \label{eq:cor-matrix-block}
    \Cor_\Delta &= \tp{1 + \frac{1}{\Delta}} \^I + \tp{\frac{1}{\Delta} - \frac{1}{\Delta+1}} \^1\^1^\top - \frac{1}{\Delta}\tp{\^I_q \otimes \tp{\^1_\Delta \^1_\Delta^\top} + \tp{\^1_q \otimes \^I_\Delta}\tp{\^1_q \otimes \^I_\Delta}^\top},
  \end{align}
  where $\otimes$ is the Kronecker product and we let $\^I$ denote the identity matrix over $\^R^{q\Delta \times q\Delta}$, $\^1$ denote the all-$1$ vector over $\^R^{q\Delta}$.
  Similarly, for $1 \leq t \leq q\Delta$, we let $\^I_t$ be the identity matrix in $\^R^{t\times t}$ and we let $\^1_t$ be the all-$1$ vector of length $t$.

  According to \eqref{eq:cor-matrix}, it holds that $\Cor_\Delta \^1 = 0$.
  Since $\Cor_\Delta$ is a symmetric matrix, to prove \Cref{lem:cor-lambda-max}, we only need to bound $\frac{f^\top \Cor_\Delta f}{f^\top f}$ for $f \neq 0$ and $f \perp \^1$.
  By \eqref{eq:cor-matrix-block}, we have
  \begin{align*}
    f^\top \Cor_\Delta f
    &\overset{(\star)}{\leq} f^\top\tp{\tp{1 + \frac{1}{\Delta}} \^I + \tp{\frac{1}{\Delta} - \frac{1}{\Delta+1}} \^1\^1^\top}f 
    \overset{f\perp \^1}{=} \tp{1 + \frac{1}{\Delta}} f^\top f,
  \end{align*}
  where $(\star)$ holds by
  \begin{align*}
    \^I_q \otimes \tp{\^1_\Delta \^1_\Delta^\top} &\succcurlyeq 0 \quad \text{and} \quad \tp{\^1_q \otimes \^I_\Delta}\tp{\^1_q \otimes \^I_\Delta}^\top \succcurlyeq 0. \qedhere
  \end{align*}
\end{proof}

\section{Block factorization via root-factorization: Proof of \Cref{thm:adjacentpair}}\label{sec:adjacentpair}

In this section, we prove \Cref{thm:adjacentpair}.
Note that \Cref{thm:adjacentpair} is a generalization of \Cref{thm:induction}, and the proof presented here is a generalization of the proof of \Cref{thm:induction} presented in \Cref{sec:proof-induction}.

Recall the definitions of $\^T_k, \^T^\star_k$ and $\mu_k, \mu^\star_k, \mu^{\star,1}_k$ from \Cref{def:Tk-Tk-star,def:mu-k}.
We will prove a more general version of \Cref{thm:adjacentpair}, which will yield \Cref{thm:adjacentpair} as a corollary.

\begin{definition} \label{def:collection-A}
For $\ell \geq 1$, let $\+A$ be a collection of edge sets of $\^T^\star_\ell$ such that $\forall S \in \+A$ satisfies:
\begin{itemize}
    \item $\abs{S} \geq 2$, i.e., $S$ is not a singleton edge in $\^T^\star_\ell$;
    \item $S \cap L_\ell(\^T^\star_\ell) = \emptyset$.
\end{itemize}
\end{definition}

Then, we define the collection $\+B$ as
\begin{align} \label{eq:def-gen-edge-block}
    \+B := \set{\set{e} \mid e \in \^T^\star_\ell} \cup \+A,
\end{align}
so that \eqref{eq:def-gen-edge-block} is a generalization of \eqref{eq:def-edge-block}.
Then, following \eqref{eqn:root-block-factorization}, we say that \emph{approximate \textbf{root}-factorization of variance} for $\mu^\star_\ell$ holds with constants $C(B)$ if for all $f:\Omega(\mu^\star_\ell) \to \^R$
\begin{align*}
\Var[\mu_\ell^\star]{\mu_{\^T_\ell^\star\setminus\{r\}}[f]} \leq \sum_{B\in \+B}C(B)\mu_\ell^\star(\Var[B]{f}).
\end{align*}

\begin{lemma} \label{lem:induction-block}
  Let $\Delta \geq 2$, $q \geq \Delta+1$ be two integers.
  Suppose there exists an integer $\ell\geq 1$ where the following holds:
  \begin{itemize}
  \item There exist constants $\alpha=(\alpha_0,\dots,\alpha_\ell)$ and $\beta$ and 
  approximate {\bf root}-factorization of variance for $\mu^\star_k$ holds with constants $C(\set{e})=\alpha_i$ where $e\in L_i(\^T_\ell^\star)$ and $C(B) = \beta$ where $\abs{B} \geq 2$;
  \item  There exists a constant $\gamma$ such that for $1 \leq j \leq \ell$, both $\mu_j$ and $\mu^{\star,1}_j$ satisfy approximate tensorization of variance with constant $\gamma$.
  \end{itemize}
  Then, for any $k \geq 1$, $\mu_k$ satisfies the following block factorization of variance: for any $f:\Omega(\mu_k) \to \^R$, 
  \begin{align} \label{eq:ind-block-target}
      \Var[\mu_k]{f} &\leq \tp{\sum_{t=1}^k F_k(t) \sum_{e \in L_t(\^T_k)} \mu[\Var[e]{f}]} + \beta \sum_{\substack{1 \leq t < k \\ t \equiv k (\-{mod}\; \ell)}} F_t(t) \sum_{r \in L_t(\^T_k)}\sum_{B \in \+A^r} \mu[\Var[B]{f}],
  \end{align}
  where $F_\cdot(\cdot)$ is the function defined in \eqref{eq:F-recursion} and for every edge $r \leq k - \ell$, we let $\+A^r$ be a copy of $\+A$ by seeing the sub-tree rooted at $r$ of depth $\ell$ as $\^T^\star_\ell$.
\end{lemma}

Now we can prove \Cref{thm:adjacentpair} as a direct corollary of \Cref{lem:induction-block}.
\begin{proof} [Proof of \Cref{thm:adjacentpair}]
    We set $\+B$ as in \eqref{eq:def-edge-block}.
    Since \Cref{thm:adjacentpair} assumes $\ell \geq 2$, we know this specific $\+B$ satisfies the requirements in \Cref{def:collection-A} and \eqref{eq:def-gen-edge-block}.
    Now, according to \Cref{lem:induction-block}, we know that we have the block factorization in \eqref{eq:ind-block-target}.
    Furthermore, according to \Cref{sec:upper-bound-F}, we know that the function $F_\cdot(\cdot)$ can be bounded by $\widehat{F}_k$ defined in \eqref{eq:Fk-explicit}, which means for any $1 \leq t \leq k$, we have
    \begin{align*}
        F_k(t) &\leq  \gamma \cdot \tp{\max_{0 \leq j \leq \ell} \alpha_j \cdot \sum_{i=0}^{\ftp{k/\ell}}\alpha_\ell^i + \max\set{1, \alpha_0}}.
    \end{align*}
    We finish the prove by comparing terms in \eqref{eq:ind-block-target} and the block factorization claimed in \Cref{thm:adjacentpair} using the fact that variance is non-negative.
\end{proof}

In the rest part of this section, we will prove \Cref{lem:induction-block}.
Fix parameters $d,q,\ell$ and $\*\alpha$ in the theorem.
Our proof is based on an induction of the depth $k$.

  \paragraph{Base case ($k \leq \ell$):}
  When $k \leq \ell$, \eqref{eq:ind-block-target} follows from the second assumption of \Cref{lem:induction-block} with $F_k(t) = \gamma$.

  \paragraph{Inductive case ($k > \ell$):}
  Suppose \eqref{eq:ind-block-target} holds for all smaller $k' < k$.
  We prove that the theorem holds for $k$.
  For convenience, we denote $\mu = \mu_k$ and $\^T = \^T_k$.
  Let $R = \cup_{i=k-\ell+1}^k L_i(\^T)$ be the edge set that contains all the  $\ell$ levels in the bottom of $\^T$.
  According to the law of total variance, 
  \begin{align} \label{eq:law-total-var-block}
    \Var[\mu]{f} &\overset{\eqref{eq:var-total-law}}{=} \mu[\Var[R]{f}] + \Var[\mu]{\mu_R[f]}.
  \end{align}

 The first term in the RHS of \eqref{eq:law-total-var-block} could be bounded by \Cref{claim-bound}.

  Next, we bound the second term in the RHS of~\eqref{eq:law-total-var-block} by induction hypothesis.
  Note that for any $\sigma \in \Omega(\mu)$, the value of function $\mu_R[f](\sigma)$ depends only on $\sigma_{\^T \setminus R}$. 
  By viewing $\mu_R[f]$ as a function $f':\Omega(\mu_{\^T \setminus R}) \to \^R$, we have
  $\Var[\mu]{\mu_R[f]} = \Var[\mu_{\^T - R}]{f'} = \Var[\mu_{{k-\ell}}]{f'}$, where the last equation uses the fact that the marginal distribution $\mu_{\^T \setminus R}$ is the same as $\mu_{{k-\ell}}$ and $\mu_{{k-\ell}}$ is the uniform distribution over proper $q$-edge-coloring in tree $\^T_{k-\ell}$.
  By the I.H. on $\mu_{k-\ell} = \mu_{\^T \setminus R}$, we have
  \begin{align} 
    \Var[\mu]{\mu_R[f]}
    \nonumber &\leq \sum_{t=1}^{k-\ell} F_{k-\ell}(t) \sum_{h \in L_t(\^T)} \mu_{\^T \setminus R}[\Var[h]{f'}] + \beta \sum_{\substack{1 \leq t < k-\ell \\ t\equiv k (\-{mod} \ell)}} F_t(t) \sum_{r \in L_t(\^T)} \sum_{B \in \+A^r} \mu_{\^T\setminus R}[\Var[B]{f'}] \\
    \nonumber &\overset{(*)}{=} \sum_{t=1}^{k-\ell-1} F_{k-\ell}(t) \sum_{h \in L_t(\^T)} \mu[\Var[h]{\mu_R[f]}] 
                + F_{k-\ell}(k-\ell) \sum_{h \in L_{k-\ell}(\^T)} \mu\left[\Var[\^T^h]{\mu_{R}[f]}\right] \\
    \nonumber &\quad + \beta \sum_{\substack{1 \leq t < k-\ell \\ t\equiv k (\-{mod} \ell)}} F_t(t) \sum_{r \in L_t(\^T)} \sum_{B \in \+A^r} \mu[\Var[B]{\mu_R[f]}] \\
    \label{eq:ind-hypo-block}
              &\overset{\eqref{eq:var-convex}}{\leq} \sum_{t=1}^{k-\ell-1} F_{k-\ell}(t) \sum_{h \in L_t(\^T)} \mu[\Var[h]{f}] 
                + F_{k-\ell}(k-\ell) \sum_{h \in L_{k-\ell}(\^T)} \mu\left[\Var[\^T^h]{\mu_{\widetilde{\^T}^h}[f]}\right] \\
    \nonumber &\quad + \beta \sum_{\substack{1 \leq t < k-\ell \\ t\equiv k (\-{mod} \ell)}} F_t(t) \sum_{r \in L_t(\^T)} \sum_{B \in \+A^r} \mu[\Var[B]{f}],
  \end{align}
  where we use $\^T^h$ to denote the subtree with the hanging root edge $h$ and use $\widetilde{\^T}^h := \^T^h \setminus h$ to denote the subtree beneath the edge $h$.  
  The equation ($*$) holds because
  \begin{itemize}
      \item for all $h \in L_t(\^T)$ with $t < k - \ell$, all the neighboring edges of $h$ are in $\^T \setminus R$, i.e. $\partial h \subseteq \^T \setminus R$, by conditional independence, $\mu_{\^T \setminus R}[\Var[h]{f'}] =  \mu[\Var[h]{\mu_R[f]}]$;
      \item for all $1 \leq t < k -\ell$ such that $t\equiv k (\-{mod} \ell)$, we know that $t \leq k - 2\ell$, and for $r \in L_t(\^T_k)$, by the definition of $\+A^r$ and $\+A$ in \Cref{def:collection-A}, we know that $\forall B \in \+A^r$, $\ol{\partial} B \subseteq \^T\setminus R$, which, by conditional independence, implies $\mu_{\^T\setminus R}[\Var[B]{f'}] = \mu[\Var[B]{\mu_R[f]}]$;
      \item for all $h \in L_{k-\ell}(\^T)$, $\mu_{\^T \setminus R}[\Var[h]{f'}] = \mu\left[\Var[\^T^h]{\mu_{R}[f]}\right]$ because for any $\sigma \in \Omega(\mu)$, the value of $\mu_R[f](\sigma)$ is independent from $\sigma_{\^T_h \setminus h}$. 
  \end{itemize}
  Since $\mu\left[\Var[\^T^h]{\mu_{\widetilde{\^T}^h}[f]}\right]=\sum_{\sigma \in \Omega(\mu_{\^T  \setminus \widetilde{\^T}^h }) } \mu_{\^T  \setminus \widetilde{\^T}^h }(\sigma)\Var[\mu^\sigma]{ \mu_{\widetilde{\^T}^h} [f]}$. Fix one $\sigma$. One can view $f$ as a function from $\Omega(\mu^\sigma_{\widetilde{\^T}^h})$ to $\^R$. By applying the first assumption of \Cref{lem:induction-block} to each $\sigma$ individually and rearranging, 
  \begin{align} \label{eq:route-decay-block}
    \sum_{h \in L_{k-\ell}(\^T)} \mu\left[\Var[\^T^h]{\mu_{\widetilde{\^T}^h}[f]}\right]
    &\leq \sum_{i=0}^\ell \alpha_i \sum_{g \in L_{k-\ell + i}(\^T)} \mu[\Var[g]{f}]
      + \beta \sum_{r \in L_{k-\ell}(\^T)} \sum_{B \in \+A^r} \mu[\Var[B]{f}].
  \end{align}
  Combining \eqref{eq:law-total-var-block}, \eqref{eq:mu[Var_R[f]]} in \Cref{claim-bound}, \eqref{eq:ind-hypo-block}, and \eqref{eq:route-decay-block}, it holds that
  \begin{align*}
    \Var[\mu]{f}
    &\leq \sum_{t=1}^{k-\ell-1} F_{k-\ell}(t) \sum_{h \in L_t(\^T)} \mu[\Var[h]{f}] \\
    &\quad + F_{k-\ell}(k-\ell) \cdot \alpha_0 \sum_{h \in L_{k-\ell}(\^T)} \mu[\Var[h]{f}] \\
    &\quad + \sum_{i=1}^\ell (\alpha_i F_{k-\ell}(k-\ell) + \gamma) \sum_{h\in L_{k-\ell+i}(\^T)} \mu[\Var[h]{f}] \\
    &\quad + \beta \sum_{\substack{1 \leq t < k \\ t\equiv k (\-{mod} \ell)}} F_t(t) 
 \sum_{r \in L_t(\^T)} \sum_{B \in \+A^r} \mu[\Var[B]{f}].
  \end{align*}
  The above inequality immediately gives the recurrence of $F$ as in \eqref{eq:F-recursion} and finishes the proof.

\section{Monotonicity of the Glauber dynamics (Proof of \Cref{lem:monotonicity-glauber})}
\label{sec:monotonicity-glauber}
In this section, we prove \Cref{lem:monotonicity-glauber}.
Recall that $\^T = (V, E)$ is a tree with root $r$ and for $H \subseteq E$ the tree $\^T_H = (V_H, H \subseteq E)$ is a connected sub-tree of $\^T$ containing $r$.
We use $\mu$ and $\nu$ to denote the distribution of uniform $q$-edge-coloring on $\^T$ and $\^T_H$, respectively.

By the definition of approximate tensorization of variance and log-Sobolev constant in \Cref{def:tensorization} and \eqref{eq:def-log-sob}, to prove \Cref{lem:monotonicity-glauber}, it is sufficient for us to prove the following inequalities for every function $g:\Omega(\mu) \to \^R$:
\begin{align} 
  \label{eq:monotone-GD-AT-target}
  \Var[\nu]{g} &\leq Cq \sum_{e \in H} \nu[\Var[e]{g}] \\
  \label{eq:monotone-GD-sob-target}
  \Ent[\nu]{g^2} &\leq c_{\-{sob}}^{-1}(\mu) q \sum_{e \in H} \nu[\Var[e]{g}],
\end{align}
where \eqref{eq:monotone-GD-AT-target} is for the first part of \Cref{lem:monotonicity-glauber} and \eqref{eq:monotone-GD-sob-target} is for the second part of \Cref{lem:monotonicity-glauber}.
We recall the notation $\Ent{\cdot}$ for entropy is defined in \eqref{eq:def-ent}.

For a fixed function $g: \Omega(\nu) \to \^R$, we construct a function $f:\Omega(\mu) \to \^R$ as follow
\begin{align} \label{eq:def-f-g}
  \forall \sigma \in \Omega(\nu), \quad f(\sigma) := g(\sigma_H),
\end{align}
where we use $\sigma_H$ to denote the sub-vector obtained by restricting $\sigma$ to the edge set $H$.
Then~\eqref{eq:monotone-GD-AT-target} and \eqref{eq:monotone-GD-sob-target} follow directly from the following claims.
Through the whole proof, the crucial property that we use is $\nu = \mu_H$.
\begin{claim} \label{claim:equal-var}
  $\Var[\mu]{f} \geq \Var[\nu]{g}$ and $\Ent[\mu]{f^2} \geq \Ent[\nu]{g^2}$.
\end{claim}
\begin{proof}

  We focus on proving $\Ent[\mu]{f^2} \geq \Var[\nu]{g^2}$.
  The fact $\Var[\mu]{f} \geq \Var[\nu]{g}$ can be proved using a similar argument.
  Let $X \sim \mu$ be a random edge coloring and let $F = f^2(X)$ be a random variable.
  Since $X \sim \mu$ and $\nu = \mu_H$, we know that $X_H \sim \nu$.
  This means we can write
  \begin{align}
      \nonumber
      \Ent[\mu]{f^2} &= \Ent[X\sim \mu]{f^2(X)} = \Ent{F} \\
      \label{eq:g-Ef}
      \Ent[\nu]{g^2} &= \Ent[X \sim \mu]{g^2(X_H)} = \Ent[X \sim \mu]{\E{f^2(X) \mid X_H}} = \Ent{\E{F \mid X_H}},
  \end{align}
  where the second equation in \eqref{eq:g-Ef} follows from the definition of $f$ in \eqref{eq:def-f-g}.
  Finally, according to the law of total entropy, $\Ent{F} - \Ent{\E{F \mid X_H}} = \E{\Ent{F \mid X_H}}$, which is $\geq 0$ according to the non-negativity of entropy.
\end{proof}

\begin{claim} \label{claim:local-var-GD}
  $\sum_{e \in E} \mu[\Var[e]{f}] \leq q \sum_{h \in H} \nu[\Var[h]{g}]$.
\end{claim}
\begin{proof}
  For $e \not\in H$, note that $\mu[\Var[e]{f}] = 0$ by the definition of $f$.
  Then for every $h \in H$, we will show that $\mu[\Var[h]{f}] \leq q \nu[\Var[h]{g}]$.
  By the definition of $\mu[\Var[h]{f}]$, we have
  \begin{align*}
    \mu[\Var[e]{f}]
    &= \frac{1}{2} \sum_{\sigma \in \Omega(\mu_{E\setminus h})} \mu_{E\setminus h}(\sigma) \sum_{a, b \in \Omega(\mu^\sigma_h)} \mu^\sigma_h(a) \mu^\sigma_h(b) (f(\sigma, a) - f(\sigma, b))^2 \\
    &\leq \frac{1}{2} \sum_{\sigma \in \Omega(\mu_{E\setminus h})} \sum_{a, b \in \Omega(\mu^\sigma_h)} \mu(\sigma, a) (f(\sigma, a) - f(\sigma, b))^2 \\
    &= \frac{1}{2} \sum_{\alpha \in \Omega(\mu_{H\setminus h})} \sum_{\beta \in \Omega(\mu^{\alpha}_{E\setminus H})} \sum_{a, b \in \Omega(\mu^{\alpha,\beta}_h)} \mu(\alpha, \beta, a) (f(\alpha, \beta, a) - f(\alpha, \beta, b))^2 \\
    &= \frac{1}{2} \sum_{\alpha \in \Omega(\mu_{H\setminus h})} \mu_{H\setminus h}(\alpha) \sum_{a, b \in \Omega(\mu^\alpha_h)} \mu^\alpha_h(a) \sum_{\beta \in \Omega(\mu^{\alpha,a}_{E\setminus H}) \cap \Omega(\mu^{\alpha,b}_{E\setminus H})} \mu^{\alpha, a}_{E\setminus H}(\beta) (g(\alpha, a) - g(\alpha, b))^2 \\
    &\leq \frac{1}{2} \sum_{\alpha \in \Omega(\mu_{H\setminus h})} \mu_{H\setminus h}(\alpha) \sum_{a, b \in \Omega(\mu^\alpha_h)} \mu^\alpha_h(a) (g(\alpha, a) - g(\alpha, b))^2 \\
    &\leq q \cdot \frac{1}{2} \sum_{\alpha \in \Omega(\mu_{H\setminus h})} \mu_{H\setminus h}(\alpha) \sum_{a, b \in \Omega(\mu^\alpha_h)} \mu^\alpha_h(a) \mu^\alpha_h(b) (g(\alpha, a) - g(\alpha, b))^2 
    = q \nu[\Var[h]{g}],
  \end{align*}
  where in the last inequality, we use the fact that $\mu^\alpha_h(b) \geq 1/q$ and $\nu = \mu_H$.
\end{proof}

\begin{proof}[Proof of \Cref{lem:monotonicity-glauber}]
    (First part): suppose $\mu$ satisfies approximate tensorization of variance with constant $C$, which could be written as the following inequality: for any $f: \Omega(\mu) \to \^R$, 
    \begin{align*}
        \Var[\mu]{f} \leq C \sum_{e \in E} \mu[\Var[e]{f}].
    \end{align*}
    Then, for any function $g: \Omega(\nu) \to \^R$, we define the function $f$ as in \eqref{eq:def-f-g}.
    According to \Cref{claim:equal-var} and \Cref{claim:local-var-GD}, it hold that
    \begin{align*}
       \Var[\nu]{g} \leq \Var[\mu]{f} \leq C \sum_{e \in E} \mu[\Var[e]{f}] \leq Cq \sum_{h \in H} \nu[\Var[h]{g}],
    \end{align*}
    which proves \eqref{eq:monotone-GD-AT-target} and hence the first part of \Cref{lem:monotonicity-glauber}.

    (Second part): Suppose $\mu$ has log-Sobolev constant $c_{\-{sob}}(\mu)$.
    According to definition, this means for any $f: \Omega(\mu) \to \^R$, it holds that
    \begin{align*}
        \Ent[\mu]{f^2} \leq c_{\-{sob}}^{-1}(\mu) \sum_{e \in E} \mu[\Var[e]{f}].
    \end{align*}
    Then, for any function $g: \Omega(\nu) \to \^R$, we define the function $f$ as in \eqref{eq:def-f-g}.
    According to \Cref{claim:equal-var} and \Cref{claim:local-var-GD}, it hold that
    \begin{align*}
       \Ent[\nu]{g^2} \leq \Ent[\mu]{f^2} \leq c_{\-{sob}}^{-1}(\mu) \sum_{e \in E} \mu[\Var[e]{f}] \leq q c_{\-{sob}}^{-1}(\mu) \sum_{h \in H} \nu[\Var[h]{g}],
    \end{align*}
    which gives \eqref{eq:monotone-GD-sob-target} and hence proves the second part of \Cref{lem:monotonicity-glauber}.
\end{proof}

\section{Monotonicity of the heat-bath edge dynamics}
\label{sec:monotonicity-edge-dynamics}
In this section, we prove \Cref{lem:monotonicity-edge-dynamics}.
Recall that $\^T = (V, E)$ is a tree with root $r$ and for $H \subseteq E$ the tree $\^T_H = (V_H, H \subseteq E)$ is a connected sub-tree of $\^T$ containing $r$
We use $\mu$ and $\nu$ to denote the distribution of uniform $q$-edge-coloring on $\^T$ and $\^T_H$, respectively.
The proof follows the same high level plan as the proof of \Cref{lem:monotonicity-glauber}.
The only difference is that we need to consider edge pairs instead of singleton edges.

Fix a function $g:\Omega(\nu) \to \^R$, let the function $f: \Omega(\mu) \to \^R$ be defined as in \eqref{eq:def-f-g}.
We note that \Cref{lem:monotonicity-edge-dynamics} follows directly from \Cref{claim:equal-var}, \Cref{claim:local-var-GD}, and the following lemma.
\begin{lemma} \label{lem:local-var-paring-edges}
  Let $\+P_E$ and $\+P_H$ be the sets of neighboring edge pairs in $E$ and $H$, respectively.
  It holds that $\sum_{\set{e,h}\in \+P_E} \mu[\Var[\set{e,h}]{f}] \leq q^2 \sum_{\set{e,h} \in \+P_H} \nu[\Var[\set{e,h}]{g}] + 2q^2 \sum_{e \in H} \nu[\Var[e]{f}]$.
\end{lemma}
\begin{proof}
  To finish the proof, we do a case study for $\set{e, h} \in \+P_E$, we claim that:
  \begin{enumerate}
  \item for $e \not\in H$ and $h \not\in H$, $\mu[\Var[\set{e,h}]{f}] = 0$;
  \item for $e \in H$ and $h \in H$, $\mu[\Var[\set{e,h}]{f}] \leq q^2 \nu[\Var[\set{e,h}]{g}]$;
  \item for $e \in H$ and $h \not\in H$, $\mu[\Var[\set{e,h}]{f}] \leq q \nu[\Var[e]{f}]$.
  \end{enumerate}
  Then, \Cref{lem:local-var-paring-edges} follows from these claims directly.
  The first case is trivial due to the definition of the function $f$.
  For the second case, by definition of $\mu[\Var[\set{e,h}]{f}$, we have
  \begin{align*}
    \mu[\Var[\set{e,h}]{f}]
    &= \frac{1}{2} \sum_{\sigma \in \Omega(E\setminus\set{e,h})}\mu_{E\setminus\set{e,h}}(\sigma) \sum_{a, b \in \Omega(\mu^\sigma_{e,h})} \mu^\sigma_{e,h}(a) \mu^\sigma_{e,h}(b) (f(\sigma, a) - f(\sigma, b))^2 \\
    &\leq \frac{1}{2} \sum_{\sigma \in \Omega(E\setminus\set{e,h})} \sum_{a, b \in \Omega(\mu^\sigma_{e,h})} \mu(\sigma, a) (f(\sigma, a) - f(\sigma, b))^2 \\
    &= \frac{1}{2} \sum_{\alpha \in \Omega(\mu_{H\setminus\set{e,h}})} \sum_{\beta \in \Omega(\mu^\sigma_{E\setminus H})} \sum_{a, b \in \Omega(\mu^{\alpha, \beta}_{e,f})} \mu(\alpha, \beta, a) (f(\alpha, \beta, a) - f(\alpha, \beta, b))^2 \\
    &= \frac{1}{2} \sum_{\alpha \in \Omega(\mu_{H\setminus\set{e,h}})} \mu_{H\setminus\set{e,h}}(\sigma) \sum_{a, b \in \Omega(\mu^{\alpha}_{e,f})} \mu^\sigma_{e,h}(a) \sum_{\beta \in \Omega(\mu^{\sigma,a}_{E\setminus H}) \cap \Omega(\mu^{\sigma,b}_{E\setminus H})} \mu^{\alpha,a}_{E\setminus H}(\beta) (f(\alpha, \beta, a) - f(\alpha, \beta, b))^2 \\
    &\leq \frac{1}{2} \sum_{\alpha \in \Omega(\mu_{H\setminus\set{e,h}})} \mu_{H\setminus\set{e,h}}(\sigma) \sum_{a, b \in \Omega(\mu^{\alpha}_{e,f})} \mu^\sigma_{e,h}(a) (g(\alpha, a) - g(\alpha, b))^2 \\
    &\leq q^2 \frac{1}{2} \sum_{\alpha \in \Omega(\mu_{H\setminus\set{e,h}})} \mu_{H\setminus\set{e,h}}(\sigma) \sum_{a, b \in \Omega(\mu^{\alpha}_{e,f})} \mu^\sigma_{e,h}(a) \mu^\sigma_{e,h}(b) (g(\alpha, a) - g(\alpha, b))^2  = q^2 \nu[\Var[\set{e,h}]{g}],
  \end{align*}
  where, in the last inequality, we use the fact that $\nu = \mu_H$ and $\mu^{\sigma}_{e,h}(b) \geq 1/q^2$.
  For the third case where $e \in H$ and $h \not\in H$, by a similar calculation, we have
  \begin{align*}
    &\mu[\Var[\set{e,h}]{f}] \\
    &= \frac{1}{2} \sum_{\sigma \in \Omega(\mu_{E\setminus\set{e,h}})} \mu_{E\setminus\set{e, h}}(\sigma) \sum_{a,b \in \Omega(\mu^\sigma_{e,h})} \mu^\sigma_{e,h}(a) \mu^\sigma_{e,h}(b) (f(\sigma, a) - f(\sigma, b))^2 \\
    &= \frac{1}{2} \sum_{\sigma \in \Omega(\mu_{E\setminus\set{e,h}})} \mu_{E\setminus\set{e, h}}(\sigma) \sum_{a,b \in \Omega(\mu^{\sigma}_e)} \mu^\sigma_{e}(a) \mu^\sigma_{e}(b) \sum_{x,y \in \Omega(\mu^{\sigma, a}_h) \cap \Omega(\mu^{\sigma, b}_h)} \mu^{\sigma,a}_h(x) \mu^{\sigma,b}_h(y) (f(\sigma, a, x) - f(\sigma, b, y))^2 \\
    &\leq \frac{1}{2} \sum_{\sigma \in \Omega(\mu_{E\setminus\set{e,h}})} \mu_{E\setminus\set{e, h}}(\sigma) \sum_{a,b \in \Omega(\mu^{\sigma}_e)} \mu^\sigma_{e}(a) (f(\sigma, a, \cdot) - f(\sigma, b, \cdot))^2 \\
    &= \frac{1}{2} \sum_{\alpha \in \Omega(\mu_{H\setminus \set{e}})} \mu_{H\setminus\set{e}}(\alpha) \sum_{a,b \in \Omega(\mu^{\alpha}_e)} \mu^\alpha_{e}(a) \sum_{\beta \in \Omega(\mu^{\alpha,a}_{E\setminus H\setminus \set{h}}) \cap \Omega(\mu^{\alpha, b}_{E\setminus H \setminus \set{h}})} \mu^{\alpha, a}_{E\setminus H \setminus \set{h}}(\beta) (f(\alpha, \beta, a, \cdot) - f(\alpha, \beta, b, \cdot))^2 \\
    &\leq \frac{1}{2} \sum_{\alpha \in \Omega(\mu_{H\setminus \set{e}})} \mu_{H\setminus\set{e}}(\alpha) \sum_{a,b \in \Omega(\mu^{\alpha}_e)} \mu^\alpha_{e}(a) (g(\alpha, a) - g(\alpha, b))^2 \\
    &\leq q \frac{1}{2} \sum_{\alpha \in \Omega(\mu_{H\setminus \set{e}})} \mu_{H\setminus\set{e}}(\alpha) \sum_{a,b \in \Omega(\mu^{\alpha}_e)} \mu^\alpha_{e}(a) \mu^\alpha_e(b) (g(\alpha, a) - g(\alpha, b))^2 
     = q \nu[\Var[e]{g}],
  \end{align*}
  where, in the last inequality, we use the fact that $\nu = \mu_H$ and $\mu^{\alpha}_{e}(b) \geq 1/q$.
\end{proof}

\section{Proof of \Cref{lem:basic-var-fact}}\label{app:proof}
  Let $G = (V, E)$ be a graph and $\mu$ be a distribution over $[q]^E$ satisfying~\eqref{eq:cond-ind}.
  Let $f: \Omega(\mu) \to \^R$ be a function.
  First, we study some exchange property for the conditional expectation.
  For any two edges $e$ and $e'$, define $\-{dist}(e,e')$ as the length of shortest path between $e$ and $e'$ in line graph.
  For the proofs in this section, we utilize the \emph{conditional independence} in~\eqref{eq:cond-ind}.

\begin{proposition} \label{prop:exchange-dist-ge-2}
  Let $S, T \subseteq E$ such that $\-{dist}(S, T) \geq 2$, then it holds that
  \begin{align*}
    \mu_{S \cup T}[f] = \mu_S\mu_T (f) = \mu_T\mu_S(f).
  \end{align*}
\end{proposition}
\begin{proof}
  For a fixed $\eta \in \Omega(\mu)$, let $\rho = \eta_{E\setminus (S\cup T)}$. Note that $\rho$ fixes the configuration both on $\partial S$ and $\partial T$. 
  By \eqref{eq:cond-ind}, conditional on $\rho$, configurations on $S$ and $T$ are independent.
  Then, by the definition of $\mu_S\mu_T[f]$ and above conditional independence property, one can verify
  \begin{align*}
    (\mu_S\mu_T(f))(\eta) =  \sum_{\sigma \in \Omega(\mu^\rho_S)} \mu^\rho_S(\sigma) \sum_{\tau \in \Omega(\mu^\rho_T)} \mu^\rho_T(\tau) f(\sigma \uplus \tau \uplus \rho).
  \end{align*}
Note that $\sigma,\tau,\rho$ are configurations over three disjoint sets (which forms a partition of $E$), we use $\sigma \uplus \tau \uplus \rho$ to denote the merged configuration on $E$. Since RHS is symmetric with respect to $S$ and $T$, we have  $(\mu_S\mu_T(f))(\eta) =  (\mu_T\mu_S(f))(\eta)$.
Also, by definition and the conditional independence,
\begin{align*}
 (\mu_{S \cup T} [f])(\eta) = \sum_{\sigma \in \Omega(\mu^\rho_S)}\sum_{\tau \in \Omega(\mu^\rho_T)}\mu^\rho_S(\sigma)\mu^\rho_T(\tau) f(\sigma \uplus \tau \uplus \rho).&\qedhere
\end{align*}

\end{proof}

\begin{proposition} \label{prop:exchange-contain}
  If $T \subseteq S \subseteq E$, then it holds that
  \begin{align*}
    \mu_S\mu_T(f) = \mu_T\mu_S(f) = \mu_S(f).
  \end{align*}
\end{proposition}
\begin{proof}
  By the law of total expectation, we have $\mu_S\mu_T(f) = \mu_S(f)$.
  Also, note that $E\setminus S \subseteq E\setminus T$, by definition, it is easy to verify $\mu_T\mu_S(f)  = \mu_S(f)$.
\end{proof}

\begin{proposition} 
  When $\partial S \cap T = \emptyset$, it holds that
  \begin{align} \label{eq:exchangable}
    \mu_S\mu_T(f) = \mu_T\mu_S(f).
  \end{align}
\end{proposition}

\begin{proof}
  We let $A = S\setminus T$ and let $B = T\setminus S$.
  Since $\partial S \cap T = \emptyset$, it holds that $A \subseteq S$ and $\-{dist}(S, B) \geq 2$, which implies $\-{dist}(A,B)\geq 2$.
  According to \Cref{prop:exchange-dist-ge-2} and \Cref{prop:exchange-contain}, we have
  \begin{align*}
    \mu_S\mu_T(f) &\overset{\text{Prop.\ref{prop:exchange-dist-ge-2}}}{=} \mu_S\mu_A\mu_B(f)
    \overset{\text{Prop.\ref{prop:exchange-contain}}}{=} \mu_S\mu_B(f)
    \overset{\text{Prop.\ref{prop:exchange-dist-ge-2}}}{=} \mu_B\mu_S(f)
    \overset{\text{Prop.\ref{prop:exchange-contain}}}{=} \mu_B\mu_A\mu_S(f)  \overset{\text{Prop.\ref{prop:exchange-dist-ge-2}}}{=} \mu_T\mu_S(f). \qedhere
  \end{align*}
\end{proof}

Finally, we are ready to prove \Cref{lem:basic-var-fact}.
Let $S, T \subseteq E$ such that $\partial S \cap T = \emptyset$.
We show that for any $f: \Omega(\mu) \to \^R$,
\begin{align*}
  \mu[\Var[S]{\mu_T[f]}] \leq \mu[\Var[S]{\mu_{S\cap T}[f]}].
\end{align*}
For two functions $g_1,g_2$, we use $(g_1-g_2)^2$ denote the function such that $(g_1-g_2)^2(\sigma) = (g_1(\sigma)-g_2(\sigma))^2$.
Note that for any function $g$, $\Var[S]{g} = \mu_S[(g - \mu_S[g])^2]$, where both LHS and RHS are functions and the equation here means equal two functions are the same.
Using all above propositions, we have
  \begin{align*}
    \mu[\Var[S]{\mu_T[f]}]
    &= \mu\left[\mu_S[(\mu_T(f) - \mu_S\mu_T(f))^2]\right] \\
    &= \mu\left[(\mu_T(f) - \mu_S\mu_T(f))^2\right] \\
    &\overset{\eqref{eq:exchangable}}{=} \mu\left[(\mu_T(f) - \mu_T\mu_S (f))^2\right] \\
    &= \mu\left[(\mu_T\mu_{S\cap T}(f) - \mu_T\mu_S \mu_{S\cap T}(f))^2\right] \\
    &\overset{(\star)}{\leq} \mu\left[(\mu_{S\cap T}(f) - \mu_S \mu_{S\cap T}(f))^2\right] \\
    &= \mu\left[\mu_S[(\mu_{S\cap T}(f) - \mu_S \mu_{S\cap T}(f))^2]\right] 
     = \mu[\Var[S]{\mu_{S\cap T}[f]}],
  \end{align*}
  where $(\star)$ is an application of Jensen's
  inequality.

\end{document}

%% file: Figures/Fig1.tex
  \def\NumberedBox#1{
    \node[fill=black, circle, inner sep=1pt, minimum size = 2pt] (tp) at (-2, 0) {\color{white} #1};
    \draw[-,thick] (-2.3, -0.3) -- (-2.3, 0.3) -- (-1.7, 0.3);
    \draw[-,thick] (2.6, -6.2) -- (3.2, -6.2) -- (3.2, -5.6);
  }
  \def\Marks{
    \node[label={[label distance = -7pt]90:{\small $v_0$}}] at (0) {};
    \node[label={[label distance = -5pt]120:{\small $v_1$}}] at (1) {};
    \node[label={[label distance = -3pt]180:{\small $v_2$}}] at (2) {};
    \node[label={[label distance = -5pt]60:{\small $v_3$}}] at (7) {};
    \node[label={[label distance = -3pt]180:{\small $v_4$}}] at (9) {};
  }
  \def\Paths{
    \begin{scope}[on background layer]
      \draw[fill=cBG, draw=black, rounded corners, dashed, thick] (2.north) -- (1.north) -- (4.north east) -- (19.east) -- (22.east) -- (23.east) -- (28.east) -- (28.south) -- (28.west) -- (23.west) -- (22.west) -- (19.west) -- (4.west) -- (1.south) -- (2.south) -- (2.west) -- cycle;
      \draw[fill=cBG, draw=black, rounded corners, dashed, thick] (9.east) -- (7.east) -- (7.north) -- (8.north) -- (8.west) -- (31.west) -- (31.south) -- (31.east) -- (8.east) -- (7.south) -- (9.west) -- (9.south) -- cycle;
    \end{scope}
  }
  \begin{tikzpicture}
    \begin{scope}
      \NumberedBox{1}
      \graph[
      tree layout,
      sibling distance = 2pt,
      nodes = {circle, inner sep = 0pt, minimum size = 3mm, as=},
      edges = {draw=white, double = black, line cap = round, double distance = 3pt}
      ]{
        0  -- [double = cG] 1,
        1  -- [double = cR] 2,  1  -- [double = cD] 3,  1  -- [double = cO] 4,
        2  -- [double = cB] 5,  2  -- [double = cY] 6,  2  -- [double = cG] 7,
        7  -- [double = cO] 8,  7  -- [double = cR] 9,  7  -- [double = cB] 10,
        9  -- [double = cY] 11, 9  -- [double = cD] 12, 9  -- [double = cB] 13,
        13 -- [double = cO] 14, 13 -- [double = cG] 15, 13 -- [double = cR] 16,
        4  -- [double = cR] 17, 4  -- [double = cB] 18, 4  -- [double = cY] 19,
        19 -- [double = cG] 20, 19 -- [double = cO] 21, 19 -- [double = cD] 22,
        22 -- [double = cB] 23, 22 -- [double = cR] 24, 22 -- [double = cG] 25,
        23 -- [double = cY] 26, 23 -- [double = cG] 27, 23 -- [double = cO] 28,
        8  -- [double = cY] 29, 8  -- [double = cG] 30, 8  -- [double = cD] 31,
        31 -- [double = cO] 32, 31 -- [double = cR] 33, 31 -- [double = cG] 34
      };
      \Paths
      \Marks
    \end{scope}
    \begin{scope}[shift = {(0, -7)}]
      \NumberedBox{2}
      \graph[
      tree layout,
      sibling distance = 2pt,
      nodes = {circle, inner sep = 0pt, minimum size = 3mm, as=},
      edges = {draw=white, double = black, line cap = round, double distance = 3pt}
      ]{
        0  -- [double = cG] 1,
        1  -- [double = cO] 2,  1  -- [double = cD] 3,  1  -- [double = cY] 4,
        2  -- [double = cB] 5,  2  -- [double = cY] 6,  2  -- [double = cG] 7,
        7  -- [double = cD] 8,  7  -- [double = cO] 9,  7  -- [double = cB] 10,
        9  -- [double = cY] 11, 9  -- [double = cD] 12, 9  -- [double = cB] 13,
        13 -- [double = cO] 14, 13 -- [double = cG] 15, 13 -- [double = cR] 16,
        4  -- [double = cR] 17, 4  -- [double = cB] 18, 4  -- [double = cD] 19,
        19 -- [double = cG] 20, 19 -- [double = cO] 21, 19 -- [double = cB] 22,
        22 -- [double = cO] 23, 22 -- [double = cR] 24, 22 -- [double = cG] 25,
        23 -- [double = cY] 26, 23 -- [double = cG] 27, 23 -- [double = cD] 28,
        8  -- [double = cY] 29, 8  -- [double = cG] 30, 8  -- [double = cB] 31,
        31 -- [double = cO] 32, 31 -- [double = cR] 33, 31 -- [double = cG] 34
      };
      \Paths
      \Marks
    \end{scope}
    \begin{scope}[shift = {(7, -7)}]
      \NumberedBox{3}
      \graph[
      tree layout,
      sibling distance = 2pt,
      nodes = {circle, inner sep = 0pt, minimum size = 3mm, as=},
      edges = {draw=white, double = black, line cap = round, double distance = 3pt}
      ]{
        0  -- [double = cR] 1,
        1  -- [double = cO] 2,  1  -- [double = cD] 3,  1  -- [double = cY] 4,
        2  -- [double = cB] 5,  2  -- [double = cY] 6,  2  -- [double = cR] 7,
        7  -- [double = cD] 8,  7  -- [double = cO] 9,  7  -- [double = cB] 10,
        9  -- [double = cY] 11, 9  -- [double = cD] 12, 9  -- [double = cB] 13,
        13 -- [double = cO] 14, 13 -- [double = cG] 15, 13 -- [double = cR] 16,
        4  -- [double = cR] 17, 4  -- [double = cB] 18, 4  -- [double = cD] 19,
        19 -- [double = cG] 20, 19 -- [double = cO] 21, 19 -- [double = cB] 22,
        22 -- [double = cO] 23, 22 -- [double = cR] 24, 22 -- [double = cG] 25,
        23 -- [double = cY] 26, 23 -- [double = cG] 27, 23 -- [double = cD] 28,
        8  -- [double = cY] 29, 8  -- [double = cG] 30, 8  -- [double = cB] 31,
        31 -- [double = cO] 32, 31 -- [double = cR] 33, 31 -- [double = cG] 34
      };
      \Paths
      \Marks
    \end{scope}
    \begin{scope}[shift={(7,0)}]
      \NumberedBox{4}
      \graph[
      tree layout,
      sibling distance = 2pt,
      nodes = {circle, inner sep = 0pt, minimum size = 3mm, as=},
      edges = {draw=white, double = black, line cap = round, double distance = 3pt}
      ]{
        0  -- [double = cR] 1,
        1  -- [double = cG] 2,  1  -- [double = cD] 3,  1  -- [double = cO] 4,
        2  -- [double = cB] 5,  2  -- [double = cY] 6,  2  -- [double = cR] 7,
        7  -- [double = cO] 8,  7  -- [double = cG] 9,  7  -- [double = cB] 10,
        9  -- [double = cY] 11, 9  -- [double = cD] 12, 9  -- [double = cB] 13,
        13 -- [double = cO] 14, 13 -- [double = cG] 15, 13 -- [double = cR] 16,
        4  -- [double = cR] 17, 4  -- [double = cB] 18, 4  -- [double = cY] 19,
        19 -- [double = cG] 20, 19 -- [double = cO] 21, 19 -- [double = cD] 22,
        22 -- [double = cB] 23, 22 -- [double = cR] 24, 22 -- [double = cG] 25,
        23 -- [double = cY] 26, 23 -- [double = cG] 27, 23 -- [double = cO] 28,
        8  -- [double = cY] 29, 8  -- [double = cG] 30, 8  -- [double = cD] 31,
        31 -- [double = cO] 32, 31 -- [double = cR] 33, 31 -- [double = cG] 34
      };
      \Paths
      \Marks
    \end{scope}
  \end{tikzpicture}

%% file: Figures/Fig2.tex
  \def\NumberedBox#1{
    \node[fill=black, circle, inner sep=1pt, minimum size = 2pt] (tp) at (-2, 0) {\color{white} #1};
    \draw[-,thick] (-2.3, -0.3) -- (-2.3, 0.3) -- (-1.7, 0.3);
    \draw[-,thick] (2.6, -6.2) -- (3.2, -6.2) -- (3.2, -5.6);
  }
  \def\Marks{
    \node[label={[label distance = -7pt]90:{\small $v_0$}}] at (0) {};
    \node[label={[label distance = -5pt]120:{\small $v_1$}}] at (1) {};
    \node[label={[label distance = -3pt]180:{\small $v_2$}}] at (2) {};
    \node[label={[label distance = -5pt]60:{\small $v_3$}}] at (7) {};
    \node[label={[label distance = -3pt]180:{\small $v_4$}}] at (9) {};
  }
  \def\Paths{
    \begin{scope}
      \draw[draw=black, rounded corners, dashed, thick] (2.north) -- (1.north) -- (4.north east) -- (19.east) -- (22.east) -- (23.east) -- (28.east) -- (28.south) -- (28.west) -- (23.west) -- (22.west) -- (19.west) -- (4.west) -- (1.south) -- (2.south) -- (2.west) -- cycle;
      \draw[draw=black, rounded corners, dashed, thick] (9.east) -- (7.east) -- (7.north) -- (8.north) -- (8.west) -- (31.west) -- (31.south) -- (31.east) -- (8.east) -- (7.south) -- (9.west) -- (9.south) -- cycle;
    \end{scope}
  }
  \def\AlterPath{
    \begin{scope}[on background layer]
      \draw[double = cBG, draw = white, line cap = round, double distance = 8pt, rounded corners] plot coordinates{(0) (1) (2) (7) (9)};
    \end{scope}
  }
  \def\AddEdgeNum(#1)--(#2)[#3,#4]#5;{
    \path (#1) -- (#2) node [midway, circle, fill=white,draw=#3, inner sep=0pt, minimum size=8pt, line width=1] {\scriptsize \color{black} \textsf{#5}};
  }
  \pgfdeclarelayer{background}
  \pgfdeclarelayer{foreground}
  \pgfsetlayers{background,main,foreground}
  \begin{tikzpicture}
    \begin{scope}[scale=0.7]
      \graph[
      tree layout,
      sibling distance = 2pt,
      nodes = {circle, inner sep = 0pt, minimum size = 3mm, as=},
      edges = {draw=white, double = black, line cap = round, double distance = 3pt}
      ]{
        0  -- [double = cG] 1,
        1  -- [double = cR] 2,  1  -- [double = cD] 3,  1  -- [double = cO] 4,
        2  -- [double = cB] 5,  2  -- [double = cY] 6,  2  -- [double = cG] 7,
        7  -- [double = cO] 8,  7  -- [double = cR] 9,  7  -- [double = cB] 10,
        9  -- [double = cY] 11, 9  -- [double = cD] 12, 9  -- [double = cB] 13,
        13 -- [double = cO] 14, 13 -- [double = cG] 15, 13 -- [double = cR] 16,
        4  -- [double = cR] 17, 4  -- [double = cB] 18, 4  -- [double = cY] 19,
        19 -- [double = cG] 20, 19 -- [double = cO] 21, 19 -- [double = cD] 22,
        22 -- [double = cB] 23, 22 -- [double = cR] 24, 22 -- [double = cG] 25,
        23 -- [double = cY] 26, 23 -- [double = cG] 27, 23 -- [double = cO] 28,
        8  -- [double = cY] 29, 8  -- [double = cG] 30, 8  -- [double = cD] 31,
        31 -- [double = cO] 32, 31 -- [double = cR] 33, 31 -- [double = cG] 34
      };
      \AlterPath
    \end{scope}

    \begin{scope}[scale=0.7, shift = {(6, 0)}]
      \graph[
      tree layout,
      sibling distance = 2pt,
      nodes = {circle, inner sep = 0pt, minimum size = 3mm, as=},
      edges = {draw=white, double = black, line cap = round, double distance = 3pt}
      ]{
        0  -- [double = cG] 1,
        1  -- [double = cR] 2,  1  -- [double = cD] 3,  1  -- [double = cO] 4,
        2  -- [double = cB] 5,  2  -- [double = cY] 6,  2  -- [double = cG] 7,
        7  -- [double = cO] 8,  7  -- [double = cR] 9,  7  -- [double = cB] 10,
        9  -- [double = cY] 11, 9  -- [double = cD] 12, 9  -- [double = cB] 13,
        13 -- [double = cO] 14, 13 -- [double = cG] 15, 13 -- [double = cR] 16,
        4  -- [double = cR] 17, 4  -- [double = cB] 18, 4  -- [double = cY] 19,
        19 -- [double = cG] 20, 19 -- [double = cO] 21, 19 -- [double = cD] 22,
        22 -- [double = cB] 23, 22 -- [double = cR] 24, 22 -- [double = cG] 25,
        23 -- [double = cY] 26, 23 -- [double = cG] 27, 23 -- [double = cO] 28,
        8  -- [double = cY] 29, 8  -- [double = cG] 30, 8  -- [double = cD] 31,
        31 -- [double = cO] 32, 31 -- [double = cR] 33, 31 -- [double = cG] 34
      };
      \AlterPath
      \Paths
    \end{scope}

    \begin{scope}[scale=0.7, shift = {(12, 0)}]
      \graph[
      tree layout,
      sibling distance = 2pt,
      nodes = {circle, inner sep = 0pt, minimum size = 3mm, as=},
      edges = {draw=white, double = black, line cap = round, double distance = 3pt}
      ]{
        0  -- [double = cG] 1,
        1  -- [double = cR] 2,  1  -- [double = cD] 3,  1  -- [double = cO] 4,
        2  -- [double = cB] 5,  2  -- [double = cY] 6,  2  -- [double = cG] 7,
        7  -- [double = cO] 8,  7  -- [double = cR] 9,  7  -- [double = cB] 10,
        9  -- [double = cY] 11, 9  -- [double = cD] 12, 9  -- [double = cB] 13,
        13 -- [double = cO] 14, 13 -- [double = cG] 15, 13 -- [double = cR] 16,
        4  -- [double = cR] 17, 4  -- [double = cB] 18, 4  -- [double = cY] 19,
        19 -- [double = cG] 20, 19 -- [double = cO] 21, 19 -- [double = cD] 22,
        22 -- [double = cB] 23, 22 -- [double = cR] 24, 22 -- [double = cG] 25,
        23 -- [double = cY] 26, 23 -- [double = cG] 27, 23 -- [double = cD] 28,
        8  -- [double = cY] 29, 8  -- [double = cG] 30, 8  -- [double = cD] 31,
        31 -- [double = cO] 32, 31 -- [double = cR] 33, 31 -- [double = cG] 34
      };
      \AddEdgeNum(23)--(28)[cD,white]{1};
      \AlterPath
    \end{scope}

    \begin{scope}[scale=0.7, shift = {(18, 0)}]
      \graph[
      tree layout,
      sibling distance = 2pt,
      nodes = {circle, inner sep = 0pt, minimum size = 3mm, as=},
      edges = {draw=white, double = black, line cap = round, double distance = 3pt}
      ]{
        0  -- [double = cG] 1,
        1  -- [double = cR] 2,  1  -- [double = cD] 3,  1  -- [double = cO] 4,
        2  -- [double = cB] 5,  2  -- [double = cY] 6,  2  -- [double = cG] 7,
        7  -- [double = cO] 8,  7  -- [double = cR] 9,  7  -- [double = cB] 10,
        9  -- [double = cY] 11, 9  -- [double = cD] 12, 9  -- [double = cB] 13,
        13 -- [double = cO] 14, 13 -- [double = cG] 15, 13 -- [double = cR] 16,
        4  -- [double = cR] 17, 4  -- [double = cB] 18, 4  -- [double = cY] 19,
        19 -- [double = cG] 20, 19 -- [double = cO] 21, 19 -- [double = cD] 22,
        22 -- [double = cO] 23, 22 -- [double = cR] 24, 22 -- [double = cG] 25,
        23 -- [double = cY] 26, 23 -- [double = cG] 27, 23 -- [double = cD] 28,
        8  -- [double = cY] 29, 8  -- [double = cG] 30, 8  -- [double = cB] 31,
        31 -- [double = cO] 32, 31 -- [double = cR] 33, 31 -- [double = cG] 34
      };
      \AddEdgeNum(23)--(28)[cD,white]{1};
      \AddEdgeNum(8)--(31)[cB,white]{2};
      \AddEdgeNum(22)--(23)[cO,white]{3};
      \AlterPath
    \end{scope}

    \begin{scope}[scale=0.7, shift = {(0, -7)}]
      \graph[
      tree layout,
      sibling distance = 2pt,
      nodes = {circle, inner sep = 0pt, minimum size = 3mm, as=},
      edges = {draw=white, double = black, line cap = round, double distance = 3pt}
      ]{
        0  -- [double = cG] 1,
        1  -- [double = cR] 2,  1  -- [double = cD] 3,  1  -- [double = cO] 4,
        2  -- [double = cB] 5,  2  -- [double = cY] 6,  2  -- [double = cG] 7,
        7  -- [double = cD] 8,  7  -- [double = cO] 9,  7  -- [double = cB] 10,
        9  -- [double = cY] 11, 9  -- [double = cD] 12, 9  -- [double = cB] 13,
        13 -- [double = cO] 14, 13 -- [double = cG] 15, 13 -- [double = cR] 16,
        4  -- [double = cR] 17, 4  -- [double = cB] 18, 4  -- [double = cY] 19,
        19 -- [double = cG] 20, 19 -- [double = cO] 21, 19 -- [double = cB] 22,
        22 -- [double = cO] 23, 22 -- [double = cR] 24, 22 -- [double = cG] 25,
        23 -- [double = cY] 26, 23 -- [double = cG] 27, 23 -- [double = cD] 28,
        8  -- [double = cY] 29, 8  -- [double = cG] 30, 8  -- [double = cB] 31,
        31 -- [double = cO] 32, 31 -- [double = cR] 33, 31 -- [double = cG] 34
      };
      \AddEdgeNum(23)--(28)[cD,white]{1};
      \AddEdgeNum(8)--(31)[cB,white]{2};
      \AddEdgeNum(22)--(23)[cO,white]{3};
      \AddEdgeNum(7)--(8)[cD,white]{4};
      \AddEdgeNum(19)--(22)[cB,white]{5};
      \AddEdgeNum(7)--(9)[cO,white]{6};
      \AlterPath
    \end{scope}

    \begin{scope}[scale=0.7, shift = {(6, -7)}]
      \graph[
      tree layout,
      sibling distance = 2pt,
      nodes = {circle, inner sep = 0pt, minimum size = 3mm, as=},
      edges = {draw=white, double = black, line cap = round, double distance = 3pt}
      ]{
        0  -- [double = cG] 1,
        1  -- [double = cR] 2,  1  -- [double = cD] 3,  1  -- [double = cO] 4,
        2  -- [double = cB] 5,  2  -- [double = cY] 6,  2  -- [double = cG] 7,
        7  -- [double = cD] 8,  7  -- [double = cO] 9,  7  -- [double = cB] 10,
        9  -- [double = cY] 11, 9  -- [double = cD] 12, 9  -- [double = cB] 13,
        13 -- [double = cO] 14, 13 -- [double = cG] 15, 13 -- [double = cR] 16,
        4  -- [double = cR] 17, 4  -- [double = cB] 18, 4  -- [double = cD] 19,
        19 -- [double = cG] 20, 19 -- [double = cO] 21, 19 -- [double = cB] 22,
        22 -- [double = cO] 23, 22 -- [double = cR] 24, 22 -- [double = cG] 25,
        23 -- [double = cY] 26, 23 -- [double = cG] 27, 23 -- [double = cD] 28,
        8  -- [double = cY] 29, 8  -- [double = cG] 30, 8  -- [double = cB] 31,
        31 -- [double = cO] 32, 31 -- [double = cR] 33, 31 -- [double = cG] 34
      };
      \AddEdgeNum(23)--(28)[cD,white]{1};
      \AddEdgeNum(8)--(31)[cB,white]{2};
      \AddEdgeNum(22)--(23)[cO,white]{3};
      \AddEdgeNum(7)--(8)[cD,white]{4};
      \AddEdgeNum(19)--(22)[cB,white]{5};
      \AddEdgeNum(7)--(9)[cO,white]{6};
      \AddEdgeNum(4)--(19)[cD,white]{7};
      \AlterPath
    \end{scope}

    \begin{scope}[scale=0.7, shift = {(12, -7)}]
      \graph[
      tree layout,
      sibling distance = 2pt,
      nodes = {circle, inner sep = 0pt, minimum size = 3mm, as=},
      edges = {draw=white, double = black, line cap = round, double distance = 3pt}
      ]{
        0  -- [double = cG] 1,
        1  -- [double = cO] 2,  1  -- [double = cD] 3,  1  -- [double = cY] 4,
        2  -- [double = cB] 5,  2  -- [double = cY] 6,  2  -- [double = cG] 7,
        7  -- [double = cD] 8,  7  -- [double = cO] 9,  7  -- [double = cB] 10,
        9  -- [double = cY] 11, 9  -- [double = cD] 12, 9  -- [double = cB] 13,
        13 -- [double = cO] 14, 13 -- [double = cG] 15, 13 -- [double = cR] 16,
        4  -- [double = cR] 17, 4  -- [double = cB] 18, 4  -- [double = cD] 19,
        19 -- [double = cG] 20, 19 -- [double = cO] 21, 19 -- [double = cB] 22,
        22 -- [double = cO] 23, 22 -- [double = cR] 24, 22 -- [double = cG] 25,
        23 -- [double = cY] 26, 23 -- [double = cG] 27, 23 -- [double = cD] 28,
        8  -- [double = cY] 29, 8  -- [double = cG] 30, 8  -- [double = cB] 31,
        31 -- [double = cO] 32, 31 -- [double = cR] 33, 31 -- [double = cG] 34
      };
      \AddEdgeNum(23)--(28)[cD,white]{1};
      \AddEdgeNum(8)--(31)[cB,white]{2};
      \AddEdgeNum(22)--(23)[cO,white]{3};
      \AddEdgeNum(7)--(8)[cD,white]{4};
      \AddEdgeNum(19)--(22)[cB,white]{5};
      \AddEdgeNum(7)--(9)[cO,white]{6};
      \AddEdgeNum(4)--(19)[cD,white]{7};
      \AddEdgeNum(1)--(4)[cY,black]{8};
      \AddEdgeNum(1)--(2)[cO,white]{9};
      \AlterPath
    \end{scope}

    \begin{scope}[scale=0.7, shift = {(18, -7)}]
      \graph[
      tree layout,
      sibling distance = 2pt,
      nodes = {circle, inner sep = 0pt, minimum size = 3mm, as=},
      edges = {draw=white, double = black, line cap = round, double distance = 3pt}
      ]{
        0  -- [double = cG] 1,
        1  -- [double = cO] 2,  1  -- [double = cD] 3,  1  -- [double = cY] 4,
        2  -- [double = cB] 5,  2  -- [double = cY] 6,  2  -- [double = cG] 7,
        7  -- [double = cD] 8,  7  -- [double = cO] 9,  7  -- [double = cB] 10,
        9  -- [double = cY] 11, 9  -- [double = cD] 12, 9  -- [double = cB] 13,
        13 -- [double = cO] 14, 13 -- [double = cG] 15, 13 -- [double = cR] 16,
        4  -- [double = cR] 17, 4  -- [double = cB] 18, 4  -- [double = cD] 19,
        19 -- [double = cG] 20, 19 -- [double = cO] 21, 19 -- [double = cB] 22,
        22 -- [double = cO] 23, 22 -- [double = cR] 24, 22 -- [double = cG] 25,
        23 -- [double = cY] 26, 23 -- [double = cG] 27, 23 -- [double = cD] 28,
        8  -- [double = cY] 29, 8  -- [double = cG] 30, 8  -- [double = cB] 31,
        31 -- [double = cO] 32, 31 -- [double = cR] 33, 31 -- [double = cG] 34
      };
    \end{scope}
  \end{tikzpicture}

%% file: Figures/Fig3.tex
  \begin{tikzpicture}
      \begin{scope}
      \graph[
      tree layout,
      nodes = {circle, inner sep = 0pt, minimum size = 3mm, as=},
      edges = {draw=white, double = black, line cap = round, double distance = 3pt}
      ]{
        0  -- [double = cG] 1,
        1  -- [double = cR] 2,  1  -- [double = cB] 3,  
        2  -- [double = cB] 4,  2  -- [double = cD] 5
      };
      \path (1) -- (2) node [midway, left] {$e_1$};
      \begin{scope}[on background layer]
      \draw[double = cBG, draw = white, line cap = round, double distance = 8pt, rounded corners] plot coordinates{(0) (1) (2)};
      \end{scope}
      \end{scope}
      \begin{scope}[shift={(7,0)}]
      \graph[
      tree layout,
      nodes = {circle, inner sep = 0pt, minimum size = 3mm, as=},
      edges = {draw=white, double = black, line cap = round, double distance = 3pt}
      ]{
        0  -- [double = cG] 1,
        1  -- [double = cR] 2,  1  -- [double = cD] 3,  
        2  -- [double = cB] 4,  2  -- [double = cG] 5,
        5  -- [double = cR] 6,  5  -- [double = cD] 7,
        6  -- [double = cB] 8,  6  -- [double = cD] 9,
        4  -- [double = cG] 10, 4  -- [double = cR] 11,
        3  -- [double = cG] 12, 3  -- [double = cR] 13
      };
      \path (1) -- (2) node [midway, above] {$e_1$};
      \path (5) -- (6) node [midway, left, label={[label distance=-3pt]180:$e_3$}] {};
      \begin{scope}[on background layer]
      \draw[double = cBG, draw = white, line cap = round, double distance = 8pt, rounded corners] plot coordinates{(2) (1) (3)};
      \draw[draw=black, rounded corners, dashed, thick] (6.south) -- (6.east) -- (5.east) -- (2.east) -- (2.north) -- (2.west) -- (4.west) -- (4.south) -- (4.east) -- (2.south) -- (5.west) -- (6.west) -- cycle;
      \end{scope}
      \end{scope}
  \end{tikzpicture}